\documentclass[letterpaper,12pt]{amsart}

\usepackage{color}
\definecolor{darkblue}{rgb}{0.,0.,0.4}
\definecolor{darkred}{rgb}{0.5,0.,0.}
\usepackage[pdftex,colorlinks=true,linkcolor=darkblue,citecolor=darkred,urlcolor=blue]{hyperref}

\usepackage{amsmath}
\usepackage{amssymb}
\usepackage{amsfonts}
\usepackage{amsthm}
\usepackage[all]{xy}

\newtheorem{theorem}{Theorem}
\newtheorem{lem}{Lemma}[section]
\newtheorem{cor}[lem]{Corollary}
\newtheorem{prop}[lem]{Proposition}
\theoremstyle{definition} \newtheorem{defn}{Definition}
\theoremstyle{definition} \newtheorem{rem}{Remark}
\theoremstyle{definition} \newtheorem{example}{Example}

\newcommand{\ket}[1]{\left| {#1} \right\rangle}
\newcommand{\tr}{\mathop{\mathrm{tr}}}
\newcommand{\im}{\mathop{\mathrm{im}}}
\newcommand{\coker}{\mathop{\mathrm{coker}}}
\newcommand{\ann}{\mathop{\mathrm{ann}}}
\newcommand{\rank}{\mathop{\mathrm{rank}}}
\newcommand{\codim}{\mathop{\mathrm{codim}}}
\newcommand{\Tor}{{\mathrm{Tor}}}
\newcommand{\rad}{\mathop{\mathrm{rad}}}

\newcommand{\FF}{\mathbb{F}}
\newcommand{\ZZ}{\mathbb{Z}}
\newcommand{\mm}{\mathfrak{m}}
\newcommand{\pp}{\mathfrak{p}}
\newcommand{\bb}{\mathfrak{b}}
\newcommand{\id}{\mathrm{id}}

\begin{document}
\title{Commuting Pauli Hamiltonians as maps between free modules}
\author{Jeongwan Haah}
\address{Institute for Quantum Information and Matter, California Institute of Technology, Pasadena, California}
\email{jwhaah@caltech.edu}
\date{8 March 2013}
\begin{abstract}
We study unfrustrated spin Hamiltonians 
that consist of commuting tensor products of Pauli matrices.
Assuming translation-invariance,
a family of Hamiltonians that belong to the same phase of matter
is described by a map between modules over the translation-group algebra,
so homological methods are applicable.
In any dimension every point-like charge
appears as a vertex of a fractal operator, and
can be isolated with energy barrier at most logarithmic in the separation distance.
For a topologically ordered system in three dimensions,
there must exist a point-like nontrivial charge.
A connection between the ground state degeneracy and the number of points on
an algebraic set is discussed.
Tools to handle local Clifford unitary transformations are given.
\end{abstract}
\maketitle
{\small \tableofcontents}

Commuting Pauli Hamiltonians form a small class of Hamiltonians
that are consisted of products of Pauli matrices such that
each term commutes with any other terms.
Classical examples are the Ising models in one or two dimensions.
Albeit its simplicity of the energy spectrum,
there are many intriguing models in this class
for which the long range entanglement of the ground state plays a very important role.
Prototypical is the Kitaev's toric code model~\cite{Kitaev2003Fault-tolerant},
which has been a solid testbed of ideas for topologically ordered systems.

The topological ordered models exhibit, as the name suggests, many properties
that are insensitive local changes or defects.
They had been discussed for the states of
the fractional quantum Hall effects and the spin liquids;
see e.g. Wen~\cite{Wen1991SpinLiquid}.
Perhaps, the most well-defined characteristic of the topological order
is the local indistinguishability of the degenerate ground states;
two different ground states gives the same expectation 
value for any local observables.
(Note that this characteristic is not directly applicable to
e.g. the topological insulators~\cite{HasanKane2010TIReview},
for which certain symmetry properties distinguish them from trivial phases.)
Due to the local indistinguishability,
the topologically ordered systems are thought
to be candidate media
on which quantum information processing is performed.
As a special application, the topologically ordered system can be used as
a quantum memory, just like the ferromagnetic system is used as a classical memory.

However, the quantum memories in the topologically ordered systems
often suffer from thermal instability.
For example, the toric code model has point-like excitations,
which can freely propagate by external noise from the thermal bath.
Although a local operator can never access to the ground space,
their accumulation may be able to.
Indeed, by the thermal fluctuation, a ground state
is often mapped to a different state,
and the anticipated protection of the stored quantum information is not viable.
The excitations that affects the stability of the quantum memory
may be called ``topological charges''.
A charge is an excitation that cannot be created alone locally
but can be created with some other excitations.
Indeed, the 4D toric code~\cite{DennisKitaevLandahlEtAl2002Topological}
has no charge at all,
and can be used as a quantum memory whose failure probability
decreases exponentially with the system size at low enough 
temperatures~\cite{AlickiHorodeckiHorodeckiEtAl2010thermal,
ChesiLossBravyiEtAl2010Thermodynamic}.

The situation in three dimensions is more subtle but interesting.
Models like 3D toric code model have charges
that can freely propagate across the system by the interaction with the thermal bath,
thereby two different ground states become mixed.
On the other hand, as in the cubic code~\cite{Haah2011Local,BravyiHaah2011Energy,BravyiHaah2011Memory},
there can exist charges that cannot propagate by any means.
This class of models provides modest reliability as a quantum memory at nonzero temperature.
However, the scaling of the memory time, until which the system is reliable as a memory medium,
is not as favorable as it is for the 4D toric code model;
the memory time grows with the system size according to a power law
whose exponent is proportional to the inverse temperature,
provided that the system size does not exceed some critical value
determined by the temperature.

The very existence of the charges seems to adversely affect the memory time.
In order to have a quantum memory,
one needs to devise a read-out procedure explicitly ---
a classical analog is the measurement of the average magnetization of 2D Ising model.
Though the charges may not propagate,
they can be separated arbitrary far from their partner charges
at a modest energy cost~\cite{BravyiHaah2011Energy}.
No good read-out procedure (sometimes called decoder)
is known for the configurations of far separated charges.
This should be contrasted to the 4D toric code model,
in which any excited state consists of several loops.
The large loops are suppressed by the Boltzmann factor,
and the small loops can be almost perfectly treated by the read-out procedure.
In short, the large entropy due to the point-likeness of the charges
would likely drag the ground state to a hard-to-decipher state.
See the discussion in~\cite{BravyiHaah2011Memory}.

Apart from the issue of the thermal stability and the possibility of quantum memory,
the cubic code model apparently necessitates new tools to analyze it.
When defined on a finite system with periodic boundary conditions,
it shows exotic dependence of the ground state degeneracy on the system size.
The degeneracy is sensitive to the number theoretic property of
the linear system size $L$. For example, when $L$ is a power of 2,
the degeneracy grows exponentially with $L$, but becomes a constant if $L = 2^p +1$.
This was a numerical observation, and was not rigorously treated~\cite{Haah2011Local}.

\begin{center}
{\bf Results}
\end{center}

In this paper, we systematically study commuting Pauli Hamiltonians
that are translation-invariant.
We always assume that our Hamiltonians are frustration-free;
every term in the Hamiltonian is minimized on the ground space.

The main observation is that there is a purely algebraic description of
commuting Pauli Hamiltonians in terms of maps between 
free modules over a Laurent polynomial ring by exploiting the translation-invariance.
A Pauli matrix can be written as two binary numbers if we ignore the phase factors.
For example, $I = (00), \sigma_x=(10), \sigma_z = (01), \sigma_y=(11)$.
We write these binary numbers in the coefficients of Laurent polynomials.
The exponents of the Laurent polynomials will represent the positions at which the Pauli matrices act.
If the Hamiltonian is translation-invariant, and there are finitely many distinct interaction types,
then it follows that only a finite number of the Laurent polynomials convey all data of the Hamiltonian.
We view this finite data as a map between two free modules over the translation-group algebra.
We will show that the physical phase is solely determined by the image of this associated map.

We provide a few tools to compute the transformations of the Hamiltonians 
by local unitary operators and coarse-graining when a translation structure is given.
They come down to a well-defined set of elementary row operations
on the matrices associated to the Hamiltonians~\cite{CalderbankRainsShorEtAl1997Quantum,KitaevShenVyalyi2002CQC}.
As we restrict our scope to the commuting Pauli Hamiltonians,
the local unitary operators are also restricted to the Clifford operators.
(Clifford operators maps a tensor product of Pauli operators to 
a tensor product of Pauli operators.)

We define the characteristic dimension $d$ associated to the Hamiltonian.
If a Hamiltonian gives rise to a map between free modules,
it is natural to think of the determinantal ideal of this map.
The characteristic dimension is the Krull dimension of
the algebraic set defined by this ideal.
It is always upper bounded by the spatial dimension $D$.
Moreover, $d$ is less than or equal to $D-2$ if the Hamiltonian is locally topologically ordered.

The characteristic dimension $d$ 
controls the rate at which the ground state degeneracy may increase.
Roughly speaking, the logarithm of degeneracy can grow like $L^d$ 
where $L$ is the linear system size.
Thus, $D=3$ is the minimal spatial dimension such that the degeneracy
of a topologically ordered system can be diverging.
For instance, the toric code models in various dimensions
all correspond to the characteristic dimension 0,
while the 3D cubic code model has characteristic dimension 1.
However, it should be pointed out that
the actual degeneracy does not behave as smooth as the function $L^d$;
it can depend very sensitively on the system size.
Indeed, it shall be shown that the degeneracy is related to the number of points in an algebraic set.
The boundary condition imposes a global constraint on the relevancy of the points in the algebraic set.
The numerically observed phenomena for the cubic code model shall be exactly calculated.

We characterize topological charges in terms of torsion submodule.
In particular, this characterization implies that
a charge always appears as a vertex of a fractal operator,
and can be separated arbitrarily far from its partners
by a local process with energy cost \emph{at most logarithmic} in the separation distance.
Here, the local process means a sequence of Pauli operators
that are obtained by successive applications of single qubit Pauli operators.
The fractal operators may be regarded as finite cellular automata~\cite{MartinOdlyzkoAndrewWolfram1984}.

Using our formalism we derive a number of consequences of the low physical dimensionality.
In one dimension, we algorithmically show how to transform
an arbitrary Hamiltonian into several copies of the Ising models~\cite{Yoshida2011Classification}.
In two dimensions, we characterize how the charges behave for topologically ordered models.
Specifically, we prove that any excited state is
a configuration of finitely many kinds of the charges,
and the charges can be moved to an arbitrary position by a string operator.
The result is a refined formulation of \cite{
Yoshida2011Classification,
Bombin2011Structure,
BombinDuclosCianciPoulin2012}.
In three dimensions, we prove that there always exists a point-like charge
for any locally topologically ordered translation-invariant commuting Pauli Hamiltonian.
This is a fundamental property of the three dimensions.
It suggests that we might not be able to have a topologically ordered
system in three dimensions where the excitations are all loop-like as in 4D toric code model.
If it is further assumed that the ground state degeneracy is constant
independent of system size, then we prove that the charges are attached to strings~\cite{Yoshida2011feasibility}.

The use of Laurent polynomials is not completely new.
For classical cyclic linear codes, the use of polynomial is routine~\cite{MacWilliamsSloane1977}.
Multivariate polynomials appear in the error correcting code theory
in the topic of multi-dimensional cyclic codes; see e.g.~\cite{GueneriOezbudak2008}.
Also, there is an algebraic-geometry based design like Goppa codes~\cite{Goppa1983}.
However, the focus is different.
We are interested in a fixed set of generators,
while in the classical error correcting codes one is interested in good codes of a fixed length.
Exact sequences of modules describing the topological order is only relevant
if we are interested in infinite systems, or infinite code length.
The question of minimum code distance is not addressed
since we take a Hamiltonian viewpoint for the codes.
The fact that the commuting Pauli operators are represented as matrix,
is very well-known in the theory of quantum error correcting 
codes~\cite{CalderbankShor1996Good,Steane1996Multiple,Gottesman1996Saturating,CalderbankRainsShorEtAl1997Quantum}.
Our treatment is different in that the entries are not the binary values but the Laurent polynomials.

Only the system of qubits, or spin-$1/2$, will be discussed,
but all of our results and argument straight-forwardly
generalize to the system of qudits of prime dimensions.
Technically, the ground field $\FF_2$ for the qubit should be
replaced by $\FF_p$ for a prime integer $p$.
It is important that the ground field is finite.
Some numerical value $2$ should be replace by the characteristic $p$ of the field.
With this generalization in mind, we keep necessary minus ($-$) signs in the statements,
which should be ignored for qu\emph{b}its.
Examples in this direction can be found in~\cite{Kim2012qupit}.

We start by deriving the matrix representation of a commuting Pauli Hamiltonian,
and explain in detail how the translation-invariance is exploited.
The notion of modules over the translation-group algebra shall naturally emerge.
Some operations on modules are induced by local unitary transformations on the physical system.
They will define an equivalence relation between the matrix representations, and hence between Hamiltonians.
We move on to the topological order and translate the conditions
into those on a chain complex of modules.
The characterization of charges will be given
by exploiting the positive characteristic of the ground field.
Finally, consequences of the topological order condition
in two and three spatial dimensions will be derived.
Explicit calculations and more examples are presented in the last section.
All ring in the present paper shall be commutative with 1.

\emph{Acknowledgements:}
The author would like to thank
Sergey Bravyi, 
Lawrence Chung, 
Alexei Kitaev, 
John Preskill, 
Eric Rains, 
and Ari Turner
for useful discussions.
The author thanks Tom Graber for giving an intuitive explanation
for Proposition~\ref{prop:support-torsion-fitting-ideal}.
The author is supported in part by the
Institute for Quantum Information and Matter, an NSF Physics Frontier Center,
and the Korea Foundation for Advanced Studies.

\begin{table}[t]
\begin{tabular}{c|l}
\hline
$\FF_2$ & binary field $\{0,1\}$\\
$D$ & spatial dimension \\
$R$ & $\FF_2[x_1, x_1^{-1}, \ldots, x_D, x_D^{-1}]$ \\
$\bb_L$ & ideal $(x_1^L-1,\ldots,x_D^L -1)$\\
$q$ & number of qubits per site \\
$t$ & number of interaction types \\
$G$ & free $R$-module of the interaction labels (rank $t$) \\
$P$ & free $R$-module of Pauli operators (rank $2q$)\\
$E$ & free $R$-module of excitations (rank $t$)\\
$\sigma$ & $G \to P$, generating matrix or map for the stabilizer module \\
$\epsilon$ & $P \to E$, generating matrix or map for excitations \\
$r \mapsto \bar r$ & antipode map of the group algebra $R$.\\
$\dagger$ & transpose followed by antipode map \\
$\lambda_q$ & anti-symmetric $2q \times 2q$ matrix $\begin{pmatrix} 0 & \id \\ -\id & 0 \end{pmatrix}$ \\
\hline
\end{tabular}
\caption{Reserved symbols. Any ring in this paper is commutative with 1.}
\end{table}

\section{Algebraic structure of commuting Pauli Hamiltonians}

\subsection{Pauli group as a vector space}

The Pauli matrices
\[
 \sigma_x = \begin{pmatrix} 0 & 1 \\ 1 & 0 \end{pmatrix}, \quad
 \sigma_y = \begin{pmatrix} 0 & -i \\ i & 0 \end{pmatrix}, \quad
 \sigma_z = \begin{pmatrix} 1 & 0 \\ 0 & -1 \end{pmatrix}
\]
satisfy
\[
 \sigma_a \sigma_b = i \varepsilon_{abc} \sigma_c, \quad
 \{ \sigma_a , \sigma_b \} = 2 \delta_{ab}.
\]
Thus, the Pauli matrices together with scalars $\pm 1, \pm i$
form a group under multiplication.
Given a system of qubits,
the set of all possible tensor products of the Pauli matrices form a group,
where the group operation is the multiplication of operators.
If the system is infinite, physically meaningful operators
are those of finite support, i.e.,
acting on all but finitely many qubits by the identity.
We shall only consider this Pauli group of finite support,
and call it simply the {\em Pauli group}.
An element of the Pauli group is called a {\em Pauli operator}.

Since any two elements of the Pauli group either commute or anti-commute,
ignoring the phase factor altogether, one obtains an {\em abelian} group.
Moreover, since any element $O$ of the Pauli group satisfies $O^2 = \pm I$,
An action of $\mathbb{Z}/2\mathbb{Z}$ on Pauli group modulo phase factors
$P / \{\pm 1, \pm i\}$ is well-defined,
by the rule $ n \cdot O = O^n$ where $n \in \mathbb{Z}/2\mathbb{Z}$.
For $\FF_2 = \mathbb{Z}/2\mathbb{Z}$ being a field,
$P / \{\pm 1, \pm i\}$ becomes a vector space over $\FF_2$.
The group of single qubit Pauli operators up to phase factors
is identified with the two dimensional $\FF_2$-vector space.
If $\Lambda$ is the index set of all qubits in the system,
the whole Pauli group up to phase factors is the direct sum $\bigoplus_{i \in \Lambda} V_i$
where $V_i$ is the vector space of the Pauli operators for the qubit at $i$.
Explicitly, $I = (00), \sigma_x = (10), \sigma_z = (01), \sigma_y = (11)$.
A multi-qubit Pauli operator
is written as a finite product of the single qubit Pauli operators,
and hence is written as a binary string in which all but finitely many entries are zero.
A pair of entries of the binary string describes a single qubit component in the tensor product expression.
The multiplication of two Pauli operators corresponds to entry-wise addition of the two binary strings modulo 2.

The commutation relation may seem at first lost,
but one can recover it by introducing a symplectic form~\cite{CalderbankRainsShorEtAl1997Quantum}.
Let $\lambda_1 = \begin{pmatrix} 0 & 1 \\ -1 & 0 \end{pmatrix}$
be a symplectic form on the vector space $(\FF_2)^2$ of a single qubit Pauli operators.
\footnote{The minus sign is not necessary for qubits, but is for qudits of prime dimensions}
One can easily check that the commutation relation of two Pauli matrices $O_1, O_2$
is precisely the value of this symplectic form
evaluated on the pair of vectors representing $O_1$ and $O_2$.
Two multi-qubit Pauli operator (anti-)commutes if and only if
there are (odd)even number of pairs of the anti-commuting single qubit Pauli operators
in their tensor product expression.
So, the two Pauli operator (anti-)commutes
precisely when the value of the direct sum of symplectic form $\bigoplus_{q \in \Lambda} \lambda_1$
is (non-)zero. ($\Lambda$ could be infinite but the form is well-defined since any vector representing a Pauli operator is of finite support.)
We shall call the value of the symplectic form the {\em commutation value}.

\subsection{Pauli space on a group}
\label{sec:Pauli-space-on-group}

Let $\Lambda$ be the index set of all qubits,
and suppose now that $\Lambda$ itself is an abelian group.
There is a natural action of $\Lambda$ on the Pauli group modulo phase factors
induced from the group action of $\Lambda$ on itself by multiplication.
For example, if $\Lambda = \mathbb{Z}$,
the action of $\Lambda$ is the translation on the one dimensional chain of qubits.
If $R=\FF_2[\Lambda]$ is the group algebra with multiplicative identity denoted by $1$,
the Pauli group modulo phase factors acquires a structure of an $R$-module.
We shall call it the {\em Pauli module}.
The Pauli module is free and has rank 2.

Let $r \mapsto \bar r$ be the antipode map of $R$, 
i.e., the $\FF_2$-linear map into itself 
such that each group element is mapped to its inverse.
Since $\Lambda$ is abelian, the antipode map is an algebra-automorphism.
Let the coefficient of $a \in R$ at $g \in \Lambda$ be denoted by $a_g$.
Hence, $a = \sum_{g \in \Lambda} a_g g$ for any $a \in R$.
One may write $a_g = (a \bar g)_1$.

Define
\[
 \tr(a) = a_1
\]
for any $a \in R$.
\begin{prop}
\cite{CalderbankRainsShorEtAl1997Quantum}
Let $(a,b), (c,d) \in R^2$ be two vectors representing Pauli operators $O_1, O_2$ up to phase factors:
\begin{align*}
 O_1 &= \left( \bigotimes_{g \in \Lambda} (\sigma_x^{(g)})^{a_g} \right) 
       \left( \bigotimes_{g \in \Lambda} (\sigma_z^{(g)})^{b_g} \right) ,
\\
 O_2 &= \left( \bigotimes_{g \in \Lambda} (\sigma_x^{(g)})^{c_g} \right) 
       \left( \bigotimes_{g \in \Lambda} (\sigma_z^{(g)})^{d_g} \right) 
\end{align*}
where $\sigma^{(g)}$ denotes the single qubit Pauli operator at $g \in \Lambda$.
Then, $O_1$ and $O_2$ commute if and only if  
\[
 \tr \left(
\begin{pmatrix} \bar a & \bar b \end{pmatrix}
\begin{pmatrix} 0 & 1 \\ -1 & 0 \end{pmatrix}
\begin{pmatrix} c \\ d \end{pmatrix}
\right) = 0 .
\]
\label{prop:trace-formula-commutation-value}
\end{prop}
\begin{proof}
The commutation value of $(\sigma_x^{(g)})^{n} (\sigma_z^{(g)})^{m}$
and $(\sigma_x^{(g)})^{n'} (\sigma_z^{(g)})^{m'}$
is $nm' - mn' \in \FF_2$.
Viewed as pairs of group algebra elements,
$(\sigma_x^{(g)})^{n} (\sigma_z^{(g)})^{m}$
and $(\sigma_x^{(g)})^{n'} (\sigma_z^{(g)})^{m'}$
are $(n g, m g)$ and $(n' g, m' g)$, respectively.
We see that
\[
 nm' - mn' =  \tr \left(
\begin{pmatrix} n g^{-1} &  m g^{-1} \end{pmatrix}
\begin{pmatrix} 0 & 1 \\ -1 & 0 \end{pmatrix}
\begin{pmatrix} n' g \\ m' g \end{pmatrix}
\right) .
\]
Since any Pauli operator is a finite product of these,
the result follows by linearity.
\end{proof}

We wish to characterize a $\FF_2$-subspace $S$ of the Pauli module
invariant under the action of $\Lambda$, i.e., a submodule,
on which the commutation value is always zero.
As we will see in the next subsection,
this particular subspace yields a local Hamiltonian whose energy spectrum is exactly solvable,
which is the main object of this paper.
Let $(a,b)$ be an element of $S \subseteq R^2 = (\FF_2[\Lambda])^2$.
For any $r \in R$, $(ra, rb)$ must be a member of $S$.
Demanding that the symplectic form on $S$ vanish,
by Proposition~\ref{prop:trace-formula-commutation-value} we have
\[
 \tr (ra \bar b - rb \bar a) = 0 .
\]
Since $r$ was arbitrary, we must have $a \bar b - b \bar a = 0$.%
\footnote{{A symmetric bilinear form $ \langle r, s \rangle = \tr(r \bar{s}) $ on $R$ is non-degenerate.}}
Let us denote $\begin{pmatrix}\bar a & \bar b \end{pmatrix}$ 
as $\begin{pmatrix} a \\ b \end{pmatrix}^\dagger$,
and write any element of $R^2$ as a $2 \times 1$ matrix.
We conclude that $S$ is a submodule of $R^2$ over $R$ generated by $s_1, \ldots, s_t$
such that any commutation value always vanishes, if and only if
\[
 s_i^\dagger \lambda_1 s_j = 0
\]
for all $i,j = 1,\ldots,t$.

The requirement that $\Lambda$ be a group might be too restrictive.
One may have a coarse group structure on $\Lambda$, the index set of all qubits.
We consider the case that the index set is a product of a finite set and a group.
By abuse of notation, we still write $\Lambda$ to denote the group part,
and insist that to each group element are associated $q$ qubits ($q \ge 1$).
Thus obtained Pauli module should now be identified with $R^{2q}$,
where $R = \FF_2[\Lambda]$ is the group algebra
that encodes the notion of translation.
We write an element $v$ of $R^{2q}$ by a $2q \times 1$ matrix,
and denote by $v^\dagger$ the transpose matrix of $v$ whose each entry is applied by the antipode map.
We always order the entries of $v$ such that the upper $q$ entries
describes the $\sigma_x$-part and the lower the $\sigma_z$-part.
Since the commutation value on $R^{2q}$ is the sum of commutation values on $R^2$,
we have the following:
If $S$ is a submodule of $R^{2q}$ over $R$ generated by $s_1, \ldots, s_t$,
the commutation value always vanishes on $S$, if and only if for all $i,j = 1, \ldots, t$
\[
 s_i^\dagger \lambda_q s_j = 0
\]
where $\lambda_q = \begin{pmatrix} 0 & \id_q \\ -\id_q & 0 \end{pmatrix}$ is a $2q \times 2q$ matrix.

Let us summarize our discussion so far.
\begin{prop}
On a set of qubits $\Lambda \times \{1,\ldots,q\}$ where $\Lambda$ is an abelian group,
the group of all Pauli operators of finite support up to phase factors,
form a free module $P=R^{2q}$ over the group algebra $R = \FF_2[\Lambda]$.
The commutation value
\[ 
 \langle a, b \rangle = \tr( a^\dagger \lambda_q b )
\]
for $a,b \in P$ is zero
if and only if the Pauli operators corresponding to $a$ and $b$ commute.
If $\sigma$ is a $2q \times t$ matrix whose columns generate a submodule $S \subseteq P$,
then the commutation value on $S$ always vanishes if and only if
\[
 \sigma^\dagger \lambda_q \sigma = 0 .
\]
\label{prop:pauli-module}
\end{prop}
Proposition~\ref{prop:trace-formula-commutation-value}~\cite{CalderbankRainsShorEtAl1997Quantum}
is a special case of Proposition~\ref{prop:pauli-module} when $\Lambda$ is a trivial group.

\subsection{Local Hamiltonians on groups}
Recall that we place $q$ qubits on each {\em site} of $\Lambda$.
The total system of the qubits is $\Lambda \times \{1,\ldots,q\}$.
\begin{defn}
Let
\[
 H = -\sum_{g \in \Lambda} h_{1,g} + \cdots + h_{t,g}
\]
be a local Hamiltonian consisted of Pauli operators that is (i)commuting, (ii) translation-invariant up to signs, and (iii) frustration-free.
We call $H$ a {\em code Hamiltonian} (also known as {\em stabilizer Hamiltonian}).
The {\em stabilizer module} of $H$ is the submodule of the Pauli module $P$
generated by the images of $h_1,\ldots,h_t$ in $P$.
The number of {\em interaction types} is $t$.
\end{defn}
The energy spectrum of the code Hamiltonian is trivial;
it is discrete and equally spaced.

\begin{example}
One dimensional Ising model is the Hamiltonian
\[
 H = - \sum_{i \in \mathbb{Z}} \sigma_z^{(i)} \otimes \sigma_z^{(i+1)} .
\]
The lattice is the additive group $\mathbb{Z}$,
and the group algebra is $R=\FF_2[x,\bar x]$.
The Pauli module is $R^2$ and the stabilizer module $S$ is generated by
\[
\begin{pmatrix}
 0 \\
 1+x
\end{pmatrix} .
\]
One can view this as the matrix $\sigma$ of Proposition~\ref{prop:pauli-module}.
$H$ is commuting; $\sigma^\dagger \lambda_1 \sigma = 0$.
\hfill $\Diamond$
\label{eg:1d-ising}
\end{example}

\subsection{Excitations}

For a code Hamiltonian $H$,
an excited state is described by the terms in the Hamiltonian that have eigenvalues $-1$.
Each of the flipped terms is interpreted as an {\em excitation}.
Although the actual set of all possible configurations of excitations
that are obtained by applying some operator to a ground state, may be quite restricted,
it shall be convenient to think of a larger set.
Let $E$ be the set of all configurations of finite number of excitations without asking physical relevance.
Since an excitation is by definition a flipped term in $H$,
the set $E$ is equal to the collection of all finite sets consisted of the terms in $H$.

If Pauli operators $U_1, U_2$ acting on a ground state creates excitations $e_1, e_2 \in E$,
their product $U_1 U_2$ creates excitations $(e_1 \cup e_2) \setminus (e_1 \cap e_2)$.
Here, we had to remove the intersection because each excitation is its own annihilator;
any term in the $H$ squares to the identity.
Exploiting this fact, we make $E$ into a vector space over $\FF_2$.
Namely, we take formal linear combinations of terms in $H$
with the coefficient $1 \in \FF_2$ when the terms has $-1$ eigenvalue,
and the coefficient $0 \in \FF_2$ when the term has $+1$ eigenvalue.
The symmetric difference is now expressed as the sum of two vectors $e_1 + e_2$ over $\FF_2$.
In view of Pauli group as a vector space,
$U_1 U_2$ is the sum of the two vectors $v_1 + v_2$ that respectively represents $U_1, U_2$.
Therefore, the association $U_i \mapsto e_i$ induces 
a linear map from the Pauli space to the space of virtual excitations $E$.

The set of all excited states obeys the translation-invariance as the code Hamiltonian $H$ does.
So, $E$ is a module over the group algebra $R=\FF_2[\Lambda]$.
The association $U_i \mapsto e_i$ clearly respects this translation structure.
Our discussion is summarized by saying that the excitations are described by an $R$-linear map
\[
 \epsilon : P \to E
\]
from the Pauli module $P$ to the {\em module of virtual excitations $E$}.

As the excitation module is the collection of all finite sets of the terms in $H$,
we can speak of the {\em module of generator labels} $G$,
which is equal to $E$ as an $R$-module.
$G$ is a free module of rank $t$ if there are $t$ types of interaction.
The matrix $\sigma$ introduced in Section~\ref{sec:Pauli-space-on-group} can be viewed as
\[
 \sigma : G \to P
\]
from the module of generator labels to the Pauli module.
\begin{prop}
If $\sigma$ is the generating map for the stabilizer module of a code Hamiltonian,
then 
\[
 \epsilon = \sigma^\dagger \lambda_q .
\]
\end{prop}
The matrix $\epsilon$ can be viewed as a generalization of
the parity check matrix of the standard theory of classical or quantum error correcting codes%
~\cite{MacWilliamsSloane1977,CalderbankShor1996Good,Steane1996Multiple,Gottesman1996Saturating},
when a translation structure is given.

\begin{proof}
This is a simple corollary of Proposition~\ref{prop:pauli-module}.
Let $h_{i,g}$ be the terms in the Hamiltonian where $i = 1,\ldots, t$, and $ g \in \Lambda $.
In the Pauli module, they are expressed as $g h_i$ where $h_i$ is the $i$-th column of $\sigma$.
For any $u \in P$, let $\epsilon(u)_i$ be the $i$-th component of $\epsilon(u)$. By definition,
\[
 \epsilon( u )_i 
= \sum_{g \in \Lambda} g~ \tr \left( (g h_i)^\dagger \lambda_q u \right)
= \sum_{g \in \Lambda} g~ \tr \left( \bar g h_i^\dagger \lambda_q u  \right)
= h_i^\dagger \lambda_q u
\]
Thus, $h_i^\dagger \lambda_q$ is the $i$-th row of $\epsilon$.
\end{proof}

\begin{rem}
The commutativity condition in Proposition~\ref{prop:pauli-module} of the code Hamiltonian
is recast into the condition that
\[
 G \xrightarrow{\sigma} P \xrightarrow{\epsilon} E
\]
be a complex, i.e., $\epsilon \circ \sigma = 0$.
Equivalently,
\[
 \im \sigma \subseteq (\im \sigma)^\perp = \ker \epsilon
\]
where $\perp$ is with respect to the symplectic form.
\end{rem}

\section{Equivalent Hamiltonians}

The stabilizer module entirely determines
the physical phase of the code Hamiltonian
in the following sense.
\begin{prop}
Let $H$ and $H'$ be code Hamiltonians on a system of qubits,
and suppose their stabilizer modules are the same.
Then, there exists a unitary
\[
 U = \bigotimes_{g \in \Lambda} U_{g}
\]
mapping the ground space of $H$ onto that of $H'$.
Moreover, there exist a continuous one-parameter family of gapped Hamiltonians
connecting $U H U^\dagger$ and $H'$.
\label{prop:associated-module-invariance}
\end{prop}
\begin{proof}
Let $\{ p_\alpha \}$ be a maximal set of $\FF_2$-linearly independent Pauli operators of finite support
that generates the common stabilizer module $S$.
$\{p_\alpha\}$ is not necessarily translation-invariant.
Any ground state $\ket \psi$ of $H$ is a common eigenspace of $\{p_\alpha\}$
with eigenvalues $p_\alpha \ket \psi = e_\alpha \ket \psi$, $e_\alpha = \pm 1$.
Similarly, the ground space of $H'$ gives the eigenvalues $e'_\alpha = \pm 1$ for each $p_\alpha$.

The abelian group generated by $\{ p_\alpha \}$ is precisely the vector space $S$,
and the assignment $p_\alpha \mapsto e_\alpha$ defines a dual vector on $S$.
If $U$ is a Pauli operator of possibly infinite support,
then $ p_\alpha U \ket \psi = e''_\alpha e_\alpha U \ket \psi$ for some $e''_\alpha = \pm 1$,
where $e''_\alpha$ is determined by the commutation relation between $U$ and $p_\alpha$.
Thus, the first statement follows if we can find $U$
such that the commutation value between $U$ and $p_\alpha$ is precisely $e''_\alpha$.
This is always possible since the dual space of the vector space $P$
is isomorphic to the direct product $\prod _{\Lambda \times \{1,\ldots, q\}} \FF_2^2$,
which is vector-space-isomorphic to the Pauli group of arbitrary support up to phase factors.%
\footnote{{If $V$ is a finite dimensional vector space over some field,%
the dual vector space of $\bigoplus_I V$ %
is isomorphic to $\prod_{I} V$ where $I$ is an arbitrary index set.}}

Now, $U H U^\dagger$ and $H'$ have the same eigenspaces, and in particular, the same ground space.
Consider a continuous family of Hamiltonians
\[
H(u,u') = u U H U^\dagger + u' H'
\]
where $u,u' \in \mathbb{R}$.
It is clear that
\[
 H = H(1,0) \to H(1,1) \to H(0,1) = H'
\]
is a desired path.
\end{proof}

The criterion of Proposition~\ref{prop:associated-module-invariance}
to classify the physical phases is too narrow.
Physically meaningful universal properties should be invariant
under simple and local changes of the system. More concretely,
\begin{defn}
Two code Hamiltonians $H$ and $H'$ are {\em equivalent} if
their stabilizer modules become the same under a finite composition
of symplectic transformations, coarse-graining, and tensoring ancillas.
\label{defn:equiv-H}
\end{defn}
We shall define the symplectic transformations,
the coarse-graining, and the tensoring ancillas shortly.

\subsection{Symplectic transformations}
\label{sec:symplectic-transformations}

\begin{defn}
A {\em symplectic transformation} $T$ is an automorphism of the Pauli module
induced by a unitary operator on the system of qubits
such that 
\[
  T^\dagger \lambda_q T = \lambda_q
\]
where $\dagger$ is the transposition followed by the entry-wise antipode map.
\label{defn:symplectic-transformation}
\end{defn}
When the translation group is trivial these transformations are given by so-called \emph{Clifford operators}.
See \cite[Chapter 15]{KitaevShenVyalyi2002CQC}.

Only the unitary operator on the physical Hilbert space 
that respects the translation can induce a symplectic transformation.
By definition, a symplectic transformation maps each local Pauli operator to a local Pauli operator, and 
preserves the commutation value for any pair of Pauli operators.
\begin{prop}
Any two unitary operators $U_1, U_2$ that induce the same symplectic transformation
differ by a Pauli operator (of possibly infinite support).
\end{prop}
If the translation group is trivial,
the proposition reduces to Theorem~15.6 of \cite{KitaevShenVyalyi2002CQC}
\begin{proof}
The symplectic transformation induced by $U = U_1^\dagger U_2$ is the identity.
Hence, $U$ maps each single qubit Pauli operator $\sigma_{x,z}^{(g,i)}$ to $\pm \sigma_{x,z}^{(g,i)}$.
By the argument as in the proof of Proposition~\ref{prop:associated-module-invariance},
there exists a Pauli operator $O$ of possibly infinite support
that acts the same as $U$ on the system of qubits.
Since Pauli operators form a basis of the operator algebra of qubits,
we have $O=U$.
\end{proof}

The effect of a symplectic transformation on the generating map $\sigma$
 is a matrix multiplication on the left.
\[
 \sigma \to U \sigma
\]
For example, the following is induced
by uniform Hadamard, controlled-Phase, and controlled-NOT gates.
For notational clarity,
define $E_{i,j}(a)\ (i \neq j)$ as the row-addition elementary $2q \times 2q$ matrix
\[
 \left[ E_{i,j}(a) \right]_{\mu \nu} = \delta_{\mu \nu} + \delta_{\mu i} \delta_{\nu j} a
\]
where $\delta_{\mu \nu}$ is the Kronecker delta and $a \in R = \FF_2[\Lambda]$.
Recall that we order the components of $P$ such that the first half components are for $\sigma_x$-part, 
and the second half components are for $\sigma_z$-part.
\begin{defn}
The following are {\em elementary symplectic transformations}:
\begin{itemize}
 \item (Hadamard) $E_{i,i+q}(-1) E_{i+q,i}(1) E_{i,i+q}(-1)$ where $1 \le i \le q$,
 \item (controlled-Phase) $E_{i+q,i}(f)$ where $f = \bar f$ and $1 \le i \le q$,
 \item (controlled-NOT) $E_{i,j}(a) E_{j+q,i+q}(-\bar a)$ where $1 \le i \ne j \le q$.
\end{itemize}
\end{defn}
For the case of a trivial translation group,
these transformations explicitly appear in \cite{CalderbankRainsShorEtAl1997Quantum}
and \cite[Chapter~15]{KitaevShenVyalyi2002CQC}.

Recall that the Hadamard gate is a unitary transformation on a qubit given by
\[
U_H = \frac{1}{\sqrt 2}
 \begin{pmatrix}
  1 & 1 \\
  1 & -1
 \end{pmatrix}
\]
with respect to basis $\{ \ket 0 , \ket 1 \}$.
At operator level,
\[
 U_H X U_H^\dagger = Z, \quad U_H Z U_H^\dagger = X 
\]
where $X$ and $Z$ are the Pauli matrices $\sigma_x$ and $\sigma_z$, respectively.
Thus, the application of Hadamard gate on every $i$-th qubit of each site of $\Lambda$
swaps the corresponding $X$ and $Z$ components of $P$.

The controlled phase gate is a two-qubit unitary operator whose matrix is
\[
U_P =
 \begin{pmatrix}
  1 & 0 & 0 & 0 \\
  0 & 1 & 0 & 0 \\
  0 & 0 & 1 & 0 \\
  0 & 0 & 0 & -1
 \end{pmatrix}
\]
with respect to basis $\{ \ket{00}, \ket{01}, \ket{10}, \ket{11} \}$.
At operator level,
\begin{align*}
U_P (X \otimes I ) U_P^\dagger = X \otimes Z, & &
U_P (Z \otimes I ) U_P^\dagger = Z \otimes I, \\
U_P (I \otimes X ) U_P^\dagger = Z \otimes X, & &
U_P (I \otimes Z ) U_P^\dagger = I \otimes Z.
\end{align*}
Note that since $U_P$ is diagonal, any two $U_P$ on different pairs of qubits commute.
Let $(g,i)$ denote the $i$-th qubit at $g \in \Lambda$.
The uniform application
\[
U^{(i)}_g = \prod_{h \in \Lambda} U_P( (h,i), (h+g,i) )
\]
of $U_P$ throughout the lattice $\Lambda$
such that each $U_P( (h,i), (h+g,i) )$ acts on the pair of qubits $(h,i)$ and $(h+g,i)$
is well-defined. From the operator level calculation of $U_P$, we see that $U^{(i)}_g$
induces
\[
P \ni (\ldots,x_i,\ldots, z_i, \ldots) \mapsto ( \ldots, x_i,\ldots, z_i + (g+ \bar g)x_i, \ldots ) \in P
\]
on the Pauli module,
which is represented as $E_{i+q,i}(g+\bar g)$. 
The composition
\[
 U^{(i)}_{g_1} U^{(i)}_{g_2} \cdots U^{(i)}_{g_n}
\]
of finitely many controlled-Phase gates $U^{(i)}_g$ with different $g$
is represented as $E_{i+q,i}(f)$ where $f = \bar f = \sum_{k=1}^{n} g_k + \bar g_k$.
The single qubit phase gate
\[
 \begin{pmatrix}
  1 & 0 \\
  0 & i
 \end{pmatrix}
\]
maps $X \leftrightarrow Y$ and $Z \mapsto Z$. On the Pauli module $P$, it is
\[
P \ni (\ldots,x_i,\ldots, z_i, \ldots)^T \mapsto ( \ldots, x_i,\ldots, z_i + x_i, \ldots )^T \in P .
\]
which is $E_{i+q,i}(1)$.
Note that any $f \in R$ such that $f = \bar f$ is always of form
$f = \sum g_k + \bar g_k$ or $f = 1 + \sum g_k + \bar g_k$
where $g_k$ are monomials.
Thus, the Phase gate and the controlled-Phase gate induce
transformations $E_{i+q,i}(f)$ where $f = \bar f$. 

The controlled-NOT gate is a two-qubit unitary operator whose matrix is
\[
U_N =
 \begin{pmatrix}
  1 & 0 & 0 & 0 \\
  0 & 1 & 0 & 0 \\
  0 & 0 & 0 & 1 \\
  0 & 0 & 1 & 0
 \end{pmatrix}
\]
with respect to basis $\{ \ket{00}, \ket{01}, \ket{10}, \ket{11} \}$.
That is, it flips the {\em target} qubit conditioned on the {\em control} qubit.
At operator level,
\begin{align*}
U_N (X \otimes I ) U_N^\dagger = X \otimes X, & &
U_N (Z \otimes I ) U_N^\dagger = Z \otimes I, \\
U_N (I \otimes X ) U_N^\dagger = I \otimes X, & &
U_N (I \otimes Z ) U_N^\dagger = Z \otimes Z.
\end{align*}
If $i < j$, the uniform application
\[
U^{(i,j)}_g = \bigotimes_{h \in \Lambda} U_P( (h,i), (h+g,j) )
\]
such that each $U_N( (h,i), (h+g,j) )$ acts on the pair of qubits $(h,i)$ and $(h+g,j)$
with one at $(h,i)$ being the control
induces
\begin{align*}
P \ni & (\ldots,x_i,\ldots,x_j,\ldots, z_i, \ldots, z_j, \ldots)^T \\
& \mapsto ( \ldots, x_i, \ldots, x_j + g x_i, \ldots, z_i + \bar g z_j, \ldots, z_j, \ldots )^T \in P .
\end{align*}
Thus, any finite composition of controlled-NOT gates with various $g$ is of form $E_{i,j}(a) E_{j+q,i+q}(\bar a)$.
It might be useful to note that the controlled-NOT and the Hadamard combined,
induces a symplectic transformation
\begin{itemize}
\item (controlled-NOT-Hadamard) $E_{i+q,j}(a) E_{j+q,i}(\bar a)$ where $a \in R$ and $ 1 \le i \ne j \le q$.
\end{itemize}

Remark that an arbitrary row operation on the upper $q$ components
can be compensated by a suitable row operation on the lower $q$ components
so as to be a symplectic transformation.

\subsection{Coarse-graining}

Not all unitary operators conform with the lattice translation.
In Example~\ref{eg:1d-ising} the lattice translation has period 1.
Then, for example, the Hadamard gate on every second qubit 
does not respect this translation structure;
it only respects a coarse version of the original translation.
We need to shrink the translation group to treat such unitary operators.

Let $\Lambda$ be the original translation group of the lattice with $q$ qubits per site,
and $\Lambda'$ be its subgroup of finite index: $|\Lambda/\Lambda'| = c < \infty$.
The total set of qubits $\Lambda \times \{1,\ldots,q\}$ is set-theoretically the same as
$\Lambda' \times \{ 1, \ldots, c \} \times \{1,\ldots,q\} = \Lambda' \times \{ 1, \ldots, cq\}$.
We take $\Lambda'$ as our new translation group under coarse-graining.
The Pauli group modulo phase factors remains the same as a $\FF_2$-vector space
for it depends only on the total index set of qubits.
We shall say that the system is {\em coarse-grained by $R'=\FF_2[\Lambda']$}
if we restrict the scalar ring $R$ to $R'$ for all modules pertaining to the system.

For example, suppose $\Lambda = \mathbb{Z}^2$,
so the original base ring is $R = \FF_2[x,y,\bar x,\bar y]$.
If we coarse-grain by $R' = \FF_2[x',y', \bar x', \bar y']$ where $x' = x^2, y' = y^2$,
we are taking the sites $1,x,y,xy$ of the original lattice as a single new site.

\subsection{Tensoring ancillas}
We have considered possible transformations
on the stabilizer modules of code Hamiltonians,
and kept the underlying index set of qubits invariant.
It is quite natural to allow tensoring ancilla qubits in trivial states.
In terms of the stabilizer module $S \subseteq P=R^{2q}$,
it amounts to embed $S$ into the larger module $R^{2q'}$ where $q' > q$.
Concretely,
let $\sigma = \begin{pmatrix} \sigma_X \\ \sigma_Z \end{pmatrix}$
be the generating matrix of $S$ as in Proposition~\ref{prop:pauli-module}.
By {\em tensoring ancilla}, we embed $S$ as
\[
\begin{pmatrix} \sigma_X \\ \sigma_Z \end{pmatrix}
 \to 
\begin{pmatrix}
 \sigma_X & 0 \\
 0        & 0 \\
 \sigma_Z & 0 \\
 0        & 1
\end{pmatrix} .
\]
This amounts to taking the direct sum of the original complex
\[
 G \xrightarrow{\sigma} P \xrightarrow{\epsilon} E
\]
and the trivial complex
\[
0
\to
 R
 \xrightarrow{ \begin{pmatrix} 0 \\ 1 \end{pmatrix} }
 R^2
 \xrightarrow{ \begin{pmatrix} 1 & 0 \end{pmatrix} }
 R
\to 0
\]
to form
\[
  G \oplus R \xrightarrow{} P \oplus R^2 \xrightarrow{} E \oplus R .
\]

\section{Topological order}
\label{sec:topological-order}
From now on we assume that $\Lambda$ is isomorphic to $\mathbb{Z}^D$ as an additive group.
$D$ shall be called the {\em spatial dimension} of $\Lambda$.

\begin{defn}
Let $\sigma: G \to P$ be the generating map for the stabilizer module of
a code Hamiltonian $H$. We say $H$ is {\em exact}
if $(\im \sigma )^\perp = \im \sigma$, or equivalently
\[
 G \xrightarrow{\sigma} P \xrightarrow{\epsilon=\sigma^\dagger \lambda_q} E
\]
is exact, i.e., $\ker \epsilon = \im \sigma$.
\end{defn}
\noindent It follows that the exactness condition is
a property of the equivalence class of code Hamiltonians in the sense of Definition~\ref{defn:equiv-H}.

By imposing periodic boundary conditions,
a translation-invariant Hamiltonian yields
a family of Hamiltonians $\{ H(L) \}$ defined on a finite system consisted of $L^D$ sites.
One might be concerned that some $H(L)$ would be frustrated.
We intentionally exclude such a situation.
The frustration might indeed occur, but it can easily be resolved
by choosing the signs of terms in the Hamiltonian.
In this way, one might loose the translation-invariance in a strict sense.
However, we retain the physical phase regardless of the sign choice
because different sign choices are related by a Pauli operator acting on
the whole system which is a product unitary operator.
Hence, the entanglement property of the ground state 
and the all properties of excitations do not change.

\begin{defn}
Let $H(L)$ be Hamiltonians on a finite system of linear size $L$
in $D$ dimensional physical space,
and $\Pi_L$ be the corresponding ground space projector.
$H(L)$ is called {\em topologically ordered} if
for any $O$ supported inside a hypercube of size $(L/2)^D$
one has
\begin{equation}
 \Pi_L O \Pi_L \propto \Pi_L .
\label{eq:tqo}
\end{equation}
\end{defn}
This means that no local operator is capable of distinguishing different ground states.
This condition is trivially satisfied if $H(L)$ has a unique ground state.
A technical condition that is used in the proof of the stability of topological order
against small perturbations is the following `local topological order'
condition~\cite{MichalakisPytel2011stability,
BravyiHastingsMichalakis2010stability,
BravyiHastings2011short}.
We say a {\em diamond region $A(r)$ of radius $r$ at $o \in \ZZ^D$} for the set
\[
 A(r)_o = \left\{ (i_1,\ldots,i_D) + o \in \ZZ^D ~\middle| \sum_\mu |i_\mu| \le r \right\} .
\]
\begin{defn}
Let $H(L)$ be code Hamiltonians on a finite system of linear size $L$
in $D$ dimensional physical space.
For any diamond region $A=A(r)$ of radius $r$,
let $\Pi_A$ be the projector onto the common eigenspace of the most negative eigenvalues of terms
in the Hamiltonian $H(L)$ that are supported in $A$.
For $b > 0$, denote by $A^b$ the distance $b$ neighborhood of $A$.
$H(L)$ is called {\em locally topologically ordered} if
there exists a constant $b > 0$ such that
for any operator $O$ supported on a diamond region $A$ of radius $r < L/2$
one has
\begin{equation}
 \Pi_{A^b} O \Pi_{A^b} \propto \Pi_{A^b} .
\label{eq:local-tqo}
\end{equation}
\end{defn}
Since any operator is a $\mathbb{C}$-linear combination of Pauli operators,
if Eq.~\eqref{eq:tqo},\eqref{eq:local-tqo} are satisfied for Pauli operators,
then the (local) topological order condition follows.
If a Pauli operator $O$ is anti-commuting with a term in a code Hamiltonian $H(L)$,
The left-hand side of Eq.~\eqref{eq:tqo},\eqref{eq:local-tqo} are identically zero.
In this case, there is nothing to be checked.
If $O$ acting on $A$ is commuting with every term in $H(L)$ supported inside $A^b$,
Eq.~\eqref{eq:tqo} demands that it act as identity on the ground space,
i.e., $O$ must be a product of terms in $H(L)$ up to $\pm i,\pm 1$.
Eq.~\eqref{eq:local-tqo} further demands that
$O$ must be a product of terms in $H(L)$ supported inside $A^b$ up to $\pm i,\pm 1$.

\begin{lem}
A code Hamiltonian $H$ is exact
if and only if $H(L)$ is locally topologically ordered for
all sufficiently large $L$.
\label{lem:local-tqo=exact}
\end{lem}
In order to see this, it will be important to use {\em Laurent polynomials}
to express elements of the group algebra 
$R = \FF_2[\mathbb{Z}^D] \cong \FF_2[x_1,x_1^{-1},\ldots,x_D,x_D^{-1}]$.
See also \cite{GueneriOezbudak2008}.
For example, 
\[
 x y^2 z^2 + x y^{-1} \quad \Longleftrightarrow \quad 1(1,2,2)+1(1,-1,0) .
\]
The sum of the absolute values of exponents of a monomial will be referred to as {\em absolute degree}.
The absolute degree of a Laurent polynomial is defined to be the maximum absolute degree of its terms.
The degree measures the distance or size in the lattice.

The Laurent polynomial viewpoint enables us to apply Gr\"obner basis techniques.
The long division algorithm for polynomials in one variable yields an effective and efficient
test whether a given polynomial is divisible by another.
When two or more but finitely many variables are involved,
a more general question is how to test
whether a given polynomial is a member of an ideal.
For instance, $f=xy-1$ is a member of an ideal $J = (x-1,y-1)$
because $xy -1 = y(x-1) + (y-1)$.
But, $g=xy$ is not a member of $J$ because $g = y(x-1) + (y-1) + 1$ and the `remainder' 1 cannot be removed.
Here, the first term is obtained by looking at the initial term $xy$ of $f$
and comparing with the initial terms $x$ and $y$ of the generators of $J$.
While one tries to eliminate the initial term of $f$ and to eventually reach zero,
if one cannot reach zero as for $g$, then the membership question is answered negatively.

Systematically, an well-ordering on the monomials, i.e., a \emph{term order}, is defined
such that the order is preserved by multiplications.
And a set of generators $\{ g_i \}$ for the ideal is given with a special property 
that any element in the ideal has an initial term (leading term)
divisible by an initial term of some $g_i$.
A Gr\"obner basis is precisely such a generating set.
This notion generalizes to free modules over polynomial ring by refining the term order with the basis of the modules.
An example is as follows. Let
\[
 \sigma_1 = 
 \begin{pmatrix}
  \mathbf{x^2} - y \\
  x^2 + 1
 \end{pmatrix} \quad
 \sigma_2 = 
 \begin{pmatrix}
  1 \\
  \mathbf{y}
 \end{pmatrix}
\]
generate a submodule $M$ of $S^2$ where $S = \FF[x,y]$ is a polynomial ring.
They form a Gr\"obner basis, and the initial terms are marked as bold.
A member of $S^2$
\[
 \begin{pmatrix}
 x^2 + x^2 y - y^2 \\
 y+2 x^2 y
 \end{pmatrix}
\]
is in $M$ because the following ``division'' results in zero.
\[
  \begin{pmatrix}
 x^2+\mathbf{x^2 y}-y^2 \\
 y+2 x^2 y
 \end{pmatrix}
 \xrightarrow{-y \sigma_1}
 \begin{pmatrix}
 x^2 \\
 \mathbf{x^2 y}
 \end{pmatrix}
 \xrightarrow{-x^2 \sigma_2} 0
\]
A comprehensive material can be found in \cite[Chapter~15]{Eisenbud}.

The situation for Laurent polynomial ring is less discussed,
but is not too different.
A direct treatment is due to Pauer and Unterkircher~\cite{PauerUnterkircher1999}.
One introduces a well-order on monomials,
that is preserved by multiplications with respect to a so-called \emph{cone decomposition}.
An ideal $J$ over a Laurent polynomial ring can be thought of
as a collection of configurations of coefficient scalars written on the sites of the integral lattice $\ZZ^D$.
If we take a cone, say,
\[
 C = \{ (i_1,i_2,i_3) \in \ZZ^3 | i_1 \le 0, i_2 \ge 0, i_3 \ge 0 \} ,
\]
then $J_C = J \cap \FF[C]$ looks very similar to an ideal $I$ over a polynomial ring $\FF[x,y,z]$.
Concretely, $I$ can be obtained by applying $x^{-1} \mapsto x, y \mapsto y, z \mapsto z$ to $J_C$.
The initial terms of $J_C$ should be treated similarly as those in $I$.
This is where the cone decomposition plays a role.
The lattice $\ZZ^D$ decomposes into $2^D$ cones,
and the initial terms of $J$ is considered in each of the cones.
Correspondingly, a Gr\"obner basis is defined 
to generate the initial terms of a given module in each of the cones.
An intuitive picture for the division algorithm is 
to consider the support of a Laurent polynomial as a finite subset of $\ZZ^D$ around the origin (the least element of $\ZZ^D$),
and to eliminate outmost points so as to finally reach the origin.
If $m$ is a column matrix of Laurent polynomials,
each step in the division algorithm by a Gr\"obener basis $\{g\}$
replaces $m$ with $m' = m - c g$, where $c$ is a monomial,
such that the initial term of $m'$ is strictly smaller than that of $m$.
Note that the absolute degree of $c$ does not exceed that of $m$.%
\footnote{Strictly speaking, one can introduce a term order such that this is true.}

\begin{proof}[Proof of \ref{lem:local-tqo=exact}]
We have to show that 
if $v \in \ker \epsilon = \im \sigma$ is supported in the diamond of radius $r$ centered at the origin,
then $v$ can be expressed as a linear combination
\[
 v = \sum_i c_i \sigma_i
\]
of the columns $\sigma_i$ of $\sigma$ 
such that the coefficients $c_i \in R$ have absolute degree not exceeding $w+r$.
for some fixed $w$.
A Gr\"obner basis~\cite{PauerUnterkircher1999} is computed solely from the matrix $\sigma$,
and the division algorithm yields desired $c_i$.%
\footnote{{This part can be adapted to an error correcting procedure or a decoder.
The bottleneck of the universal decoder presented in~\cite{BravyiHaah2011Memory}
is the routine that tests 
whether a given cluster of excitations can be created by a Pauli operator 
supported in the box that envelops the cluster.
The Gr{\"o}bner basis for $\im \epsilon$ in the degree monomial order
provides a fast algorithm for it: 
The division algorithm yields zero remainder with respect to the Gr\"obner basis,
if and only if
the given cluster is in $\im \epsilon$.
Note also that this argument proves that 
the topological order condition as defined in~\cite{BravyiHaah2011Memory} is always satisfied
if the code Hamiltonian is exact.
}}

Conversely,
suppose $v \in \ker \epsilon$. We have to show $v \in \im \sigma$.
Choose so large $L$ that the Pauli operator $O$ representing $v$ is contained
in a pyramid region far from the boundary.
The local topological order condition implies
that $O$ is a product of terms near the pyramid region.
Since this product expression is independent of the boundary,
we see $v \in \im \sigma$.
\end{proof}

The Buchsbaum-Eisenbud theorem~\cite{BuchsbaumEisenbud1973Exact}
below characterizes an exact sequence
from the properties of connecting maps.
(See also \cite[Theorem~20.9, Proposition~18.2]{Eisenbud},\cite[Chapter~6 Theorem~15]{Northcott}.)
A few notions should be recalled.
Let $\mathbf M$ be a matrix, not necessarily square, over a ring.
A minor is the determinant of a square submatrix of $\mathbf M$.
{\em $k$-th determinantal ideal} $I_k(\mathbf M)$ is the ideal generated by all $k \times k$ minors of $\mathbf M$.
It is not hard to see that the determinantal ideal is invariant under any invertible
matrix multiplication on either side. 
The {\em rank} of $\mathbf M$ is the largest $k$ such that $k$-th determinantal ideal is nonzero.
Thus, the rank of a matrix over an arbitrary ring is defined,
although the dimension of the image in general is not defined or is infinite.
The $0$-th determinantal ideal is taken to be the unit ideal by convention.
For a map $\phi$ between free modules,
we write $I(\phi)$ to denote the $k$-th determinantal
ideal of the matrix of $\phi$ where $k$ is the rank of that matrix.
Fitting Lemma~\cite[Corollary-Definition~20.4]{Eisenbud}
states that determinantal ideals only depend on $\coker \phi$.

The {\em (Krull) dimension} of a ring is the supremum of lengths of chains of prime ideals.
Here, the length of a chain of prime ideals
\[
 \pp_0 \subsetneq \pp_1 \subsetneq \cdots \subsetneq \pp_n
\]
is defined to be $n$.
Most importantly, the dimension of $\FF[x_1,\ldots,x_n]$ is $n$ where $\FF$ is a field, as
\[
 (0) \subset (x_1) \subset (x_1,x_2) \subset \cdots \subset (x_1,\ldots,x_n) .
\]
Dimensions are in general very subtle,
but intuitively, it counts the number of independent `variables.'
Geometrically, a ring is a function space of a geometric space,
and the independent variables define a coordinate system on it.
So the Krull dimension correctly captures the intuitive dimension.
For instance, $y-x^2=0$ defines a parabola in a plane,
and the functions that vanish on the parabola form an ideal $(y-x^2) \subset \FF[x,y]$.
Thus, the function space is identified with $\FF[x,y]/(y-x^2) \cong \FF[x]$,
whose Krull dimension is, as expected, 1.

Facts we need are quite simple:
\begin{itemize}
 \item In a zero-dimensional ring, every prime ideal is maximal.
 \item $\dim R = \dim \FF_2[x_1^{\pm 1},\ldots,x_D^{\pm 1}] = D$
 \item When $I$ is an ideal of $R$, $\dim R/I + \codim I = D$.%
\footnote{
The {\em codimension} or {\em height} of a prime ideal $\pp$ is 
the supremum of the lengths of chains of prime ideals contained in $\pp$.
That is, the codimension of $\pp$ is the Krull dimension of the local ring $R_\pp$.
The codimension of an arbitrary ideal $I$ is the minimum of codimensions of primes that contain $I$.
If $S$ is an affine domain, i.e., 
a homomorphic image of a polynomial ring over a field with finitely many variables
such that $S$ has no zero-divisors,
it holds that $\codim I + \dim R/I = \dim S$~\cite[Chapter~13]{Eisenbud}.
}
\end{itemize}

We shall be dealing with three different kinds of `dimensions':
The first one is the spatial dimension $D$,
which has an obvious physical meaning.
The second one is the Krull dimension of a ring, just introduced.
The Krull dimension is upper bounded by the spatial dimension in any case.
The last one is the dimension of some module as a vector space.
Recall that all of our base ring contains a field -- $\FF_2$ for qubits.
The vector space dimension arises naturally
when we actually count the number of orthogonal ground states.
The dimension as a vector space will always be denoted with a subscript like $\dim_{\FF_2}$.

\begin{prop}
\cite{BuchsbaumEisenbud1973Exact}%
\footnote{The original result is stronger than what is presented here.
It is stated with the \emph{depth}s of the determinantal ideals.}
If a complex of free modules over a ring
\[
 0 \to F_n \xrightarrow{\phi_n} F_{n-1} \to \cdots \to F_1 \xrightarrow{\phi_1} F_0
\]
is exact, then
\begin{itemize}
\item $\rank F_k = \rank \phi_k + \rank \phi_{k+1}$ for $k=1,\ldots,n-1$
\item $\rank F_n = \rank \phi_n$.
\item $I(\phi_k)=(1)$ or else $\codim I(\phi_{k}) \ge k$ for $k=1,\ldots,n$.
\end{itemize}
\label{prop:exact-sequence}
\end{prop}

\begin{rem}
For an exact code Hamiltonian,
we have a exact sequence $G \xrightarrow{\sigma} P \xrightarrow{\epsilon=\sigma^\dagger \lambda} E$.
As we will see in Lemma~\ref{lem:coker-epsilon-resolution-length-D},
$\coker \sigma$ has a finite free resolution,
and we may apply the Proposition~\ref{prop:exact-sequence}.
Since $\overline{ I_k(\sigma) } = I_k(\epsilon)$ for any $k \ge 0$, we have
\[
 2q = \rank P = \rank \sigma + \rank \epsilon = 2~ \rank \sigma.
\]
The size $2q \times t$ of the matrix $\sigma$ satisfies $t \ge q$.
If $I_q(\sigma) \ne R$, then $\codim I_q(\sigma) \ge 2$.
\label{rem:rank-sigma-m}
\end{rem}

\section{Ground state degeneracy}
\label{sec:degeneracy}

Let $H(L)$ be the Hamiltonians on finite systems
obtained by imposing periodic boundary conditions 
as in Section~\ref{sec:topological-order}.
A symmetry operator of $H(L)$ is 
a $\mathbb{C}$-linear combination of Pauli operator that commutes with $H(L)$.
In order for a Pauli symmetry operator
to have a nontrivial action on the ground space,
it must not be a product of terms in $H(L)$.
In addition, since $H(L)$ is a sum of Pauli operators,
a symmetry Pauli operator must commute with each term in $H(L)$.
Hence, a symmetry Pauli operator $O$ with nontrivial action on the ground space
must have image $v$ in the Pauli module such that
\[
 v(O) \in \ker \epsilon_L \setminus \im \sigma_L
\]
where
\[
 G / \bb_L G \xrightarrow{ \sigma_L } P / \bb_L P \xrightarrow{ \epsilon_L } E / \bb_L E 
\]
and
\[
 \bb_L = (x^L_1 -1,\ldots, x^L_D -1) \subseteq R, 
\]
which effectively imposes the periodic boundary conditions.
Since each term in $H(L)$ acts as an identity on the ground space,
if $O'$ is a term in $H(L)$, the symmetry operator $O$ and the product $OO'$
has the same action on the ground space. 
$OO'$ is expressed in the Pauli module as $v(O) + v'(O')$ for some $v' \in \im \sigma_L$.
Therefore, the set of Pauli operators of distinct actions on the ground space is
in one-to-one correspondence with the factor module
\[
 K(L) = \ker \epsilon_L ~/~ \im \sigma_L .
\] 
The vector space dimension $\dim_{\FF_2} K(L)$ is precisely
the number of independent Pauli operators that have nontrivial action on the ground space.
Since $\ker \epsilon_L = (\im \sigma_L)^\perp$ by definition of $\epsilon$,
and $\im \sigma_L$ as an $\FF_2$-vector space is a null space of the symplectic vector space $P/\bb_L P$,
it follows that $\ker \epsilon_L = \im \sigma_L \oplus W$ for some hyperbolic subspace $W$.
The quotient space $K(L) \cong W$ is thus hyperbolic and has even vector space dimension $2k$.
Choosing a symplectic basis for $K(L)$, 
it is clear that $K(L)$ represents the tensor product of $k$ qubit-algebras.
Therefore, the ground space degeneracy is exactly $2^k$~\cite{Gottesman1996Saturating,CalderbankRainsShorEtAl1997Quantum}.
In the theory of quantum error correcting codes,
$k$ is called the number of logical qubits,
and the elements of $K(L)$ are called the logical operators.
In this section, $k$ will always denote $\frac{1}{2} \dim_{\FF_2} K$.

\begin{defn}
The {\em associated ideal} for a code Hamiltonian
is the $q$-th determinantal ideal $I_q(\sigma) \subseteq R$
of the generating map $\sigma$.
Here, $q$ is the number of qubits per site.
The {\em characteristic dimension} is the Krull dimension $\dim R / I_q(\sigma)$.
\end{defn}

The associated ideals appears in Buchsbaum-Eisenbud theorem~(Proposition~\ref{prop:exact-sequence}),
which says that the homology $K(L)$ is intimately related to the associated ideal.
Imposing boundary conditions such as $x^L=1$
amounts to treating $x$ not as variables any more,
but as a `solution' of the equation $x^L-1=0$.
In order for $K(L)$ to be nonzero,
the `solution' $x$ should make the associated ideal to vanish.
Hence, by investigating the solutions of $I_q(\sigma)$
one can learn about the relation between the degeneracy and the boundary conditions.
Roughly, a large number of solutions of $I_q(\sigma)$ 
compatible with the boundary conditions
means a large degeneracy.
As $d = \dim R/I_q(\sigma)$ is the geometric dimension of the algebraic set defined by $I_q(\sigma)$,
a larger $d$ means a larger number of solutions.
Hence, the characteristic dimension $d$
controls the growth of the degeneracy as a function of the system size.

For example, consider a chain complex over $R = \FF[x^{\pm 1},y^{\pm 1}]$.
\[
0 \to
 R^1 
\xrightarrow{ \partial_2 = \begin{pmatrix} x-1 \\ y-1 \end{pmatrix} } 
 R^2
\xrightarrow{ \partial_1 = \begin{pmatrix} y-1 & -x +1 \end{pmatrix} }
 R^1
\]
It is exact at $R^2$.
The smallest nonzero determinantal ideal $I$ for either $\partial_1$ or $\partial_2$ is $I=(x-1,y-1)$.
If we impose `boundary conditions' such that $x=1$ and $y=1$,
then $I$ becomes zero, and according to Buchsbaum-Eisenbud theorem,
the homology $K$ at $R^2$ should be nontrivial.
Since the solution of $I$ consists of a single point $(1,1)$ on a 2-plane,
it is conceivable that `boundary conditions' of form $\bb_L$
would always give $K(L)$ of a constant $\FF$-dimension,
which is true in this case.
If we insist that the complex is over $R' = \FF[x^{\pm 1},y^{\pm 1},z^{\pm 1}]$,
then the zero set of $I$ is a line $(1,1,z)$ in 3-space;
there are many `solutions.'
In this case, $K^{R'}(L)$ has $\FF$-dimension $2L$.

An obvious example where the homology $K$ is always zero regardless of the boundary conditions
is this:
\[
0 \to
 R^1 
\xrightarrow{ \begin{pmatrix} 1 \\ 0 \end{pmatrix} } 
 R^2
\xrightarrow{ \begin{pmatrix} 0 & 1 \end{pmatrix} }
 R^1
\]
Here, the determinantal ideal is $(1)=R$, and thus has no solution.

The intuition from these examples are made rigorous below.

\subsection{Condition for degenerate Hamiltonians}

A routine yet very important tool is \emph{localization}.
The origin of all difficulties in dealing with general rings
is that nonzero elements do not always have multiplicative inverse;
one cannot easily solve linear equations.
The localization is a powerful technique to get around this problem.
As we build rational numbers from integers by \emph{declaring} that
nonzero numbers have multiplicative inverse,
the localization enlarges a given ring
and \emph{formally allows} certain elements to be invertible.
It is necessary and sometimes desirable not to invert all nonzero elements,
in order for the localization to be useful.
For a consistent definition, we need a multiplicatively closed subset $S$
containing 1, but not containing 0,
of a ring $R$ and declare that the elements of $S$ is invertible.
The new ring is written as $S^{-1}R$,
in which a usual formula $\frac{r_1}{s_1} + \frac{r_2}{s_2} = \frac{r_1 s_2 + r_2 s_1}{s_1 s_2}$ holds.
The original ring naturally maps into $S^{-1}R$ as $\phi : r \mapsto \frac{r}{1}$.
The localization means that one views all data as defined over $S^{-1}R$ via the natural map $\phi$.%
\footnote{It is a functor from the category of $R$-modules to that of $S^{-1}R$-modules.}

A localized ring, by definition, has more invertible elements, and hence has less nontrivial ideals.
In fact, our Laurent polynomial ring is a localized ring of the polynomial ring by inverting monomials,
e.g., $\{ x^i y^j | i,j \ge 0 \}$.
Nontrivial ideals such as $(x)$ or $(x,y)$ in the polynomial ring
become the unit ideal $(1)$ in the Laurent polynomial ring.
Further localizations in this paper are with respect to prime ideals.
In this case, we say the ring is \emph{localized at a prime ideal $\pp$}.
A prime ideal $\pp$ has a defining property that
$a b \notin \pp$ whenever $a \notin \pp$ and $b \notin \pp$.
Thus, the set-theoretic complement of $\pp$ is a multiplicatively closed set containing 1.
In $(R\setminus \pp)^{-1}R$, denoted by $R_\pp$, any element outside $\pp$ is invertible,
and therefore $\pp$ becomes a unique maximal ideal of $R_\pp$.
Moreover, the localization sometimes simplifies the generators of an ideal.
For instance, if $R=\FF[x,x^{-1}]$ and $\pp = (x-1)$,
the ideal $((x -1)(x^5-x+1)) \subseteq R$ localizes to $(x-1)_\pp \subseteq R_\pp$
since $x^5-x+1$ is an invertible element of $R_\pp$.

An important fact about the localization is that a module is zero
if and only if its localization at every prime ideal is zero.
Further, the localization preserves exact sequences.
So we can analyze a complex by localizing at various prime ideals.
For a thorough treatment about localizations,
see Chapter 3 of \cite{AtiyahMacDonald}.
The term `localization' is from geometric considerations
where a ring is viewed as a function space on a geometric space.

\begin{lem}
Let $I$ be the associated ideal of an exact code Hamiltonian,
and $\mm$ be a prime ideal of $R$.
Then, $I \not\subseteq \mm$ implies that the localized homology
\[
K(L)_\mm = \ker (\epsilon_L)_\mm ~/~ \im (\sigma_L)_\mm
\]
is zero for all $L \ge 1$.
\label{lem:associated-maximal-localized-homology}
\end{lem}
It is a simple variant of a well-known fact that
a module over a local ring is free if its first non-vanishing Fitting ideal is the unit ideal~\cite[Chapter~1 Theorem~12]{Northcott}.
\begin{proof}
Recall that the localization and the factoring commute.
By assumption, $(I_q(\epsilon))_\mm = \overline{ (I_q(\sigma))_\mm } = (1) = R_\mm =: S$.
Recall that the local ring $S$ has the unique maximal ideal $\mm$,
and any element outside the maximal ideal is a unit.
If every entry of $\epsilon$ is in $\mm$, then $I_q(\epsilon) \subseteq \mm \ne S$.
Therefore, there is a unit entry, and by column and row operations,
$\epsilon$ is brought to
\[
 \epsilon \cong 
\begin{pmatrix}
 1 & 0 \\
 0 & \epsilon'
\end{pmatrix}
\]
where $\epsilon'$ is a submatrix.
It is clear that $I_{q-1}(\epsilon') \subseteq I_q(\epsilon)$
since any $q-1 \times q-1$ submatrix of $\epsilon'$ can be thought of
as a $q \times q$ submatrix of $\epsilon$ where the first column and first row
have the unique nonzero entry 1 at $(1,1)$.
It is also clear that $I_{q-1}(\epsilon') \supseteq I_q(\epsilon)$
since any $q \times q$ submatrix of $\epsilon$ contains
either zero row or column, or the $(1,1)$ entry $1$ of $\epsilon$.
Hence, $I_{q-1}(\epsilon') = (1)$, and
we can keep extracting unit elements
into the diagonal by row and column operations~\cite[Chapter~1 Theorem~12]{Northcott}.
After $q$ steps,
$t \times 2q$ matrix $\epsilon$ becomes precisely
\[
\epsilon \cong 
\begin{pmatrix}
 \id_q & 0 \\
 0 & 0
\end{pmatrix}
\]
where $\id_q$ is the $q \times q$ identity matrix.
Since localization preserves the exact sequence $G \to P \to E$,
$\sigma$ maps to the lower $q$ components of $P$ with respect to the basis
where $\epsilon$ is in the above form.
Since $I_q(\sigma) = (1)$, we must have (after basis change)
\[
\sigma \cong 
\begin{pmatrix}
 0 & 0 \\
 \id_q & 0
\end{pmatrix}.
\]
Therefore, even after factoring by the proper ideal $\bb_L$, 
the homology $K(L) = \ker \epsilon_L ~/~ \im \sigma_L$ is still zero.
\end{proof}
\begin{cor}
The associated ideal of an exact code Hamiltonian
is the unit ideal, i.e., $I_q(\sigma) = R$,
if and only if
\[
K(L) = \ker \epsilon_L ~/~ \im \sigma_L = 0
\]
for all $L \ge 1$.
\label{cor:unit-characteristic-ideal-means-nondegeneracy}
\end{cor}
\begin{proof}
If $I(\sigma) = R$,
$I(\sigma)$ is not contained in any prime ideal $\mm$.
The above lemma says $K(L)_\mm = 0$.
Since a module is zero 
if and only if its localization at every prime ideal is zero,
$K(L) = 0$ for all $L \ge 1$.

For the converse, observe that
if $\FF$ is any extension field of $\FF_2$,
for any $\FF_2$-vector space $W$,
we have $\dim_\FF \FF \otimes_{\FF_2} W = \dim_{\FF_2} W$.
We replace the ground field $\FF_2$ with its algebraic closure $\FF^a$ to test whether $K(L) \ne 0$.
If $I_q(\sigma)$ is not the unit ideal, then it is contained in a maximal ideal $\mm \subsetneq R$.
By Nullstellensatz, $\mm = (x_1 - a_1,\ldots,x_D - a_D)$ for some $a_i \in \FF^a$.
Since in $R$ any monomial is a unit, we have $a_i \ne 0$.
Therefore, there exists $L \ge 1$ such that $a_i^L = 1$ and $2 \nmid L$.
The equation $x^L-1=0$ has no multiple root.

We claim that $K(L) \ne 0$. It is enough to verify this for the localization at $\mm$.
Since anything outside $\mm$ is a unit in $R_\mm$
and each $x_i^L-1$ contains exactly one $x_i - a_i$ factor,
we see $(\bb_L)_\mm = \mm_\mm$.
Therefore, $(\epsilon_L)_\mm = \epsilon_\mm / (\bb_L)_\mm$ and
$(\sigma_L)_\mm = \sigma_\mm / (\bb_L)_\mm$ is a matrix over the field $R/\mm = \FF^a$.
Since $I_q(\sigma) \subseteq \mm$, we have $I_q(\sigma_L)_\mm = 0$.
That is, $\rank_{\FF^a} (\sigma_L)_\mm < q$.
It is clear that $\dim_{\FF^a} K(L)_\mm = \dim_{\FF^a} \ker (\epsilon_L)_\mm / \im (\sigma_L)_\mm \ge 1$.
\end{proof}
This corollary says that in order to have a {\em degenerate} Hamiltonian $H(L)$,
one must have a proper associated ideal.
We shall simply speak of a {\em degenerate} code Hamiltonian
if its associated ideal is proper.

\subsection{Counting points in an algebraic set}
It is important that the factor ring
\[
 R/\bb_L = \FF_2 [x_1, \ldots, x_D]~/~(x^L_1 -1,\ldots,x^L_D -1)
\]
is finite dimensional as a vector space over $\FF_2$,
and hence is Artinian. In fact, $\dim_{\FF_2} R/\bb_L = L^D$.
This ring appears also in \cite{GueneriOezbudak2008}.
Due to the following structure theorem of Artinian rings,
$K(L)$ can be explicitly analyzed by the localizations.
\begin{prop}
\cite[Chapter~8]{AtiyahMacDonald}\cite[Section~2.4]{Eisenbud}
Let $S$ be an Artinian ring.
(For example, $S$ is a homomorphic image of a polynomial ring 
over finitely many variables with coefficients in a field $\FF$,
and is finite dimensional as a vector space over $\FF$.)
Then, there are only finitely many maximal ideals of $S$, and
\[
 S \cong \bigoplus_\mm S_\mm
\]
where the sum is over all maximal ideals $\mm$ of $S$ and $S_\mm$ is the localization of $S$ at $\mm$. 
\label{prop:Artin-ring}
\end{prop}
The following calculation tool is sometimes useful.
Recall that a group algebra is equipped with a non-degenerate 
scalar product $\langle v,w \rangle = \tr (v \bar w)$.
This scalar product naturally extends to a direct sum of group algebras.
\begin{lem}
Let $\FF$ be a field,
and $S = \FF[\Lambda]$ be the group algebra of a finite abelian group $\Lambda$.
If $N$ is a submodule of $S^n$, then the dual vector space $N^*$
is vector-space isomorphic to $S^n / N^\perp$,
where $\perp$ is with respect to the scalar product $\langle \cdot , \cdot \rangle$.
\end{lem}
\begin{proof}
Consider $\phi : S^n \ni x \mapsto \langle \cdot, x \rangle \in N^*$.
The map $\phi$ is surjective since the scalar product is non-degenerate
and $S^n$ is a finite dimensional vector space.
The kernel of $\phi$ is precisely $N^\perp$.
\end{proof}
\begin{cor}
Put $2k = \dim_{\FF_2} K(L)$. Then,
\[
 k = qL^D - \dim_{\FF_2} \im \sigma_L = \dim_{\FF_2} \ker \epsilon_L - qL^D.
\]
Further, if $q=t$, then
\[
 k = \dim_{\FF_2} \coker \epsilon_L .
\]
\label{cor:k-formulas}
\end{cor}
\noindent
The first formula is a rephrasing of the fact that the number of encoded qubits
is the total number of qubits minus the number of independent stabilizer 
generators~\cite{Gottesman1996Saturating,CalderbankRainsShorEtAl1997Quantum}.
\begin{proof}
Put $S = R/\bb_L$. 
If $v_1,\ldots, v_t$ denote the columns of $\sigma_L$, we have
\begin{equation}
\ker \sigma_L^\dagger = \lambda_q \ker \epsilon_L 
= \bigcap_i v_i^\perp = \left( \sum_i S v_i \right)^\perp = \left( \im \sigma_L \right)^\perp.
\label{eq:orthogonal-to-dual}
\end{equation}
Hence, $\dim_{\FF_2} \ker \epsilon_L = \dim_{\FF_2} S^{2q} - \dim_{\FF_2} \im \sigma_L.$
Since $\dim_{\FF_2} S = L^D$ and $K(L) = \ker \epsilon_L / \im \sigma_L$,
the first claim follows.

Since $\im \sigma_L \cong S^t / \ker \sigma_L$, if $t=q$,
we have $k = \dim_{\FF_2} \ker \sigma_L$ by the first claim.
From Eq.~\eqref{eq:orthogonal-to-dual}, we conclude that 
$k = \dim_{\FF_2} S^t / \im \sigma_L^\dagger = \dim_{\FF_2} \coker \epsilon_L$.
\end{proof}
\noindent We will apply these formulas in Section~\ref{sec:eg}.

The characteristic dimension is related to the rate at which
the degeneracy increases as the system size increases
in the following sense.
Recall that $2k = \dim_{\FF_2} K(L)$ and the ground state degeneracy is $2^k$.
\begin{lem}
Suppose $2 \nmid L$. Let $\FF^a$ be the algebraic closure of $\FF_2$.
If $N$ is the number of maximal ideals in $\FF^a \otimes_{\FF_2} R$ 
that contains $\bb_L + I_q(\sigma)$, 
then
\[
N \le \dim_{\FF_2} K(L) \le 2q N.
\]
\label{lem:K-L-number-points}
\end{lem}
\begin{proof}
We replace the ground field $\FF_2$ with $\FF^a$.
Any maximal ideal of an Artinian ring $\FF^a[x_i^{\pm 1}]/\bb_L$ is of form
$\mm = (x_1 - a_1, \ldots, x_D - a_D)$ where $a_i^L = 1$ by Nullstellensatz.
Since $2 \nmid L$, we see that $(\bb_L)_\mm = \mm_\mm$ 
and that $(R/\bb_L)_\mm \cong \FF^a$ is the ground field.
(See the proof of Corollary~\ref{cor:unit-characteristic-ideal-means-nondegeneracy}.)

Now, $I_q(\sigma) + \bb_L \subseteq \mm$
iff $I_q(\sigma)_\mm + (\bb_L)_\mm \subseteq \mm_\mm = (\bb_L)_\mm$
iff $I_q(\sigma)$ becomes zero over $R_\mm / (\bb_L)_\mm \cong \FF^a$
iff $1 \le \dim_{\FF^a} K(L)_\mm \le 2q$.
Since by Proposition~\ref{prop:Artin-ring},
$\dim_{\FF^a} K(L)$ is a finite direct sum of localized ones,
we are done.
\end{proof}

\begin{lem}
Let $I$ be an ideal such that $\dim R/I = d$. We have
\[
\dim_{\FF_2} R / (I+\bb_L) \le cL^d
\]
for all $L \ge 1$ and some constant $c$ independent of $L$.
\end{lem}
\begin{proof}
We replace the ground field with its algebraic closure $\FF^a$.
Write $\tilde x_i$ for the image of $x_i$ in $R/I$.
By Noether normalization theorem~\cite[Theorem~13.3]{Eisenbud},
there exist $y_1,\ldots, y_d \in R/I$ such that
$R/I$ is a finitely generated module over $\FF^a[y_1,\ldots,y_d]$.
Moreover, one can choose $y_i = \sum_{j=1}^D M_{ij} \tilde x_j$
for some rank $d$ matrix $M$ whose entries are in $\FF^a$.
Making $M$ into the reduced row echelon form,
we may assume $y_i = \tilde x_i + \sum_{j>d} a_{ij} \tilde x_j$
for each $1 \le i \le d$.

Let $S=\FF^a[z_1,\ldots,z_D]$ be a polynomial ring in $D$ variables.
Let $\phi : S \to R/( I + \bb_L )$ be the ring homomorphism
such that $z_i \mapsto y_i$ for $1 \le i \le d$
and $z_j \mapsto \tilde x_j$ for $ d < j \le D$.
By the choice of $y_i$, $\phi$ is clearly surjective.
Consider the ideal $J$ of $S$ generated by the initial terms of $\ker \phi$ 
with respect to the lexicographical monomial order in which $z_1 \prec \cdots \prec z_D$.
Since $\tilde x_j$ is integral over $\FF[y_1,\ldots,y_d]$,
the monomial ideal $J$ contains $z_j^{n_j}$ for some positive $n_j$ for $d < j \le D$.
Here, $n_j$ is independent of $L$.
Since $z_i^L \in J$ for $1 \le i \le d$, we conclude that
\[
 \dim_{\FF^a} R/(I+\bb_L) = \dim_{\FF^a} S/J \le L^d \cdot n_{d+1} n_{d+2} \cdots n_{D}
\]
by Macaulay theorem \cite[Theorem~15.3]{Eisenbud}.
\end{proof}

\begin{cor}
If $2 \nmid L$, and $d = \dim R/I_q(\sigma)$ is the characteristic dimension of a code Hamiltonian,
then 
\[
 \dim_{\FF_2} K(L) \le c L^d
\]
for some constant $c$ independent of $L$.
\end{cor}
\begin{proof}
If $J  = \bb_L + I(\sigma)$,
$N$ in Lemma~\ref{lem:K-L-number-points} is equal to $\dim_{\FF^a} \FF^a \otimes R /\mathop{\mathrm{rad}} J$.
This is at most $\dim_{\FF^a} \FF^a \otimes R / J = \dim_{\FF_2} R/J$.
\end{proof}

\begin{lem}
Let $d$ be the characteristic dimension.
There exists an infinite set of integers $\{ L_i \}$ such that
\[
 \dim_{\FF_2} K(L_i) \ge {L_i}^d /2
\]
\label{lem:k-growing-sequence-L}
\end{lem}
\begin{proof}
We replace the ground field with its algebraic closure $\FF^a$.
Let $\pp' \supseteq I(\sigma)$ be a prime of $R$ of codimension $D-d$.
Let $\pp$ be the contraction (pull-back) of $\pp'$ in the polynomial ring
$S = \FF^a[x_1,\ldots,x_D]$. Since the set of all primes of $R$ is in one-to-one
correspondence with the set of primes in $S$ that does not include monomials,
it follows that $\pp$ has codimension $D-d$ and does not contain any monomials.
Let $V$ denote the affine variety defined by $\pp=(g_1,\ldots,g_n)$.
Since $\pp$ contains no monomials, $V$ is not contained in any hyperplanes $x_i = 0$
($i = 1,\ldots, D$).

Let $A_1$ be a finite subfield of $\FF^a$ that contains all the coefficients of $g_i$,
so $V$ can be defined over $A_1$.
Let $A_n \subseteq \FF^a$ be the finite extension fields of $A_1$ of extension degree $n$.
Put $L_n = |A_n|-1$.
For any subfield $A$ of $\FF^a$, let us say a point of $V$ is \emph{rational} over $A$
if its coordinates are in $A$.
The number $N'(L_n)$ of points $(a_i) \in V$ satisfying $a_i^{L_n} = 1$
is precisely the number of the rational points of $V$ over $A_n$
that are not contained in the hyperplanes $x_i=0$.
Since $I(\sigma) \subseteq \pp'$,
the number $N$ in Lemma~\ref{lem:K-L-number-points} is at least $N'(L_n)$.
It remains to show $N'(L_n) \ge L_n^d /2$ for all sufficiently large $n$.

This follows from the result by Lang and Weil~\cite{LangWeil1954},
which states that the number of points of a projective variety
of dimension $d$ that are rational over a finite field of $m$ elements
is $m^d + O\left(m^{d-\frac{1}{2}} \right)$ asymptotically in $m$.
Since Lang-Weil theorem is for projective variety
and we are with an affine variety $V$,
we need to subtract the number of points in the hyperplanes $x_i = 0$
($i = 0,1,\ldots,D$)
from the Zariski closure of $V$.
The subvarieties in the hyperplanes, being closed,
have strictly smaller dimensions, and we are done.
\end{proof}

\section{Fractal operators and topological charges}

This section is to provide a characterization of topological charges,
and their dynamical properties.
Before we turn to a general characterization and define fractal operators,
let us review familiar examples.
Note that for two dimensions the base ring is $R = \FF_2[x,\bar x, y, \bar y]$.

\begin{example}[Toric Code]
\label{eg:2d-toric}
Although the original two-dimensional toric code has qubits on edges~\cite{Kitaev2003Fault-tolerant},
we put two qubits per site of the square lattice to fit it into our setting.
Concretely, the first qubit to each site represents the one on its east edge,
and the second qubit the one on its north edge. With this convention,
the Hamiltonian is the negative sum of the following two types of interactions:
\[
\xymatrix@!0{
XI \ar@{-}[r] & XX \ar@{-}[d] \\
II \ar@{-}[u] & IX \ar@{-}[l]
} \quad
\xymatrix@!0{
ZI \ar@{-}[r] & II \ar@{-}[d] \\
ZZ \ar@{-}[u] & IZ \ar@{-}[l]
} \quad \quad
\xymatrix@!0{
y \ar@{-}[r] & xy \ar@{-}[d] \\
1 \ar@{-}[u] & x \ar@{-}[l]
}
\]
where we used $X,Z$ to abbreviate $\sigma_x,\sigma_z$, and omitted the tensor product symbol.
Here, the third square specifies the coordinate system of the square lattice.
Since there are $q=2$ qubits per site, the Pauli module is of rank 4.
The corresponding generating map $\sigma : R^2 \to R^4$ is given by the matrix
\[
 \sigma_{\text{2D-toric}} = 
\begin{pmatrix}
 y+xy & 0 \\
 x+xy & 0 \\
\hline
 0    & 1 + y \\
 0    & 1 + x
\end{pmatrix}
\cong
\begin{pmatrix}
 1 + \bar x & 0 \\
 1 + \bar y & 0 \\
\hline
 0          & 1 + y \\
 0          & 1 + x
\end{pmatrix} .
\]
Here, the each column expresses each type of interaction.
It is clear that
\[
\epsilon_{\text{2D-toric}} = \sigma^\dagger \lambda_2 =
\begin{pmatrix}
 0         & 0         & 1 + x & 1 + y \\
 1+ \bar y & 1+ \bar x & 0     & 0
\end{pmatrix}
\]
and $\ker \epsilon = \im \sigma$;
the two dimensional toric code satisfies our exactness condition.
The associated ideal is $I(\sigma) = ( (1+x)^2, (1+x)(1+y), (1+y)^2 )$.
The characteristic dimension is $\dim R / I(\sigma) = 0$.
Note also that $\ann \coker \epsilon = (x-1,y-1)$.
The electric and magnetic charge are represented by
$\begin{pmatrix} 1 \\ 0 \end{pmatrix}, \begin{pmatrix} 0 \\ 1 \end{pmatrix} \in E \setminus \im \epsilon$,
respectively.

The connection with cellular homology should be mentioned.
$\sigma$ can be viewed as the boundary map
from the free module of all 2-cells with $\mathbb{Z}_2$ coefficients
of the cell structure of 2-torus
induced from the tessellation by the square lattice.
Then, $\epsilon$ is interpreted as the boundary map from the free module of all 1-cells
to that of all 0-cells.
$\sigma$ or $\epsilon$ is actually the direct sum of two boundary maps.
Indeed, the space 
$K(L) = \ker \epsilon_L / \im \sigma_L$
of operators acting on the ground space (logical operators)
has four generators
\begin{align*}
 l_y(X) = \begin{pmatrix} 1+y+\cdots+y^{L-1} \\ 0 \\ 0 \\ 0 \end{pmatrix}, & & 
 l_x(X) = \begin{pmatrix} 0 \\ 1+ x+ \cdots + x^{L-1} \\ 0 \\ 0 \end{pmatrix}, \\
 l_x(Z) = \begin{pmatrix} 0 \\ 0 \\ 1+x+\cdots+x^{L-1} \\ 0 \end{pmatrix}, & &
 l_y(Z) = \begin{pmatrix} 0 \\ 0 \\ 0 \\ 1+y+\cdots+y^{L-1} \end{pmatrix},
\end{align*}
which correspond to the usual nontrivial first homology classes of 2-torus.

The description by the cellular homology might be advantageous for the toric code
over our description with pure Laurent polynomials;
in this way, it is clear that the toric code can be defined on
an arbitrary tessellation of compact orientable surfaces.
However, it is unclear whether this cellular homology description
is possible after all for other topologically ordered code Hamiltonians.
\hfill $\Diamond$
\end{example}

\begin{example}[2D Ising model on square lattice]
The Ising model has nearest neighbor interactions that are horizontal and vertical.
In our formalism, they are represented as $1+x$ and $1+y$. Thus,
\[
 \sigma_\text{2D Ising} =
\begin{pmatrix}
 0   & 0 \\
 1+x & 1+y
\end{pmatrix} .
\]
As it is not topologically ordered, the complex $G \to P \to E$ is not exact.
Moreover, $\sigma$ is not injective.
\[
 \sigma_\text{2D Ising;1} =
\begin{pmatrix}
 1+y \\
 1+x
\end{pmatrix}
\]
generates the kernel of $\sigma$. That is, the complex 
$0 \to G_1 \xrightarrow{ \sigma_\text{2D Ising;1}} G \xrightarrow{ \sigma_\text{2D Ising} } P$ is exact.
\hfill $\Diamond$
\end{example}

In both examples, there exist isolated excitations.
In the toric code, the isolated excitation can be
(topologically) nontrivial since the electric charge is not in $\im \epsilon$.
On the contrary, in 2D Ising model, any isolated excitation is 
actually created by an operator of finite support
because any excitation created by some Pauli operator appears as several connected loops.
This difference motivates the following definition for charges.

Let $\tilde R$ be the set of all $\FF_2$-valued functions on the translation group $\Lambda$,
not necessarily finitely supported. For instance, if $\Lambda = \ZZ$,
\[
\tilde f = \cdots + x^{-4} + x^{-2} + 1 + x^2 + x^4 + \cdots \in \tilde R
\]
represents a function whose value is 1 at even lattice points, and 0 at odd points.
Note that $\tilde R$ is a $R$-module,
since the multiplication is a convolution between an arbitrary function and a finitely supported function.
For example,
\begin{align*}
 (1+x) \cdot \tilde f &= \cdots + x^{-2} + x^{-1} + 1 + x + x^2 + \cdots , \\
 (1+x)^2 \cdot \tilde f &= 0 .
\end{align*}
Let $\tilde P = \tilde R^{2q}$ be the module of Pauli operators of possibly infinite support.
Similarly, let $\tilde E$ be the module of virtual excitations of possibly infinitely many terms.
Formally, $\tilde P$ is the module of all $2q$-tuples of functions on the translation-group,
and $\tilde E$ is that of all $t$-tuples.
Clearly, $P \subseteq \tilde P$ and $E \subseteq \tilde E$.
The containment is strict if and only if the translation-group is infinite.
Since the matrix $\epsilon$ consists of Laurent polynomials with finitely many terms,
$\epsilon:P \to E$ extends to a map from $\tilde P$ to $\tilde E$.
\begin{defn}
A \emph{(topological) charge} $e = \epsilon( \tilde p ) \in E$ 
is an excitation of finite energy (an element of the virtual excitation module)
created by a Pauli operator $\tilde p \in \tilde P$ of possibly infinite support.
A charge $e$ is called \emph{trivial} if $e \in \epsilon(P)$.
\end{defn}
\noindent By definition, the set of all charges modulo trivial ones is
in one-to-one correspondence with the superselection sectors.
According to the definition, any charge of 2D Ising model is trivial.
A nontrivial charge may appear due to the following fractal generators.
\begin{defn}
We call zero-divisors on $\coker \epsilon$ as \emph{fractal generators}.
In other words, 
an element $f \in R \setminus \{0\}$ is a {\em fractal generator} if
there exists $v \in E \setminus \im \epsilon$ 
such that $f v \in \im \epsilon$.
\end{defn}

There is a natural reason the fractal generator deserves its name.
Consider a code Hamiltonian with a single type of interaction: $t=1$.
So each configuration of excitations is described by one Laurent polynomial.
For example, in two dimensions, $f = 1 + x + y = \epsilon (p)$ represents
three excitations, one at the origin of the lattice and the others at $(1,0)$ and $(0,1)$
created by a Pauli operator represented by $p$.
(This example is adopted from \cite{NewmanMoore1999Glassy}.)
In order to avoid repeating phrase,
let us call each element of the Pauli module a Pauli operator,
and instead of using multiplicative notation
we use module operation $+$ to mean the product of the corresponding Pauli operators.

Consider the Pauli operator $fp = p + x p + y p \in P$. 
It describes the Pauli operator $p$ at the origin 
multiplied by the translations of $p$ at $(1,0)$ and at $(0,1)$.
So $fp$ consists of three copies of $p$.
This Pauli operator maps the ground state to the excited state $f^2 = 1 + x^2 + y^2$.
The number of excitations is still three, but the excitations at $(1,0), (0,1)$ 
have been replaced by those at $(2,0),(0,2)$.
Similarly, the Pauli operator $f^{2+1} p = f^2 (fp)$ 
consists of three copies of $fp$, or $3^2$ copies of $p$.
The excited state created by $f^3 p$ is $f^4 = (f^2)^2 = 1 + x^{2^2} + y^{2^2}$.
Still it has three excitations, but they are further apart.
The Pauli operator $f^{2^n -1} p$ consists of $3^n$ copies of $p$ in a self-similar way,
and the excited state caused by $f^{2^n -1} p$ consists of a constant number of excitations.
More generally, if there are $t > 1$ types of terms in the Hamiltonian,
the excitations are described by a $t \times 1$ matrix.
If it happens to be of form $f v$ for some $f \in R$ consisted of two or more terms,
there is a family of Pauli operators $f^{2^n -1} p$ with self-similar support
such that it only creates a bounded number of excitations.
An obvious but uninteresting way to have such a situation is to put
$f v = \epsilon (f p')$ for a Pauli operator $p'$ where $v = \epsilon( p' )$.
Our definition avoids this triviality by requiring $v \notin \im \epsilon$.
The reader may wish to compare the fractals with finite cellular automata~\cite{MartinOdlyzkoAndrewWolfram1984}.

\begin{prop}
\cite[16.33]{BrunsVetter}
Suppose $\coker \epsilon \ne 0$.
Then, the following are equivalent:
\begin{itemize}
 \item There does not exist a fractal generator.
 \item $\coker \epsilon$ is torsion-free.
 \item There exists a free $R$-module $E'$ of finite rank such that
\[
 P \xrightarrow{\epsilon} E \to E'
\]
is exact.
\label{prop:nofractal-torsionless}
\end{itemize}
\end{prop}
\begin{proof}
The first two are equivalent by definition.
The sequence above is exact if and only if
$ 0 \to \coker{\epsilon} \to E' $
is exact.
Since $\coker \epsilon$ has a finite free resolution,
the second is equivalent to the third.
\end{proof}
The following theorem states that the fractal operators produces all nontrivial charges.
\begin{theorem}
Suppose $\Lambda = \mathbb{Z}^D$ is the translation-group of the underlying lattice.
The set of all charges modulo trivial ones is in one-to-one correspondence
with the torsion submodule of $\coker \epsilon$.
\label{thm:charge-equals-torsion}
\end{theorem}

To illustrate the idea of the proof,
consider a (classical) excitation map%
\footnote{
It is classical because it is not derived from an interesting quantum commuting Pauli Hamiltonian.
For a classical Hamiltonian where all terms are tensor products of $\sigma_z$,
there is no need to keep a $t \times 2q$ matrix $\epsilon$
since the right half $\epsilon$ is zero.
Just the left half suffices, which can be arbitrary
since the commutativity equation $\epsilon \lambda \epsilon^\dagger =0$ is automatic.
Nevertheless, the excitations and fractal operators are relevant.
Our proof of the theorem is not contingent on the commutativity equation.
}
\[
 \phi =
 \begin{pmatrix}
  1+x+y & 0   \\
     0  & 1+x \\
     0  & 1+y \\
 \end{pmatrix}
 : R^2 \to R^3 .
\]
A nonzero element $f=1+x+y \in R$ is a fractal generator
since $\begin{pmatrix} 1 & 0 & 0 \end{pmatrix}^T \notin \im \phi$ 
and $(1+x+y)\begin{pmatrix} 1 & 0 & 0 \end{pmatrix}^T \in \im \phi$;
$f$ is a zero-divisor on a torsion element $\begin{pmatrix} 1 & 0 & 0 \end{pmatrix}^T \in \coker \epsilon$.
It is indeed a charge since 
$\phi(\tilde f \begin{pmatrix} 1 & 0 \end{pmatrix}^T) = \begin{pmatrix} 1 & 0 & 0 \end{pmatrix}^T$ where
\[
 \tilde f = \lim_{n \to \infty} f^{2^n-1} \in \FF[[x,y]]
\]
is a formal power series, which can be viewed as an element of $\tilde R$.
The limit is well-defined since $f^{2^{n+1}-1} - f^{2^n -1}$ only contains terms of degree $2^n$ or higher.
That is to say, only higer order `corrections' are added and lower order terms are not affected.
Of course, there is no natural notion of smallness in the ring $\FF[x,y]$.
But one can formally call the members of the ideal power $(x,y)^n \subseteq \FF[x,y]$ \emph{small}.
It is legitimate to introduce a topology in $R$ defined by the \emph{ever shrinking} ideal powers $(x,y)^n$.
They play a role analogous to the ball of radius $1/n$ in a metric topological space.
The \emph{completion} of $\FF[x,y]$
where every Cauchy sequence with respect to this topology is promoted to a convergent sequence,
is nothing but the formal power series ring $\FF[[x,y]]$.
For a detailed treatment, see Chapter 10 of \cite{AtiyahMacDonald}.

The completion and the limit only make sense in the polynomial ring $\FF[x,y]$.
The reason $\tilde f$ is well-defined is that $f \in \FF[x^{\pm 1},y^{\pm 1}]$ is accidentally expressed
as a usual polynomial with lowest order term $1$.
In the proof below we show that every fractal generator can be expressed in this way.
Hence, a torsion element of $\coker \epsilon$ is really a charge.

\begin{proof}
For a module $M$, let $T(M)$ denote the torsion submodule of $M$:
\[
 T(M) = \{ m \in M \ |\  \exists \, r \in R \setminus \{0\} \text{ such that } rm = 0 \}
\]

Suppose first that $T(\coker \epsilon) = 0$.
We claim that in this case there is no nontrivial charge.
Let $e = \epsilon( \tilde p ) \in E$ be a charge, where $\tilde p \in \tilde P$.
By Proposition~\ref{prop:nofractal-torsionless}
we have an exact sequence of finitely generated free modules
$ P \xrightarrow{\epsilon} E \xrightarrow{\epsilon_1} E_1$.
Since the matrix $\epsilon_1$ is over $R$,
the complex extends to a complex of modules of tuples of functions on the translation-group.
\[
 \tilde P \xrightarrow{\epsilon} \tilde E \xrightarrow{\epsilon_1} \tilde E_1
\]
(This extended sequence may not be exact.)
Then, $\epsilon_1(e) = \epsilon_1( \epsilon( \tilde p ) ) = 0$
since $\epsilon_1 \circ \epsilon = 0$ identically.
But, $e \in E$, and therefore, $e \in \ker \epsilon_1 \cap E = \epsilon(P)$.
It means that $e$ is a trivial charge, i.e., $e$ maps to zero in $\coker \epsilon$,
and proves the claim.%
\footnote{
One may wish to consider $\epsilon$ to consist of the second column of $\phi$ above.
Then $\epsilon_1 = \begin{pmatrix} 1+y & -1-x \end{pmatrix}$.
}

Now, allow $\coker (P \xrightarrow{\epsilon} E)$ to contain torsion elements.
$Q = (\coker \epsilon) / T(\coker \epsilon)$ is torsion-free,
and is finitely presented as $Q = \coker( \epsilon' : P' \to E )$
where $P'$ is a finitely generated free module.
In fact, we may choose $\epsilon'$ by adding more columns representing the generators of
the torsion submodule of $\coker \epsilon$ to the matrix $\epsilon$.
\[
 \epsilon = \begin{pmatrix} \# & \# \\ \# & \# \end{pmatrix} \quad \quad
 \epsilon'= \begin{pmatrix} \# & \# & * & * \\ \# & \# & * & * \end{pmatrix}
\]
Then, $P$ can be regarded as a direct summand of $P'$.%
\footnote{
If we take $\epsilon = \phi$ above,
then
\[
 \epsilon' =  \begin{pmatrix}
  1+x+y & 0   & 1 \\
     0  & 1+x & 0 \\
     0  & 1+y & 0 \\
 \end{pmatrix}.
\]
Note that $P' = P \oplus R$.
}

Let $e = \epsilon( \tilde p ) \in E$ be any charge.
Since the matrix $\epsilon'$ contains $\epsilon$ as submatrix,
we may write $e = \epsilon'(\tilde p) \in E$.
Since $T(\coker \epsilon')=0$,
we see by the first part of the proof that $e = \epsilon'( p' )$ for some $p' \in P'$.
Then, $e$ maps to zero in $Q$, and it follows that $e$ maps into $T(\coker \epsilon)$
in $\coker \epsilon$. 
In other words, the equivalence class of $e$ modulo trivial charges 
is a torsion element of $\coker \epsilon$.

Conversely,
we have to prove that for every element $e \in E$ 
such that $fe = \epsilon(p)$ for some $f \in R \setminus \{0\}$ and $p \in P$,
there exists $\tilde p \in \tilde P$ such that $ e = \epsilon( \tilde p )$.
Here, $\tilde P$ is the module of all $2q$-tuples of 
$\FF_2$-valued functions on the translation-group.
Consider the lexicographic total order on $\mathbb{Z}^D$ in which $x_1 \succ x_2 \succ \cdots \succ x_D$.
It induces a total order on the monomials of $R$.
Choose the least term $f_0$ of $f$.
By multiplying $f_0^{-1}$, we may assume $f_0 = 1$.%
\footnote{
If $D=1$, $f$ would be a polynomial of nonnegative exponents with the lowest order term being 1.
If $D=2$ and $f = y + y^2 + x$, then the least term is $y$.
After multiplying $f_0^{-1}$, it becomes $1+y+xy^{-1}$.
}

We claim that the sequence
\begin{equation}
 f,\ \  f^2 f,\ \  f^4 f^2 f,\ \  \ldots,\ \  f^{2^n}f^{2^{n-1}} \cdots f^2 f,\ \  \ldots
\label{eq:ftn-series}
\end{equation}
converges to $\tilde f \in \tilde R$,
where $\tilde R$ is the set of all $\FF_2$-valued functions on $\Lambda$.
Given the claim,
since $f^{2^n} e = e + (f-1)^{2^n}e = \epsilon(f^{2^{n-1}} \cdots f^2 f p)$ where $p \in P$,
we conclude that $e = \epsilon( \tilde f p )$ is a charge.

If $f$ is of nonnegative exponents, and hence $f \in S=\FF_2[x_1,\ldots,x_D]$,
then the claim is clearly true.
Indeed, the positive degree terms of $f^{2^n} = 1 + (f-1)^{2^n}$ 
are in the ideal power $(x_1,\ldots,x_D)^{2^p} \subset S$.
Therefore, the sequence Eq.~\eqref{eq:ftn-series} converges
in the formal power series ring $\FF_2[[x_1,\ldots,x_D]]$,
which can be regarded as a subset of $\tilde R$.
If $f$ is not of nonnegative exponents, one can introduce the following change of basis
of the lattice $\mathbb{Z}^D$ such that $f$ becomes of nonnegative exponents.
In other words, the sequence Eq.~\eqref{eq:ftn-series} is in fact contained in a ring 
that is isomorphic to the formal power series ring, where the convergence is clear.

For any nonnegative integers $m_1,\ldots,m_{D-1}$, define a linear transformation
\[
\zeta_m =
\zeta_{(m_1,m_2,\ldots,m_{D-1})} :
\begin{pmatrix}
 a_1 \\ a_2 \\ \vdots \\ a_D
\end{pmatrix}
\mapsto
 \begin{pmatrix}
  a'_1 \\ a'_2 \\ \vdots \\ a'_D
 \end{pmatrix}
=
\begin{pmatrix}
 1      & 0      & 0      & \cdots & 0       \\
 m_1    & 1      &  0     &        & 0       \\
 m_1    & m_2    &  1     &        & 0       \\
 \vdots &        &        & \ddots & \vdots  \\
 m_1    & m_2    & \cdots & m_{D-1}& 1       \\
\end{pmatrix}
\begin{pmatrix}
 a_1 \\ a_2 \\ \vdots \\ a_D
\end{pmatrix}
\text{ on } \mathbb{Z}^D .
\]
$\zeta_m$ induces the map $x_1^{a_1} \cdots x_D^{a_D} \mapsto x_1^{a'_1} \cdots x_D^{a'_D}$ on $R$.
Let $u=x_1^{a_1} \cdots x_D^{a_D}$ be an arbitrary term of $f$ other than $1$, so $u \succ 1$.
For the smallest $i \in \{1,\ldots,D\}$ such that $a_i \neq 0$,
one has $a_i > 0$ due to the lexicographic order.
Hence, if we choose $m_i$ large enough and set $m_j=0\,(j\neq i)$,
then $\zeta_m (u)$ has nonnegative exponents.
Since any $\zeta_{m}$ maps a nonnegative exponent term to a nonnegative exponent term,
and there are only finitely many terms in $f$,
it follows that
there is a finite composition $\zeta$ of $\zeta_m$'s 
which maps $f$ to a polynomial of nonnegative exponents.%
\footnote{
For our previous example $f = 1+y+xy^{-1}$,
one takes $\zeta : x^i y^j \mapsto x^i y^{i+j}$, so $\zeta(f) = 1 + y + x$.
}
\end{proof}

Since a nontrivial charge $v$ has finite {\em size} anyway
(the maximum exponent minus the minimum exponent of the Laurent polynomials 
in the $t \times 1$ matrix $v$),
we can say that the charge $v$ is {\em point-like}.
Moreover, we shall have a description how the point-like charge
can be separated from the other by a local process.
By the {\em local process} we mean a sequence of Pauli operators $[[o_1, \ldots, o_n]]$
such that $o_{i+1} - o_i$ is a monomial.
The number of excitations, i.e., {\em energy}, at an instant $i$ 
will be the number of terms in $\epsilon(o_i)$.

\begin{theorem}
\cite{NewmanMoore1999Glassy}
If there is a fractal generator of a code Hamiltonian,
then for all sufficiently large $r$,
there is a local process starting from the identity
by which a point-like charge is separated from the other excitations by distance at least $2^r$.
One can choose the local process in such a way that at any intermediate step
there are at most $c r$ excitations for some constant $c$ independent of $r$.
\label{thm:logarithmic-energy-barrier-for-fractal-operator}
\end{theorem}
For notational simplicity, we denote the local process $[[o_1,\ldots,o_n]]$
by
\[
s = [o_1,~ o_2-o_1,~ o_3 - o_2, \ldots, o_n - o_{n-1}].
\]
It is a {\em recipe} to construct $o_n$, consisted of single qubit operators.
$o_n$ can be expressed as ``$o_n = \int s$'', the sum of all elements in the recipe.
\begin{proof}
Let $f$ be a fractal generator, 
and put $fv = \epsilon (p)$
where $v \notin \im \epsilon$.
We already know $v$ is a point-like nontrivial charge.
Write
\[
 p = \sum_{i=1}^n p_i, \quad f = \sum_{i=1}^l f_i
\]
where each of $p_i$ and $f_i$ is a monomial.
Let $s_0 = [ 0, p_1, p_2, \ldots, p_{n}]$ be a recipe for constructing $p$; $\int s_0 = p$.
Given $s_i$, define inductively
\[
 s_{i+1} = (f_1^{2^i} \cdot s_i) \circ (f_2^{2^i} \cdot s_i) \circ \cdots \circ (f_l^{2^i} \cdot s_i) 
\]
where $\circ$ denotes the concatenation and $f_i \cdot [u_1,\ldots,u_{n'}] = [f_i u_1, \ldots, f_i u_{n'}]$.
It is clear that $s_{i+1}$ constructs the Pauli operator
\begin{align*}
\int s_{r} 
= f^{2^{r-1}} \int s_{r-1} = f^{2^{r-1}} f^{2^{r-2}} \int s_{r-2} 
= f^{2^{r-1} + 2^{r-2} + \cdots + 1} \int s_0 = f^{2^r-1} p 
\end{align*}
whose image under $\epsilon$ is $f^{2^r}v$. 
Thus, if $r$ is large enough so that $2^r$ is greater than the size of $v$,
the configuration of excitations is precisely $l$ copies of $v$.
The distance between $v$'s is at least $2^r$ minus twice the size of $v$.

Therefore, there is a constant $e > 0$ such that for any $r \ge 0$
the energy of $f^{2^r}v \in E$ is $\le e$.
Let $\Delta(r)$ be the maximum energy during the process $s_r$.
We prove by induction on $r$ that 
\[
\Delta(r) \le el (r+1).
\]
When $r = 0$, it is trivial.
In $s_{r+1}$, the energy is $\le \Delta(r)$ until $f_1^{2^r} s_r$ is finished.
At the end of $f_1^{2^r} s_r$, the energy is $\le e $.
During the subsequent $f_2^{2^r} s_r$, the energy is $\le \Delta(r) + e$,
and at the end of $(f_1^{2^r} s_r) \circ (f_2^{2^r} s_r)$, the energy is $\le 2 e$.
During the subsequent $f_j^{2^r} s_r$, the energy is $\le \Delta(r) + je$.
Therefore, 
\[
\Delta(r+1) \le \Delta(r) + el \le el(r+2)
\]
by the induction hypothesis.
\end{proof}

Fractal operators appear in Newman-Moore model~\cite{NewmanMoore1999Glassy} where classical spin glass is discussed.
Their model has generating matrix $\sigma = \begin{pmatrix} 0 & 1+x+y \end{pmatrix}^T$.
The theorem is a simple generalization of Newman and Moore's construction.
Another explicit example of fractal operators in a quantum model can be found in \cite{BravyiHaah2011Energy}.

Note that the notion of fractal generators includes that of `string operators'.
In fact, a fractal generator that contains exactly two terms
gives a family of nontrivial {\em string segments} of unbounded length,
as defined in \cite{Haah2011Local}.

Below, we point out a couple of sufficient conditions 
for nontrivial charges, or equivalently, fractal generators to exist.
\begin{prop}
For code Hamiltonians,
the existence of a fractal generator is a property of an equivalence class of Hamiltonians
in the sense of Definition~\ref{defn:equiv-H}.
\label{prop:fractal-property-of-eq-class}
\end{prop}
\begin{proof}
Suppose $\im \sigma = \im \sigma'$.
Each column of $\sigma'$ is a $R$-linear combination of those of $\sigma$, and vice versa.
Thus, there is a matrix $B$ and $B'$ such that $\epsilon' = B\epsilon$ and $\epsilon = B' \epsilon'$.
$BB'$ and $B'B$ are identity on $\im \epsilon'$ and $\im \epsilon$ respectively.
In particular, $B'$ and $B$ are injective on $\im \epsilon'$ and $\im \epsilon$ respectively.
Suppose $f$ is a fractal generator for $\epsilon$, i.e., $fv = \epsilon p \ne 0$.
Then, $0 \ne Bfv = f B v = B\epsilon(p) = \epsilon'(p)$.
If $Bv \in \im \epsilon'$, then $v = B'Bv \in \im \epsilon$, a contradiction.
Therefore, $f$ is also a fractal generator for $\epsilon'$.
By symmetry, a fractal generator for $\epsilon'$ is a fractal generator for $\epsilon$, too.

Suppose $R' \subseteq R$ is a coarse-grained base ring.
If $\coker \epsilon$ is torsion-free as an $R$-module,
then so it is as an $R'$-module.
If $f \in R$ is a fractal generator,
the determinant of $f$ as a matrix over $R'$ is a fractal generator.

A symplectic transformation or tensoring ancillas
does not change $\coker \epsilon$.
\end{proof}

\begin{prop}
For any ring $S$ and $t \ge 1$,
if $0 \to S^t \to S^{2t} \xrightarrow{\phi} S^t$ is exact and $I(\phi) \ne S$,
then $\coker \phi$ is not torsion-free.
In particular,
for a degenerate exact code Hamiltonian,
if $\sigma$ is injective, then there exists a fractal generator.
\label{prop:injective-sigma-implies-fractal}
\end{prop}
\begin{proof}
By Proposition~\ref{prop:exact-sequence}, $\rank \phi = t$.
Since $0 \subsetneq I_t(\phi) \subsetneq S$ is the initial Fitting ideal,
we have $0 \ne \ann \coker \phi \ne S$. That is, $\coker \phi$ is not torsion-free.

For the second statement, set $S=R$.
If $\sigma$ is injective, we have an exact sequence
\[
 0 \to G \xrightarrow{\sigma} P \xrightarrow{\epsilon} E .
\]
By Remark~\ref{rem:rank-sigma-m}, $t = \rank G = \rank \sigma = \rank \epsilon = q$.
\end{proof}

\begin{prop}
Suppose the characteristic dimension is $D-2$
for a degenerate exact code Hamiltonian.
Then, there exists a fractal generator.
\label{prop:maximal-characteristic-dimension-means-fractal}
\end{prop}
\begin{proof}
Suppose on the contrary there are no fractal generators.
Then, by Proposition~\ref{prop:nofractal-torsionless},
\[
 G \xrightarrow{\sigma} P \xrightarrow{\epsilon} E \to E'
\]
is exact for some finitely generated free module $E'$.
Since $\coker \epsilon$ has finite free resolution by Lemma~\ref{lem:coker-epsilon-resolution-length-D},
Proposition~\ref{prop:exact-sequence} implies
$\codim I(\sigma) \ge 3$ unless $I(\sigma) = R$.
But, $\codim I(\sigma) = 2$ and $I(\sigma) \ne R$ by 
Corollary~\ref{cor:unit-characteristic-ideal-means-nondegeneracy}.
This is a contradiction.
\end{proof}

\section{One dimension}

The group algebra $R = \FF_2[x,\bar x]$ for the one dimensional lattice $\mathbb{Z}$
is a Euclidean domain
where the degree of a polynomial is defined to be the maximum exponent minus the minimum exponent.
(In particular, any monomial has degree $0$.)
Given two polynomials $f,g$ in $R$, one can find their $\gcd$ by the Euclid's algorithm.
It can be viewed as a column operation on the $1 \times 2$ matrix $\begin{pmatrix} f & g \end{pmatrix}$.
Similarly, one can find $\gcd$ of $n$ polynomials by column operations on $1 \times n$ matrix
\[
 \begin{pmatrix}
  f_1 & f_2 & \cdots & f_n
 \end{pmatrix} .
\]
The resulting matrix after the Euclid's algorithm will be
\[
\begin{pmatrix}
\gcd(f_1,\ldots,f_n) & 0 & \cdots & 0
\end{pmatrix} .
\]

Given a matrix $\mathbf M$ of univariate polynomials, we can apply Euclid's algorithm
to the first row and first column by elementary row and column operations
in such a way that the degree of $(1,1)$-entry $\mathbf M_{11}$ decreases unless all other
entries in the first row and column are divisible by $\mathbf M_{11}$.
Since the degree cannot decrease forever, this process must end
with all entries in the first row and column being zero except $\mathbf M_{11}$.
By induction on the number of rows or columns,
we conclude that $\mathbf M$ can be transformed to a diagonal matrix
by the elementary row and column operations.
This is known as the Smith's algorithm.

The following is a consequence of the finiteness of the ground field.
\begin{lem}
Let $\FF$ be a finite field and $S = \FF[x]$ be a polynomial ring.
Let $\phi : S \xrightarrow{f(x) \times} S$ be a $1 \times 1$ matrix such that $f(0) \ne 0$.
$\phi$ can be viewed as an $n \times n$ matrix acting on the free $S'$-module $S$
where $S' = \FF[x']$ and $x' =x^n$.
Then, for some $n \ge 1$, the matrix $\phi$ is transformed by elementary row and column operations
into a diagonal matrix with entries $1$ or $x'-1$.
The number of $x'-1$ entries in the transformed $\phi$ is equal to the degree of $f$.
\label{lem:coarse-graining-single-polynomial-in-1D}
\end{lem}
\begin{proof}
The splitting field $\tilde \FF$ of $f(x)$ is a finite extension of $\FF$.
Since $\tilde \FF$ is finite, every root of $f(x)$ is a root of $x^{n'} -1$ for some $n' \ge 1$.
Choose an integer $p \ge 1$ such that $2^p$ is greater than any multiplicity of the roots of $f(x)$.
Then, clearly $f(x)$ divides $(x^{n'}-1)^{2^p} = x^{2^p n'} - 1$.
Let $n$ be the smallest positive integer such that $f(x)$ divides $x^n-1$.%
\footnote{This part is well-known, at least in the linear cyclic coding theory~\cite{MacWilliamsSloane1977}.}

Consider the coarse-graining by $S' = \FF[x']$ where $x'=x^n$.
$S$ is a free $S'$-module of rank $n$, and $(f)$ is now an endomorphism of the module $S$
represented as an $n \times n$ matrix.
Since $f(x)g(x) = x^n -1$ for some $g(x) \in \FF[x]$,
we have
\[
 A B = (x'-1) \id_n
\]
where $x' = x^n$, and $A, B$ are the matrix representation of $f(x)$ and $g(x)$ respectively as endomorphisms.
$A$ and $B$ have polynomial entries in variable $x'$.
The determinants of $A,B$ are nonzero for their product is $(x'-1)^n \ne 0$.
Let $E_1$ and $E_2$ be the products of elementary matrices such that $A' = E_1 A E_2$ is diagonal.
Such matrices exist by the Smith's algorithm.
Put $B' = E_2^{-1} B E_1^{-1}$.
Then,
\[
 A' B' = E_1 A E_2 E_2^{-1} B E_1^{-1} = E_1 A B E_1^{-1} = (x'-1)\id_n .
\]
Since $A'$ and $I_n$ are diagonal of non-vanishing entries, $B'$ must be diagonal, too.
It follows that the diagonal entries of $A'$ divides $(x'-1)$; that is, they are $1$ or $x'-1$.

The number of entries $x'-1$ can be counted by considering $S/(f(x))$ as an $\FF$-vector space.
It is clear that $\dim_{\FF} S/(f(x)) = \deg f(x)$.
$S/(f(x)) = \coker \phi$ viewed as a $S'$-module is isomorphic to $S'^n / \im A'$,
the vector space dimension of which is precisely the number of $x'-1$ entries in $A'$.
\end{proof}

For example, consider $f(x) = x^2 + x + 1 \in S = \FF_2[x]$.
It is the primitive polynomial of the field $\FF_4$ of four elements over $\FF_2$.
Any element in $\FF_4$ is a solution of $x^4-x=0$.
Since $f(0) = 1$, we see that $n=3$ is the smallest integer such that $f(x)$ divides $x^n-1$.
As a module over $S' = \FF_2[x^3]$, the original ring $S$ is free with (ordered) basis $\{1, x, x^2\}$.
The multiplication by $x$ on $S$ viewed as an endomorphism has a matrix representation
\[
x = 
 \begin{pmatrix}
  0 & 0 & x^3 \\
  1 & 0 & 0   \\
  0 & 1 & 0   
 \end{pmatrix} .
\]
Thus, $f(x)$ as an endomorphism of $S'$-module $S$ has a matrix representation as follows.
\[
 f(x) =
 \begin{pmatrix}
  1 & x^3 & x^3 \\
  1 & 1   & x^3 \\
  1 & 1   & 1  
 \end{pmatrix} 
\cong
 \begin{pmatrix}
  1 & 0      & 0 \\
  0 & x^3 +1 & 0 \\
  0 & 0      & x^3 +1
 \end{pmatrix}
\]
Here, the second matrix is obtained by row and column operations.
There are 2 diagonal entries $x^3+1$ as $f(x)$ is of degree 2.

\begin{theorem}
If $\Lambda = \mathbb{Z}$,
any system governed by a code Hamiltonian
is equivalent to finitely many copies of Ising models,
plus some non-interacting qubits.
In particular, the topological order condition is never satisfied.
\end{theorem}
\noindent
Yoshida~\cite{Yoshida2011Classification} arrived at a similar conclusion
assuming that the ground space degeneracy when the Hamiltonian is defined on a ring 
should be independent of the length of the ring.
If translation group is trivial,
the proof below reduces to a well-known fact that 
the Clifford group is generated by controlled-NOT, Hadamard, and Phase
gates~\cite[Proposition~15.7]{KitaevShenVyalyi2002CQC}.
The proof in fact implies that the group of all symplectic transformations
in one-dimension is generated by elementary symplectic transformations
of Section~\ref{sec:symplectic-transformations}.

We will make use of the elementary symplectic transformations
and coarse-graining to deform $\sigma$ to a familiar form.
Recall that for any elementary row-addition $E$ on the upper block of $\sigma$
there is a unique symplectic transformation that restricts to $E$.

\begin{proof}
Applying Smith's algorithm to the first row and the first column of $2q \times t$ matrix $\sigma$,
one gets
\[
\begin{pmatrix}
  f_1 & 0 \\
  0   & A \\
  \hline
  g_1 & g_2 \\
  \vdots   & B  
\end{pmatrix}
\]
by elementary symplectic transformations.
Let $1 \le i < j \le q$ be integers.
If some $(1,q+j)$-entry is not divisible by $f_1$,
apply Hadamard on $j$-th qubit to bring $(q+j)$-th row to the upper block,
and then run Euclid's algorithm again to reduce the degree of $(1,1)$-entry.
The degree is a positive integer, so this process must end after a finite number of iteration.
Now every $(q+j,1)$-entry is divisible by $f_1$ and hence can be made to be $0$
by the controlled-NOT-Hadamard:
\[
\begin{pmatrix}
  f_1 & 0 \\
  0   & A \\
  \hline
  g_1 & g_2 \\
  0   & B  
\end{pmatrix} .
\]
Further we may assume $\deg f_1 \le \deg g_1$.
Since $\sigma^\dagger \lambda_q \sigma = 0$,
we have a commutativity condition
\[
 \bar f_1 g_1 - \bar g_1 f_1 = 0 .
\]
Write $f_1 = \alpha x^a + \cdots + \beta x^b $
and $ g_1 = \gamma x^c + \cdots + \delta x^d $
where $ a \le b$ and $c \le d$ and $\alpha,\beta,\gamma,\delta \neq 0$.
Then, $\bar f_1 g_1 = \beta \gamma x^{c-b} + \cdots + \alpha \delta x^{d-a}$.
Since $f_1 \bar g_1 = \bar f_1 g_1$, it must hold that $-(c-b)=d-a$
and $\alpha \delta = \beta \gamma$.
Since $\deg f_1 \le \deg g_1$, we have $d-b = -(c-a) \ge 0$.
The controlled-Phase $E_{1+q,1}( -(x^{d-b} + x^{c-a})\delta / \beta )$
will decrease the degree of $g_1$ by two,
which eventually becomes smaller than $\deg f_1$.
One may then apply Hadamard to swap $f_1$ and $g_1$.
Since the degree of $(1,1)$-entry cannot decrease forever,
the process must end with $g_1 = 0$.

The commutativity condition between $i$-th($i>1$) column and the first
is $ f_1 \bar g_i = 0 $. Since $f_1 \ne 0$, we get $g_i = 0$:
\[
\begin{pmatrix}
  f_1 & 0 \\
  0   & A \\
  \hline
  0 & 0 \\
  0 & B  
\end{pmatrix} .
\]
Continuing, we transform $\sigma$ into a diagonal matrix.
(We have shown that $\sigma$ can be transformed 
via elementary symplectic transformations to the Smith normal form.)

Now the Hamiltonian is a sum of non-interacting purely classical spin chains
plus some non-interacting qubits ($f_i = 0$).
It remains to classify classical spin chains
whose stabilizer module is generated by
\[
\begin{pmatrix}
f
\end{pmatrix}
\]
where we omitted the lower half block.
We can always choose $f=f(x)$ such that $f(x)$ has only non-negative exponents and
$f(0)\ne 0$ since $x$ is a unit in $R$.
Lemma~\ref{lem:coarse-graining-single-polynomial-in-1D} says that
$(f)$ becomes a diagonal matrix of entries $1$ or $x'-1$
after a suitable coarse-graining followed by a symplectic transformation
and column operations.
$1$ describes the ancilla qubits,
and $x'-1 = x'+1$ does the Ising model.
\end{proof}

\section{Two dimensions}
In the following two sections we will be mainly interested in exact code Hamiltonians.
If $D=2$, the lattice is $\Lambda = \mathbb{Z}^2$, and our base ring is $R = \FF_2 [x,\bar x, y, \bar y]$.

The following asserts that the local relations --- 
a few terms in the Hamiltonian that multiply to identity in a nontrivial way
as in 2D Ising model, or the kernel of $\sigma$ ---
among the terms in a code Hamiltonian,
can be completely removed for exact Hamiltonians in two dimensions~\cite{Bombin2011Structure}.
We prove a more general version.
\begin{lem}
If $G \xrightarrow{\sigma} P \xrightarrow{\epsilon} E$ is exact
over $R = F_2[ x_1, \bar x_1, \ldots, x_D, \bar x_D ]$,
There exists $\sigma' : G' \to P$ such that $\im \sigma' = \im \sigma$ and
\[
 0 \to G_{D-2} \to \cdots \to G_1 \to G' \xrightarrow{\sigma'} P \xrightarrow{\epsilon} E
\]
is an exact sequence of free $R$-modules. If $D=2$, one can choose $\sigma'$ to be injective.
\label{lem:coker-epsilon-resolution-length-D}
\end{lem}
\noindent
The lemma is almost the same as the Hilbert syzygy theorem~\cite[Corollary~15.11]{Eisenbud}
applied to $\coker \epsilon$,
which states that any finitely generated module over a polynomial ring with $n$ variables
has a finite free resolution of length $\le n$, by finitely generated free modules.
A difference is that our two maps on the far right in the resolution 
has to be related as $\epsilon = \sigma^\dagger \lambda$.
To this end, we make use of a constructive version of Hilbert syzygy theorem via Gr\"obner basis.
\begin{prop}\cite[Theorem~15.10, Corollary~15.11]{Eisenbud}
Let $\{ g_1, \ldots, g_n \}$ be a Gr\"obner basis of a submodule of
a free module $M_0$ over a polynomial ring.
Then, the S-polynomials $\tau_{ij}$ of $\{ g_i \}$
in the free module $M_1 = \bigoplus_{i=1}^n S e_i$
generate the syzygies for $\{g_i\}$.
If the variable $x_1,\ldots, x_s$ are absent from the initial terms of $g_i$,
one can define a monomial order on $M_1$
such that $x_1,\ldots,x_{s+1}$ is absent from
the initial terms of $\tau_{ij}$.
If all variables are absent from the initial terms of $g_i$,
then $M_0/(g_1,\ldots,g_n)$ is free.
\label{prop:constructive-syzygy}
\end{prop}
\begin{proof}[Proof of \ref{lem:coker-epsilon-resolution-length-D}]
Without loss of generality
assume that the $t \times 2q$ matrix $\epsilon$ have entries with nonnegative exponents,
so $\epsilon$ has entries in $S=\FF_2[x_1,\ldots,x_D]$.
Below, every module is over the polynomial ring $S$ unless otherwise noted.
Let $E_+$ be the free $S$-module of rank equal to $\mathrm{rank}_R~ E$.

If $g_1,\cdots, g_{2q}$ are the columns of $\epsilon$,
apply Buchberger's algorithm to obtain a Gr\"obner basis
$g_1,\cdots, g_{2q}, \ldots, g_n$ of $\im \epsilon$.
Let $\epsilon'$ be the matrix whose columns are $g_1,\ldots, g_n$.
We regard $\epsilon'$ as a map $M_0 \to E_+$.
By Proposition~\ref{prop:constructive-syzygy},
the initial terms of the syzygy generators (S-polynomials) 
$\tau_{ij}$ for $\{ g_i \}$ lacks the variable $x_1$.
Writing each $\tau_{ij}$ in a column of a matrix $\tau_1$, we have a map $\tau_1 : M_1 \to M_0$.

By induction on $D$, we have an exact sequence
\[
 M_{D} \xrightarrow{\tau_D} M_{D-1} \xrightarrow{\tau_{D-1}}
 \cdots \xrightarrow{\tau_1} M_0 \xrightarrow{\epsilon'} E_+ 
\]
of free $S$-modules, where the initial terms of columns of $\tau_D$ lack all the variables.
By Proposition~\ref{prop:constructive-syzygy} again, $M'_{D-1} = M_{D-1} / \im \tau_{\tau_D}$ is free.
Since $\ker \tau_{D-1} = \im \tau_D$, we have
\[
 0 \to M'_{D-1} \xrightarrow{\tilde \tau_{D-1}}
 \cdots \xrightarrow{\tau_1} M_0 \xrightarrow{\epsilon'} E_+ 
\]
Since $g_{2q+1},\ldots,g_n$ are $S$-linear combinations of $g_1,\ldots, g_{2q}$,
there is a basis change of $M_0$ so that the matrix representation of $\epsilon'$ becomes
\[
 \epsilon' \cong
\begin{pmatrix}
 \epsilon & 0
\end{pmatrix} .
\]
With respect to this basis of $M_0$,
the matrix of $\tau_1$ is
\[
 \tau_1 \cong
\begin{pmatrix}
 \tau_{1u} \\
 \tau_{1d}
\end{pmatrix}
\]
where $\tau_{1u}$ is the upper $2q \times t'$ submatrix.
Since $\ker \epsilon' = \im \tau_1$, 
The first row $r$ of $\tau_{1d}$ should generate $1 \in S$.
(This property is called unimodularity.)
Quillen-Suslin theorem~\cite[Chapter~XXI Theorem~3.5]{Lang}
states that there exists a basis change of $M_1$
such that $r$ becomes $\begin{pmatrix}1 & 0 & \cdots & 0 \end{pmatrix}$.
Then, by some basis change of $M_0$, one can make 
\[
\epsilon' \cong
\begin{pmatrix}
 \epsilon & 0
\end{pmatrix} ,
\quad
 \tau_{1d} \cong
\begin{pmatrix}
 1 & 0 \\
 0 & \tau'_{1d}
\end{pmatrix} .
\]
where $\tau'_{1d}$ is a submatrix.
By induction on the number of rows in $\tau_{1d}$,
we deduce that the matrix of $\tau_1$ can be brought to
\[
\epsilon' \cong
\begin{pmatrix}
 \epsilon & 0
\end{pmatrix} ,
\quad
 \tau_1 \cong
\begin{pmatrix}
 \sigma'' & \sigma' \\
 I & 0
\end{pmatrix}
\]
Note that $\epsilon \sigma'' = 0$ and $\epsilon \sigma' = 0$.
The basis change of $M_0$ by $\begin{pmatrix} I & -\sigma'' \\ 0 & I \end{pmatrix}$ gives
\[
\epsilon' \cong
\begin{pmatrix}
 \epsilon & 0
\end{pmatrix} ,
\quad
 \tau_1 \cong
\begin{pmatrix}
 0 & \sigma' \\
 I & 0
\end{pmatrix} .
\]
The kernel of $\begin{pmatrix} \sigma' \\ 0 \end{pmatrix}$ determines $\ker \tau_1 = \im \tau_2$.
Let $M'_1$ denote the projection of $M_1$ such that the sequence
\[
 0 \to M'_{D-1} \xrightarrow{\tilde \tau_{D-1}}
 \cdots \to M_2 \to M'_1 \xrightarrow{\sigma'} M'_0 \xrightarrow{\epsilon} E_+ 
\]
of free $S$-modules is exact.

Taking the ring of fractions with respect to the multiplicatively closed set 
\[
U = \{x_1^{i_1} \cdots x_D^{i_D} | i_1,\ldots,i_D \ge 0\},
\]
we finally obtain the desired exact sequence over $U^{-1}S = R$ 
with $P = U^{-1}M'_0$ and $E = U^{-1}E_+$.
Since $\im \sigma = \ker \epsilon$, we have $\im \sigma' = \im \sigma$.
\end{proof}

\begin{lem}
Let $R$ be a Laurent polynomial ring in $D$ variables over a finite field $\FF$,
and $N$ be a module over $R$.
Suppose $J = \ann_R N$ is a proper ideal such that $\dim R/J = 0$.
Then, there exists an integer $L \ge 1$ such that
\[
 \ann_{R'} N = (x_1^L -1, \ldots, x_D^L -1) \subseteq R'
\]
where $R' = \FF[x_1^{\pm L},\ldots,x_D^{\pm L}]$ is a subring of $R$.
\label{lem:zero-dim-ideal-over-finite-field}
\end{lem}

This is a variant of Lemma~\ref{lem:coarse-graining-single-polynomial-in-1D}.
The annihilator $J = \ann_R N$ is the set of all elements $r \in R$ such that
$r n = 0$ for any $n \in N$.
It is an ideal;
if $r_1, r_2 \in \ann_R N$, then
$r_1+r_2$ is an annihilator since $(r_1 + r_2)n = r_1 n + r_2 n = 0$,
and $a r_1 \in \ann_R N$ for any $a \in R$ since $(a r_1) n = a(r_1 n)=0$.
If $R' \subseteq R$ is a subring and $N$ is an $R$-module,
$N$ is an $R'$-module naturally.
Clearly, $J' = \ann_{R'} N$ is by definition equal to $(\ann_{R} N) \cap R'$.
Note that $J'$ is the kernel of the composite map $R' \hookrightarrow R \to R/J$.
Hence, we have an algebra homomorphism $\varphi' : R'/J' \to R/J$.
Although $R'$ is a subring, it is isomorphic to $R$ via the correspondence $x_i^L \leftrightarrow x_i$.
Therefore, we may view $\varphi'$ as a map $\varphi: R/I \to R/J$ for some ideal $I \subseteq R$.
It is a homomorphism such that $\varphi(x_i) = x_i^L$.
Considering the algebras as the set of all functions on the algebraic sets $V(I)$ and $V(J)$ defined by $I$ and $J$, respectively,
we obtain a map $\hat \varphi : V(J) \to V(I)$.
Intuitively, $\hat \varphi$ maps each point $(a_1,\ldots,a_D) \in \FF^D$ to $(a_1^L,\ldots,a_D^L) \in \FF^D$.
In a finite field, any nonzero element is a root of unity.
Since $\dim R/J = 0$, which means that $V(J)$ is a finite set, we can find a certain $L$
so $V(I)$ would consist of a single point. A formal proof is as follows.

\begin{proof}
Since $R$ is a finitely generated algebra over a field,
for any maximal ideal $\mm$ of $R$, the field $R/\mm$ is a finite extension of $\FF$
(Nullstellensatz~\cite[Theorem~4.19]{Eisenbud}).
Hence, $R/\mm$ is a finite field.
Since $x_i$ is a unit in $R$, the image $a_i \in R/\mm$ of $x_i$ is nonzero.
$a_i$ being an element of finite field, a power of $a_i$ is $1$.
Therefore, there is a positive integer $n$ such that
$\bb_n = (x_1^n-1,\ldots,x_D^n-1) \subseteq \mm$.
Since $x^n-1$ divides $x^{nn'}-1$,
we see that there exists $n \ge 1$ such that $\bb_n \subseteq \mm_1 \cap \mm_2$
for any two maximal ideals $\mm_1, \mm_2$.
One extends this by induction to any finite number of maximal ideals.

Since $\dim R/J = 0$, any prime ideal of $R/J$ is maximal and 
the Artinian ring $R/J$ has only finitely many maximal ideals.
$\mathrm{rad}~ J$ is then the intersection of the contractions (pull-backs)
of these finitely many maximal ideals.
Therefore, there is $n \ge 1$ such that
\[
 \bb_n \subseteq \mathrm{rad}~ J .
\]
Since $R$ is Noetherian, $(\mathrm{rad}~J)^{p^r} \subseteq J$ for some $r \ge 0$
where $p$ is the characteristic of $\FF$. Hence, we have 
\[
\bb_{np^r} \subseteq \bb_n^{p^r} \subseteq (\mathrm{rad}~J)^{p^r} \subseteq J.
\]

Let $L = np^r$.
If $R' = \FF[x_1^L,\bar x_1^L,\ldots,x_D^L,\bar x_D^L]$,
$\ann_{R'} N$ is nothing but $J \cap R'$.
We have just shown $\bb_L \cap R' \subseteq J \cap R'$.
Since $J$ is a proper ideal, we have $1 \notin J \cap R'$.
Thus, $\bb_L \cap R'= J \cap R'$ since $\bb_L \cap R'$ is maximal in $R'$.
\end{proof}

\begin{theorem}
For any two dimensional degenerate exact code Hamiltonian,
there exists an equivalent Hamiltonian such that
\[
\ann \coker \epsilon = (x-1,y-1).
\]
Thus, $\coker \epsilon$ is a torsion module.
\label{thm:structure-2d-ann-coker-epsilon}
\end{theorem}
The content of Theorem~\ref{thm:structure-2d-ann-coker-epsilon}
is presented in \cite{Bombin2011Structure}.
We will comment on it after the proof.
\begin{proof}
By Lemma~\ref{lem:coker-epsilon-resolution-length-D},
we can find an equivalent Hamiltonian such that the generating map $\sigma$ for its stabilizer module
is injective:
\[
 0 \to G \xrightarrow{\sigma} P .
\]
Let $t$ be the rank of $G$.
The exactness condition says
\[
 0 \to G \xrightarrow{\sigma} P \xrightarrow{\epsilon} E
\]
is exact where $\epsilon = \sigma^\dagger \lambda_q$ and $E$ has rank $t$.
Applying Proposition~\ref{prop:exact-sequence},
since $\overline{ I(\sigma) } = I(\epsilon)$ 
and hence in particular $\codim I(\sigma) = \codim I(\epsilon)$,
we have that $q=t$ and $\codim I(\epsilon) \ge 2$ if $I(\epsilon) \ne R$.
But, $I(\epsilon) \ne R$ by Corollary~\ref{cor:unit-characteristic-ideal-means-nondegeneracy}.

Since $q=t$, $I(\epsilon)$ is equal to the initial Fitting ideal,
and therefore has the same radical as the annihilator of $\coker \epsilon = E/ \im \epsilon$.
(See \cite[Proposition~20.7]{Eisenbud} or \cite[Chapter~XIX Proposition~2.5]{Lang}.)
In particular, $\dim R / (\ann \coker \epsilon) = 0$.
Apply Lemma~\ref{lem:zero-dim-ideal-over-finite-field} to conclude the proof.
\end{proof}

An interpretation of the theorem is the following.
For systems of qubits, Theorem~\ref{thm:structure-2d-ann-coker-epsilon} says that
$x+1$ and $y+1$ are in $\ann \coker \epsilon$.
In other words, any element $v$ of $E$ is a charge,
and a pair of $v$'s of distance 1 apart
can be created by a local operator.
Equivalently, $v$ can be translated by distance 1 by the local operator.
Since translation by distance 1 generates all translations of the lattice,
we see that any excitation can be moved through the system by some sequence of local operators.
This is exactly what happens in the 2D toric code:
Any excited state is described by a configuration of magnetic and electric charge,
which can be moved to a different position by a string operator.

Moreover, since $(x-1,y-1) = \ann \coker \epsilon$,
the action of $x,y \in R$ on $\coker \epsilon$ is the same as the identity action.
Therefore, the $R$-module $\coker \epsilon$ is completely determined
up to isomorphism by its dimension $k$ as an $\FF_2$-vector space.
In particular, $\coker \epsilon$ is a finite set,
which means there are finitely many charges.
The module $K(L)$ of Pauli operators acting on the ground space (logical operators),
can be viewed as $K(L) = \Tor_1(\coker \epsilon, R/\bb_L)$.
Thus, $K(L)$ is determined by $k$ up to $R$-module isomorphisms.
This implies that the translations of a logical operator are all equivalent.
It is not too obvious at this moment whether the symplectic structure,
or the commutation relations among the logical operators,
of $K(L)$ is also completely determined.

Yoshida~\cite{Yoshida2011Classification} argued a similar result assuming
that the ground state degeneracy should be independent of system size.
Bombin~\cite{Bombin2011Structure} later claimed without the constant degeneracy assumption
that one can choose locally independent stabilizer generators in a `translationally invariant way'
in two dimensions, for which Lemma~\ref{lem:coker-epsilon-resolution-length-D} is a generalization,
and that there are finitely many topological charges,
which is immediate from Theorem~\ref{thm:structure-2d-ann-coker-epsilon} since $\coker \epsilon$ is a finite set.
The claim is further strengthened assuming extra conditions by Bombin~\emph{et al.}~\cite{BombinDuclosCianciPoulin2012},
which can be summarized by saying
that $\sigma$ is a finite direct sum of $\sigma_\text{2D-toric}$ in Example~\ref{eg:2d-toric}.

\begin{rem}
Although the strings are capable of moving charges on the lattice,
it could be very long compared to the interaction range.
Consider
\[
 \epsilon_p = 
\begin{pmatrix}
 p(x) & p(y) & 0         & 0          \\
 0    & 0    & p(\bar y) & -p(\bar x)
\end{pmatrix}
\]
where $p$ is any polynomial. It defines an exact code Hamiltonian.
For instance, the choice $p(t) = t-1$ reproduces the 2D toric code of Example~\ref{eg:2d-toric}.
Now let $p(t)$ be a primitive polynomial of the extension field $\FF_{2^w}$ over $\FF_2$.
$p(t)$ has coefficients in the base field $\FF_2$ and factorizes in $\FF_{2^w}$
as $p(t) = (t-\theta)(t-\theta^2)(t-\theta^{2^2})\cdots(t-\theta^{2^{w-1}})$.
(See \cite[Chapter~V Section~5]{Lang}.)
The multiplicative order of $\theta$ is $N = 2^w -1$.
The degree $w$ of $p(t)$ may be called the interaction range.
If the charge $e = \begin{pmatrix}1 & 0 \end{pmatrix}^T$ at $(0,0)\in \ZZ^2$
is transported to $(a,b) \in \ZZ^2 \setminus \{ (0,0) \}$ by some finitely supported operator,
we have $(x^a y^b -1)e \in \im \epsilon$.
That is, $x^a y^b -1 \in (p(x),p(y))$.
Substituting $x \mapsto \theta$ and $y \mapsto \theta^{2^m}$,
we see that $\theta^{a+2^m b} = 1$ or $a+2^m b \equiv 0 \pmod N$ for any $m \in \ZZ$.
In other words, $a \equiv -b \equiv -2b \pmod N$. It follows that $|a| + |b| \ge \frac{N}{2-1}$.
Therefore, the length of the string segment transporting a charge is exponential in the interaction range $w$.
\end{rem}

\section{Three dimensions}
\label{sec:3d}

In the previous section, we derived a consequence of the exactness of code Hamiltonians.
The two-dimensional Hamiltonian was special 
so we were able to characterize the behavior of the charges more or less completely.
Here, we prove a weaker property of three dimensions that there must exist a nontrivial charge
for any exact code Hamiltonian.
It follows from Theorems~\ref{thm:charge-equals-torsion},%
\ref{thm:logarithmic-energy-barrier-for-fractal-operator}
that such a charge can spread through the system
by surmounting the logarithmic energy barrier.
\begin{lem}
Suppose $D=3$,
\[
  0 \to G_1 \xrightarrow{\sigma_1} G \xrightarrow{\sigma} P \xrightarrow{\epsilon=\sigma^\dagger \lambda_q} E
\]
is exact, and $I(\sigma) \subseteq \mm = (x-1,y-1,z-1)$.
Then, $\coker \epsilon$ is not torsion-free.
\label{lem:char-ideal-with-support-on-origin-implies-fractal}
\end{lem}
\begin{proof}
Suppose on the contrary $\coker \epsilon$ is torsion-free.
We have an exact sequence
\[
 0 \to G_1 \xrightarrow{\sigma_1} G \xrightarrow{\sigma} P \xrightarrow{\epsilon} E \to E'.
\]
If $G_1 = 0$, Proposition~\ref{prop:injective-sigma-implies-fractal} implies the conclusion.
So we assume $G_1 \ne 0$, and therefore we have $I(\sigma_1) = R$ by Proposition~\ref{prop:exact-sequence}.

Let us localize the sequence at $\mm$, so $I(\sigma_1)_\mm = R_\mm$.
Since $\rank (G_1)_\mm = \rank (\sigma_1)_\mm$, the matrix of $(\sigma_1)_\mm$ becomes
\[
 (\sigma_1)_\mm =
\begin{pmatrix}
 0 \\
 I
\end{pmatrix}
\]
for some basis of $(G_1)_\mm$ and $G_\mm$. See the proof of 
Lemma~\ref{lem:associated-maximal-localized-homology}.
In other words, there is an invertible matrix $B \in \mathrm{GL}_{ t \times t}( R_\mm )$ such that
\[
 \sigma_\mm B = 
\begin{pmatrix}
 \tilde \sigma & 0
\end{pmatrix}
\]
where $\tilde \sigma$ is the $2q \times t'$ submatrix.
Note that the antipode map is a well-defined automorphism of $R_\mm$
since $\overline \mm = \mm$.

Since $\epsilon = \sigma^\dagger \lambda_q$, we have
\begin{equation}
 B^\dagger \epsilon_\mm = 
\begin{pmatrix}
 \tilde \sigma^\dagger \\
 0
\end{pmatrix} \lambda_q
=
\begin{pmatrix}
 \tilde \sigma^\dagger \lambda_q \\
 0
\end{pmatrix} .
\label{eq:tilde-epsilon-sigma-dagger}
\end{equation}
Therefore, we get a new exact sequence
\[
   0 \to G' \xrightarrow{\tilde \sigma} P_\mm \xrightarrow{\tilde \epsilon = \tilde \sigma^\dagger \lambda_q} R_\mm^{t'}
\]
where $G' = G_\mm / \im (\sigma_1)_\mm$ is a free $R_\mm$-module and $t' = \rank G'$.
It is clear that 
$\rank \tilde \epsilon = \rank \tilde \sigma$.
Setting $S = R_\mm$ in Proposition~\ref{prop:injective-sigma-implies-fractal} implies that
$\coker \tilde \epsilon$ is not torsion-free.
But, since we are assuming $\coker \epsilon_m$ is torsion-free,
$\coker \tilde \sigma^\dagger$ is also torsion-free by Eq.~\eqref{eq:tilde-epsilon-sigma-dagger}.
This is a contradiction.
\end{proof}

\begin{theorem}
For any three-dimensional, degenerate and exact code Hamiltonian,
there exists a fractal generator.
\label{thm:fractal-exists-in-3D}
\end{theorem}
\begin{proof}
By Lemma~\ref{lem:coker-epsilon-resolution-length-D}, there exists an equivalent Hamiltonian
such that
\[
  0 \to G_1 \xrightarrow{\sigma_1} G \xrightarrow{\sigma} P \xrightarrow{\epsilon=\sigma^\dagger \lambda_q} E
\]
is exact. The existence of a fractal generator is a property of the equivalence class
by Proposition~\ref{prop:fractal-property-of-eq-class}.
If we can find a coarse-graining such that $I(\sigma') \subseteq (x'-1,y'-1,z'-1)$,
then Lemma~\ref{lem:char-ideal-with-support-on-origin-implies-fractal} shall imply the conclusion.

Recall that $\epsilon_L$ and $\sigma_L$ denote the induced maps
by factoring out $\bb_L = (x^L-1,y^L-1,z^L-1)$. See Sec.~\ref{sec:degeneracy}.
There exists $L$ such that $K(L) = \ker \epsilon_L / \im \sigma_L \ne 0$ by 
Corollary~\ref{cor:unit-characteristic-ideal-means-nondegeneracy}.
Consider the coarse-grain by $x'=x^L,~ y'=y^L,~ z'=z^L$.
Let $R' = F_2[x'^{\pm 1}, y'^{\pm 1}, z'^{\pm 1}]$ denote the coarse-grained base ring.
If $K'(L')$ denotes $\ker \epsilon'_{L'} / \im \sigma'_{L'}$ as $R'$-module,
we see that $K'(1) = K(L)$ as $\FF_2$-vector space.
In particular, $K'(1) \ne 0$.
Put $\mm = (x'-1,y'-1,z'-1) = \bb'_1 \subseteq R'$.
Then, $K'(1)_{\mm} = K'(1) \ne 0$.
By Lemma~\ref{lem:associated-maximal-localized-homology},
we have $I(\sigma') \subseteq \mm$.
\end{proof}

Yoshida argued that when the ground state degeneracy is constant independent of system size
there exists a string operator~\cite{Yoshida2011feasibility}.
To prove it, we need an algebraic fact.
\begin{prop}
Let $M$ be a finitely presented $R$-module, and $T$ be its torsion submodule.
Let $I$ be the first non-vanishing Fitting ideal of $M$.
Then,
\[
 \rad I \subseteq \rad \ann T .
\]
\label{prop:support-torsion-fitting-ideal}
\end{prop}
\begin{proof}
Let $\pp$ be any prime ideal of $R$ such that $I \not\subseteq \pp$.
By the calculation of the proof of Lemma~\ref{lem:associated-maximal-localized-homology},
$M_\pp$ is a free $R_\pp$-module, and hence is torsion-free.
Since $T$ is embedded in $M$, it follows that $T_\pp = 0$, or equivalently, $\ann T \not\subseteq \pp$.
Since the radical of an ideal is the intersection of all primes containing it~\cite[Proposition 1.8]{AtiyahMacDonald},
the claim is proved.
\end{proof}

\begin{cor}
Let $T$ be the set of all point-like charges modulo locally created ones
of a degenerate and exact code Hamiltonian in three dimensions
of characteristic dimension zero.
Then, one can coarse-grain the lattice such that
\[
 \ann T = (x-1,y-1,z-1) .
\]
\label{cor:annT-3d-characteristic-dimension-0}
\end{cor}
\noindent
The corollary says that any point-like charge is attached to strings and is able to move freely through the lattice.
The condition is implied by Lemma~\ref{lem:k-growing-sequence-L} 
if the ground state degeneracy is constant independent of the system size
when defined on a periodic lattice.
\begin{proof}
By Theorem~\ref{thm:charge-equals-torsion}, $T$ is the torsion submodule of $\coker \epsilon$.
By Theorem~\ref{thm:fractal-exists-in-3D}, $T$ is nonzero.
Setting $M = \coker \epsilon$ in Proposition~\ref{prop:support-torsion-fitting-ideal},
the associated ideal $I_q(\epsilon)$ is the first non-vanishing Fitting ideal of $M$.
Since $\dim R / I_q(\epsilon) = 0$ by assumption, we have $\dim R / \ann T = 0$.
Lemma~\ref{lem:zero-dim-ideal-over-finite-field} implies the claim.
\end{proof}

\section{More examples}
\label{sec:eg}

\begin{example}[Toric codes in higher dimensions]
\label{eg:highD-toric-codes}
Any higher dimensional toric code can be treated similarly as for two dimensional case.
In three dimensions one associates each site with $q=3$ qubits.
It is easily checked that
\[
 \sigma_\text{3D-toric} = 
\begin{pmatrix}
1 + \bar x & 0   & 0      & 0    \\
1 + \bar y & 0   & 0      & 0    \\
1 + \bar z & 0   & 0      & 0    \\
\hline
0          & 0   & 1+z    & 1+y  \\
0          & 1+z & 0      & 1+x  \\
0          & 1+y & 1+x    & 0    \\
\end{pmatrix} .
\]

Both two- and three-dimensional toric codes have the property 
that $\coker \epsilon$ is not torsion-free.
However, in two dimensions any element of $E$ is a physical charge,
whereas in three dimensions $E$ contains physically irrelevant elements.
Note that in both cases, $1+x$ and $1+y$ are fractal generators.
Being consisted of two terms, they generate the `string operators.'

The 4D toric code~\cite{DennisKitaevLandahlEtAl2002Topological}
has $\sigma_x$-type interaction and $\sigma_z$-type interaction.
Originally the qubits are placed on every plaquette of 4D hypercubic lattice;
instead we place $q=6$ qubits on each site.
The generating map $\sigma$ for the stabilizer module is written as
a $12 \times 8$-matrix ($t=8$)
\[
\sigma_\text{4D-toric} =
\begin{pmatrix}
 \sigma_X & 0 \\
 0        & \sigma_Z
\end{pmatrix}
\]
where
\begin{align*}
 \sigma_X &= 
\begin{pmatrix}
1+y & 1+x & 0   & 0   \\
1+w & 0   & 0   & 1+x \\
1+z & 0   & 1+x & 0   \\
0   & 1+z & 1+y & 0   \\
0   & 1+w & 0   & 1+y \\
0   & 0   & 1+w & 1+z
\end{pmatrix}, \\
\bar \sigma_Z &=
\begin{pmatrix}
0   & 0   & 1+w & 1+z \\
0   & 1+z & 1+y & 0   \\
0   & 1+w & 0   & 1+y \\
1+w & 0   & 0   & 1+x \\
1+z & 0   & 1+x & 0   \\
1+y & 1+x & 0   & 0
\end{pmatrix}.
\end{align*}
Note the bar on $\sigma_Z$.

Theorem~\ref{thm:fractal-exists-in-3D} does not prevent
the absence of a fractal generator in four or higher dimensions.
Indeed, this 4D toric code lacks any fractal generator.
To see this, it is enough to consider $\sigma_Z$
since $\overline{\coker \sigma_X^\dagger} \cong \coker \sigma_Z^\dagger$ as $R_4$-modules,
where $R_4 = \FF_2[x^{\pm 1},y^{\pm 1},z^{\pm 1},w^{\pm 1}]$.
If
\[
\epsilon_1 =
 \begin{pmatrix}
1+x & 1+y & 1+z & 1+w
 \end{pmatrix} : R_4^4 \to R_4,
\]
then
\[
 R_4^6 \xrightarrow{\sigma_Z^\dagger} R_4^4 \xrightarrow{\epsilon_1} R_4
\]
is exact.
(A direct way to check it is to compute S-polynomials~\cite[Chapter~15]{Eisenbud}
of the entries of $\epsilon_1$,
and to verify that they all are in the rows of $\sigma_Z$.)
Hence, $\coker \sigma_Z^\dagger$ is torsion-free by Proposition~\ref{prop:nofractal-torsionless}.

For the toric codes in any dimensions, $\sigma$ has nonzero entries of form $x_i-1$.
The radical of the associated ideal $I(\sigma)$ is equal to $\mm = (x_1-1,\ldots, x_D-1)$.
So $\mm$ is the only maximal ideal of $R$ that contains $I(\sigma)$.
The characteristic dimension is zero.
If $2 \nmid L$, since $(\bb_L)_{\mm} = \mm_{\mm}$,
$(\sigma_L)_\mm$ is a zero matrix.
Any other localization of $\sigma_L$ does not contribute to $\dim_{\FF_2} K(L)$
by Lemma~\ref{lem:associated-maximal-localized-homology}.
Therefore, if $2 \nmid L$, $K(L)$ has constant vector space dimension independent of $L$.

There is a more direct way to compute the $R$-module $K(L)$.
For the three-dimensional case, consider a free resolution of $R_3/\mm$,
where $R_3 = \FF_2[x^{\pm 1},y^{\pm 1},z^{\pm 1}]$,
as
\[
 0 \to 
R_3^1 
\xrightarrow{
\partial_3 = 
\begin{pmatrix} a \\ b \\ c \end{pmatrix}
}
R_3^3
\xrightarrow{  
\partial_2 =
\begin{pmatrix}
 0 & c & b  \\
 c & 0 & a  \\
 b & a & 0   
\end{pmatrix}
}
R_3^3
\xrightarrow{
\partial_1 =
\begin{pmatrix}
a & b  & c
\end{pmatrix}
}
R_3^1
\to
R_3/\mm
\to 0
\]
where $a=1+x$, $b=1+y$, and $c=1+z$.
We see that
\begin{equation}
 \sigma_\text{3D-toric} = \bar \partial_3 \oplus \partial_2 ,
 \quad \text{and} \quad
 \epsilon_\text{3D-toric} = \bar \partial_2 \oplus \partial_1  .
\label{eq:resolution-F}
\end{equation}
Therefore,
\begin{align*}
 K(L)_\text{3D-toric} \cong \Tor_1(\coker \epsilon_\text{3D-toric},R_3/\bb_L) 
&\cong \Tor_2( \overline{ R_3/\mm}, R_3/\bb_L ) \oplus \Tor_1( R_3/\mm, R_3/\bb_L ).
\end{align*}
Using $\Tor(M,N) \cong \Tor(N,M)$ and the fact that a resolution of $R_3/\bb_L$ is Eq.~\eqref{eq:resolution-F}
with $a,b,c$ replaced by $x^L-1,y^L-1,z^L-1$, respectively,
we have 
\[
 \Tor_i(R_3/\mm,R_3/\bb_L) \cong \Tor_i(R_3/\mm,R_3/\mm) \cong (\FF_2)^{_3 C _i}
\]
for each $0 \le i \le 3$.
Therefore, $K(L)_\text{3D-toric} \cong (\FF_2)^{_3 C _2} \oplus (\FF_2)^{_3 C _1} \cong (\FF_2)^6$.
The four-dimensional case is similar:
\[
 K(L)_\text{4D-toric} \cong \Tor_2 (R_4/\mm, R_4/\bb_L) \oplus \Tor_2 (\overline{R_4/\mm}, R_4/\bb_L) \cong \left( (\FF_2)^{_4 C _2} \right)^2 .
\]
The calculation here is closely related
to the cellular homology interpretation of toric codes.
\hfill $\Diamond$
\end{example} 

\begin{example}[Wen plaquette~\cite{Wen2003Plaquette}]
This model consists of a single type of interaction ($t=q=1$)
\[
\xymatrix@!0{
X \ar@{-}[r] & Y \ar@{-}[d] \\
Y \ar@{-}[u] & X \ar@{-}[l]
} \quad \quad
\sigma_\text{Wen} = 
\begin{pmatrix}
 1 + x + y + xy \\
\hline
 1 + xy
\end{pmatrix}
\]
where $X,Y$ are abbreviations of $\sigma_x,\sigma_y$.
It is known to be equivalent to the 2D toric code.
Take the coarse-graining given by $R' = \FF_2[x',y',\bar x', \bar y']$ where
\[
 x' = x \bar y, \quad \quad y' = y^2 .
\]
(The coarse-graining considered in this example
is intended to demonstrate a non-square blocking of the old lattice
to obtain a `tilted' new lattice, and is by no means special.)
As an $R'$-module, $R$ is free with basis $\{ 1, y \}$.
With the identification $R = (R' \cdot 1) \oplus (R' \cdot y)$,
we have $x \cdot 1 = x' \cdot y$, $x \cdot y = x'y' \cdot 1$,
and $y \cdot 1 = 1 \cdot y$, $y \cdot y = y' \cdot 1$.
Hence, $x$ and $y$ act on $R'$-modules
as the matrix-multiplications on the left:
\[
x \mapsto 
\begin{pmatrix}
0  & x'y'\\
x' &   0
\end{pmatrix},
\quad
y \mapsto
\begin{pmatrix}
0 & y'\\
1 & 0
\end{pmatrix}.
\]
Identifying
\[
 R^n = [ (R' \cdot 1) \oplus (R' \cdot y) ] \oplus \cdots \oplus [(R' \cdot 1) \oplus (R' \cdot y)] ,
\]
our new $\sigma$ on the coarse-grained lattice becomes
\[
 \sigma' =
\begin{pmatrix}
1+x'y' & y'+x'y'\\
1+x'   & 1+x'y'\\
\hline
1+x'y' & 0 \\
0      & 1+x'y'
\end{pmatrix}.
\]
By a sequence of elementary symplectic transformations, we have
\begin{align*}
\sigma'
&
\xrightarrow[E_{1,3}(1)]{E_{2,4}(1)}
\begin{pmatrix}
0      & y'+x'y' \\
1+x'   & 0       \\
1+x'y' & 0       \\
0      & 1+x'y'
\end{pmatrix} 
\xrightarrow[E_{3,2}(y')]{E_{4,1}(\bar y')}
\begin{pmatrix}
0      & y'+x'y' \\
1+x'   & 0       \\
1+y' & 0       \\:
0      & x'y'+x'
\end{pmatrix}
\\
& \xrightarrow[\times \bar x' \bar y']{\text{col.2}}
\begin{pmatrix}
0      & 1 + \bar x' \\
1+x'   & 0       \\
1+y' & 0       \\
0      & 1+\bar y'
\end{pmatrix}
\xrightarrow{1 \leftrightarrow 3}
\begin{pmatrix}
1+y' & 0     \\
1+x'   & 0   \\
0      & 1+ \bar x'\\
0      & 1 + \bar y'\\
\end{pmatrix},
\end{align*}
which is exactly the 2D toric code.
\hfill $\Diamond$
\end{example} 

\begin{example}[Chamon model~\cite{Chamon2005Quantum,BravyiLeemhuisTerhal2011Topological}]
This three-dimensional model consists of single type of term in the Hamiltonian.
The generating map is
\[
\sigma_\text{Chamon} =
\begin{pmatrix}
 x+\bar x + y + \bar y \\
\hline
 z+\bar z + y + \bar y
\end{pmatrix}.
\]
Since
\[
 \sigma^\dagger \lambda_1
\begin{pmatrix}
0 \\ 
1
\end{pmatrix}
=
(1+ x \bar y)
\begin{pmatrix}
 0 \\
\bar x + y
\end{pmatrix},
\]
$1+x \bar y$ is a fractal generator.
Consisted of two terms, it generates a string operator.
The degeneracy can be calculated using Corollary~\ref{cor:k-formulas}.
Assume all the three linear dimensions of the system are even.
Put
\[
 S = R / (x + \bar x + y + \bar y, z + \bar z + y + \bar y, x^{2l} -1, y^{2m}-1, z^{2n} -1 ).
\]
Then, the $\log_2$ of the degeneracy is $k = \dim_{\FF_2} S$.
In $S$, we have $x + \bar x = y + \bar y = z + \bar z$.
Since $S$ has characteristic 2, it holds that
\[
 w^{p+1} + w^{-p-1} = (w + w^{-1})(w^p + w^{p-2} + \cdots + w^{-p})
\]
for $p \ge 1$ and $w = x,y,z$. By induction on $p$, we see that $w^{p} + w^{-p}$
is a polynomial in $w + w^{-1}$. Therefore,
\[
 x^p + \bar x ^p = y^p + \bar y^p = z^p + \bar z^p
\]
for all $p \ge 1$ in $S$.
Put $g = \gcd(l,m,n)$.
Since $x^l + x^{-l} = y^m + y^{-m} = z^n + z^{-n} = 0$ in $S$,
we have $x^g + x^{-g} = y^g + y^{-g} = z^g + z^{-g} = 0$.

Applying Buchberger's criterion with respect to the lexicographic order in which $x \prec y \prec z$,
we see that
\[
 S = \FF_2[x,y,z] / (z^2 + zx^{2l-1}+ zx+1, y^2 + yx^{2l-1}+yx+1, x^{2g}+1)
\]
is expressed with a Gr\"obner basis. Therefore,
\[
 k = \dim_{\FF_2} S = 8 \gcd(l,m,n).
\]
\hfill $\Diamond$
\end{example}

\begin{example}[Cubic Code]
\newcommand{\drawgenerator}[8]{%
\xymatrix@!0{%
& #8 \ar@{-}[ld]\ar@{.}[dd] \ar@{-}[rr] & & #7 \ar@{-}[ld]  \\%
#1 \ar@{-}[rr] \ar@{-}[dd] &  & #2 \ar@{-}[dd] &            \\%
& #6 \ar@{.}[ld] &  & #5 \ar@{-}[uu] \ar@{.}[ll]       \\%
#3 \ar@{-}[rr] &  & #4 \ar@{-}[ru]                       %
}%
}
The Hamiltonian of code 1 in \cite{Haah2011Local}
is the translation-invariant negative sum of the following two types of interaction terms:
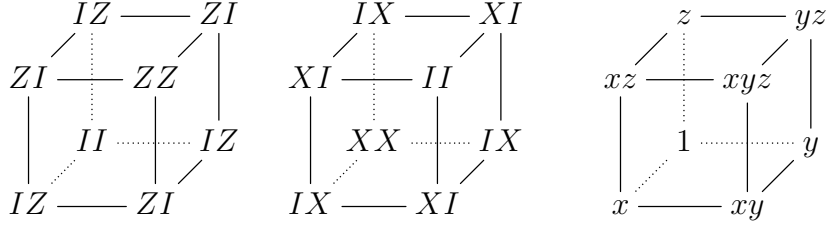
\[
\drawgenerator{ZI}{ZZ}{IZ}{ZI}{IZ}{II}{ZI}{IZ} 
\quad
\drawgenerator{XI}{II}{IX}{XI}{IX}{XX}{XI}{IX}
\quad \quad
\drawgenerator{xz}{xyz}{x}{xy}{y}{1}{yz}{z}
\]
Here, the third cube specifies the coordinate system of the simple cubic lattice.
The corresponding generating map for the stabilizer module is
\[
\sigma_\text{cubic-code} =
\begin{pmatrix}
1 + xy + yz + zx & 0 \\
1 + x + y + z    & 0 \\
0 & 1 + \bar x + \bar y + \bar z \\
0 & 1 + \bar x \bar y + \bar y \bar z + \bar z \bar x
\end{pmatrix}
\]
The associated ideal is contained in a prime ideal of codimension 2:
\[
 I(\sigma) \subseteq ( 1+x+y+z, 1+xy+yz+zx ) = \pp.
\]
Since $\codim I(\sigma) \ge 2$, the characteristic dimension is 1.
Since $\coker \epsilon_\text{cubic-code} = R / \pp \oplus R / \bar \pp$,
any nonzero element of $\pp$ is a fractal generator.

Let us explicitly calculate the ground state degeneracy
when the Hamiltonian is defined on $L \times L \times L$ 
cubic lattice with periodic boundary conditions.
By Corollary~\ref{cor:k-formulas},
\[
 k = \dim_{\FF_2} R / (\pp + \bb_L) \oplus R / (\bar \pp + \bb_L) = 2 \dim_{\FF_2} R / (\pp + \bb_L).
\]
So the calculation of ground state degeneracy comes down to the calculation of
\[
 d = \dim_{\FF_2} T' / \pp
\]
where $T' = \FF_2[x,y,z]/(x^{n_1}-1, y^{n_2}-1, z^{n_3}-1 )$.

We may extend the scalar field to any extension field
without changing $d$.
Let $\FF$ be the algebraic closure of $\FF_2$ and let
\[
 T=\FF[x,y,z]/(x^{n_1}-1, y^{n_2}-1, z^{n_3}-1 )
\]
be an Artinian ring.
By Proposition~\ref{prop:Artin-ring},
it suffices to calculate for each maximal ideal $\mm$ of $T$
the vector space dimension 
\[
d_\mm = \dim_{\FF} (T / \pp)_\mm
\]
of the localized rings, and sum them up.

Suppose $n_1,n_2,n_3 > 1$.
By Nullstellensatz, any maximal ideal of $T$ is of form
$\mm=(x-x_0,y-y_0,z-z_0)$ where $x_0^{n_1}=y_0^{n_2}=z_0^{n_3}=1$.
(If $n_1 = n_2 = n_3 = 1$, then $T$ becomes a field, and there is no maximal ideal other than zero.)
Put $n_i = 2^{l_i}n_i'$ where $n_i'$ is not divisible by $2$.
Since the polynomial $x^{n_1} -1$ contains the factor $x-x_0$ with multiplicity $2^{l_1}$,
it follows that
\[
 T_\mm = \FF [x,y,z]_\mm / ( x^{2^{l_1}}+a',~ y^{2^{l_2}}+b',~ z^{2^{l_3}}+c')
\]
where $a' = x_0^{2^{l_1}}, b' = y_0^{2^{l_2}}, c' = z_0^{2^{l_3}}$.
Hence, $(T/\pp)_\mm \cong \FF[x,y,z] / I'$ where
\[
 I' = (x+y+z+1, xy+xz+yz+1 ,~ x^{2^{l_1}}+a',~ y^{2^{l_2}}+b',~ z^{2^{l_3}}+c').
\]
If $I' = \FF[x,y,z]$, then $d_\mm = 0$.

Without loss of generality, we assume that $l_1 \le l_2 \le l_3$.
By powering the first two generators of $I'$,
we see that $(x_0,y_0,z_0)$ must be a solution of them in order for $I'$ not to be a unit ideal.
Eliminating $z$ and shifting $x \to x+1$, $y \to y+1$,
our objective is to calculate the Gr\"obner basis for the proper ideal
\[
 I = (x^2+xy+y^2, x^{2^{l_1}}+a,~ y^{2^{l_2}}+b )
\]
where $a=a'+1$ and $b = b'+1$. So
\[
 d_\mm = \dim_{\FF} \FF[x,y]/I.
\]
One can easily deduce by induction that $y^{2^m} + x^{2^m-1}(m x +y) \in I$ for any integer $m \ge 0$.
And $b = \omega a^{2^{l_2 - l_1}}$ for a primitive third root of unity $\omega$.
So we arrive at
\[
 I = ( y^2 + yx + x^2,~ yx^{2^{l_2} -1 } + b ( 1 + l_2 \omega^2 ),~ x^{2^{l_1}} + a )
\]
We apply the Buchberger criterion.
If $a \neq 0$, i.e., $x_0 \neq 1$,
then $b \neq 0$ and $I = (x+(\omega^2 + l_2)y, x^{2^{l_1}} + a)$,
so $d_\mm = 2^{l_1}$

If $a=b=0$,
then $I = (y^2 + yx + x^2, yx^{2^{l_2}-1}, x^{2^{l_1}} )$.
The three generators form Gr\"obner basis if $l_2 = l_1$.
Thus, in this case, $d_\mm = 2^{l_1+1}-1$.
If $ l_2 > l_1$, then $d_\mm = 2^{l_1+1}$.

To summarize, except for the special point $(1,1,1)\in \FF^3$ of the affine space,
each point in the algebraic set
\[
 V = \left\{ (x,y,z) \in \FF^3 ~\middle|~
\begin{matrix}
x+y+z+1 = xy+xz+yz+1 = 0 \\  
x^{n_1'}-1 = y^{n_2'}-1 = z^{n_3'}-1 = 0
\end{matrix}
\right\}
\]
contribute $2^{l_1}$ to $d$. The contribution of $(1,1,1)$ is either $2^{l_1 +1}$ or $2^{l_1 +1}-1$.
The latter occurs if and only if $l_1$ and $l_2$,
the two smallest numbers of factors of $2$ in $n_1,n_2,n_3$,
are equal. Let $d_0 = \#V$ be the number of points in $V$.
The desired answer is
\[
 d = 2^{l_1} (d_0 -1) + 
\begin{cases} 
2^{l_1 +1} -1 & \text{if $l_1 = l_2$ } \\
2^{l_1 +1}    & \text{otherwise}
\end{cases}
\]
where $l_1 \le l_2 \le l_3$ are the number of factors of $2$ in $n_i$.

The algebraic set defined by $(x+y+z+1,~xy+xz+yz+1)$ is the union of two isomorphic lines
intersecting only at $x=y=z=1$, one of which is parametrized by $x \in \FF$ as
\[
 (1+x, 1+\omega x, 1+\omega^2 x) \in \FF^3,
\]
and another is parametrized as
\[
 (1+x,1+\omega^2 x, 1+\omega x) \in \FF^3.
\]
where $\omega$ is a primitive third root of unity.
Therefore, the purely geometric number $d_0 = 2 d_1 -1$ can be calculated by
\[
 d_1 = \deg_x \gcd \left( (1+x)^{n_1'}+1, (1+\omega x)^{n_2'}+1, (1+\omega^2 x)^{n_3'}+1 \right) .
\]
Using $(\alpha+\beta)^{2^p} = \alpha^{2^p} + \beta^{2^p}$ and $\omega^2 + \omega + 1 = 0$,
one can easily compute some special cases as summarized in the following corollary.
\hfill $\Diamond$
\end{example}

\begin{cor}
Let $2^k$ be the ground state degeneracy of
the cubic code on the cubic lattice of size $L^3$ with periodic boundary conditions.
($k$ is the number of encoded qubits.)
Then
\begin{align*}
\frac{k+2}{4} 
&= \deg_x \gcd \left( (1+x)^{L}+1,~ (1+\omega x)^{L}+1,~ (1+\omega^2 x)^{L}+1 \right)_{\FF_4} \\
&= \begin{cases}
     1    & \text{if $L = 2^p+1$}, \\
     L    & \text{if $L =2^p$}, \\
     L-2  & \text{if $L = 4^p -1$}, \\
     1    & \text{if $L = 2^{2p+1} -1$}.
   \end{cases}
\label{eq:cubic-code-formula-k}
\end{align*}
where $\omega^2 + \omega + 1 = 0$ and $p \ge 1$ is any integer.
\end{cor}

\begin{example}[Levin-Wen fermion model~\cite{LevinWen2003Fermions}]
The 3-dimensional model is originally defined 
in terms of hermitian bosonic operators $\{ \gamma^{ab} \}_{a,b=1,\ldots,6}$,
squaring to identity if nonzero,
such that $\gamma^{ab}=-\gamma^{ba}$,
$[\gamma^{ab},\gamma^{cd}]=0$ if $a,b,c,d$ are distinct,
and
$\gamma^{ab}\gamma^{bc}=i\gamma^{ac}$ if $a \neq c$.
An irreducible representation is given by Pauli matrices acting on $\mathbb{C}^2 \otimes \mathbb{C}^2$,
and their commuting Hamiltonian fits nicely into our formalism.
The model was proposed to demonstrate that the point-like excitations may actually be fermions.
\begin{align*}
 \sigma_\text{Levin-Wen} &=
\begin{pmatrix}
 1+z & 1+z & x+y \\
 y+y z & x+x z & x+y \\
 y+z & 1+x & 1+x \\
 y+z & z+x z & y+x y
\end{pmatrix} \\
 \epsilon_\text{Levin-Wen} &=
\begin{pmatrix}
 y+z & y+z & y+y z & 1+z \\
 z+x z & 1+x & x+x z & 1+z \\
 y+x y & 1+x & x+y & x+y
\end{pmatrix}
\end{align*}
Here we multiplied the rows of $\epsilon_\text{Levin-Wen}$ by suitable monomials to avoid negative exponents.
One readily verifies that $\ker \epsilon_\text{Levin-Wen} = \im \sigma_\text{Levin-Wen}$.
The model is symmetric under the spatial rotation by $\pi/3$ about $(1,1,1)$ axis.
Indeed, if one changes the variables as $x \mapsto y \mapsto z \mapsto x$ and apply a symplectic transformation
\begin{equation}
 \omega = 
 \begin{pmatrix}
 1 & 0 & 1 & 0 \\
 0 & 1 & 0 & 1 \\
 1 & 0 & 0 & 0 \\
 0 & 1 & 0 & 0
 \end{pmatrix} :
 \begin{cases}
   XI \mapsto YI \\
   IX \mapsto IY \\
   ZI \mapsto XI \\
   IZ \mapsto IX
\end{cases} ,
\label{eq:symplectic-Levin-Wen}
\end{equation}
then $\sigma_\text{Levin-Wen}$ remains the same up to permutations of columns.

The torsion submodule $T$ of $C = \coker \epsilon_\text{Levin-Wen}$,
which describes the point-like charges according to Theorem~\ref{thm:charge-equals-torsion},
is
\begin{equation}
 T = R \cdot \begin{pmatrix} 1+y \\ 1+x \\ 0 \end{pmatrix} .
\label{eq:Levin-Wen-torsion}
\end{equation}
In order to see this, first shift the variables $a = x+1, b=y+1, c=z+1$.
Then, $\epsilon_\text{Levin-Wen}$ becomes
\[
\epsilon_\text{Levin-Wen} =
\begin{pmatrix}
 b+c & b+c & c+b c & c \\
 a+a c & a & c+a c & c \\
 a+a b & a & a+b & a+b
\end{pmatrix}
 =: \phi
\]
We will verify that $N = C / T$ is torsion-free.
A presentation of $N = \coker \phi'$ is obtained by joining the generator of $T$ to the matrix $\phi$.
\[
 \phi' = 
\begin{pmatrix}
 b+c & b+c & c+b c & c & b \\
 a+a c & a & c+a c & c & a \\
 a+a b & a & a+b & a+b & 0
\end{pmatrix}
\]
Column operations of $\phi'$ give
\[
 \phi' \cong
\begin{pmatrix}
 0 & c & b & 0 & 0 \\
 c & 0 & a & 0 & 0 \\
 b & a & 0 & 0 & 0
\end{pmatrix} = 
\begin{pmatrix}
 \partial_2 & 0 & 0
\end{pmatrix}
\]
where $\partial_2$ is from Eq.~\eqref{eq:resolution-F}.
Therefore, $\phi'$ generates the kernel of $\partial_1$,
and by Proposition~\ref{prop:nofractal-torsionless}, $N = \coker \phi'= \coker \partial_2$ is torsion-free.

The torsion submodule $T$ of $C = \coker \phi$ is annihilated by $a$, $b$, or $c$
(See Corollary~\ref{cor:annT-3d-characteristic-dimension-0}):
\[
 a \begin{pmatrix} b \\ a \\ 0 \end{pmatrix} =\phi \begin{pmatrix} 1 \\ 1+a \\ 0 \\ a \end{pmatrix}, \quad
 b \begin{pmatrix} b \\ a \\ 0 \end{pmatrix} =\phi \begin{pmatrix} 1 \\ 1+b \\ 1 \\ 1 \end{pmatrix}, \quad
 c \begin{pmatrix} b \\ a \\ 0 \end{pmatrix} =\phi \begin{pmatrix} 0 \\ 0 \\ 1 \\ 1 \end{pmatrix}.
\]
Therefore, $T$ is isomorphic to $\coker \partial_1 \cong \FF_2$ of Eq.~\eqref{eq:resolution-F}.
The arguments $h_x,h_y,h_z$ of $\phi$ can be thought of as \emph{hopping} operators for the charge.
According to \cite{LevinWen2003Fermions}, one can check that the charge is actually a fermion
from the commutation values among, for example, $h_x,h_y,\bar y h_y$.

Consider a short exact sequence
\[
 0 \to T \to C \to N \to 0 .
\]
The corresponding sequence for 3D toric code splits,
i.e., $C \cong T \oplus N$, while this does not.
It implies that this model is not equivalent to the 3D toric code.

Now we can compute the ground state degeneracy, or $\dim_{\FF_2} K(L)$.
Tensoring the boundary condition
\[
B = R/\bb_L = R/(x^L-1,y^L-1,z^L-1)
\]
to the short exact sequence,
we have a long exact sequence
\[
 \cdots \to
\Tor_1(T,B) \xrightarrow{\delta'} \Tor_1( C, B) \xrightarrow{\delta} \Tor_1(N,B) 
\to T \otimes B \to C \otimes B \to N \otimes B \to 0 .
\]
Hence, $K(L) \cong \Tor_1( C,B)$ has vector space dimension $\dim_{\FF_2} \im \delta + \dim_{\FF_2} \ker \delta$.
Since the sequence is exact, $\dim_{\FF_2} \ker \delta = \dim_{\FF_2} \im \delta'$.
As we have seen in Example~\ref{eg:highD-toric-codes},
\begin{align*}
 \Tor_1(T,B) & \cong \Tor_1(R/\mm,B) \cong (\FF_2)^3 , \quad \text{and} \\
 \Tor_1(N,B) & \cong \Tor_2(R/\mm,B) \cong (\FF_2)^3 .
\end{align*}
It follows that $\dim_{\FF_2} K(L) \le \dim_{\FF_2} \Tor_1(N,B) + \dim_{\FF_2} \Tor_1(T,B) = 6$.

It is routine to verify that $\bb_4 \subseteq I_2(\phi) \subseteq \mm := (x+1,y+1,z+1)$.
Recall the decomposition $K(L) = \bigoplus_\pp K(L)_\pp$ where $\pp$ runs over all maximal ideals of $R/\bb_L$.
Due to Lemma~\ref{lem:associated-maximal-localized-homology}, this decomposition consists of only one summand $K(L)_\mm$.
When $L$ is odd, since $(\bb_L)_\mm = \mm_\mm$, we know $K(L)_\mm = K(1)_\mm$.
Since $\phi \mapsto 0$ under $a=b=c=0$, we see $\dim_{\FF_2} K(1) = 4$.
The logical operators in this case are
\[
 \begin{pmatrix}  0 \\ 0 \\ \widehat z  \\ \widehat z  \end{pmatrix}
\smile
 \widehat x \cdot \widehat{y \bar z}
 \begin{pmatrix}  0 \\ 1 \\ 0 \\ 0 \end{pmatrix}
 \quad ; \quad
 \begin{pmatrix}  \widehat x \\ \widehat x \\ 0 \\ 0 \end{pmatrix}
\smile
 \widehat z \cdot \widehat{x y} 
 \begin{pmatrix}  1 \\ 1 \\ 1 \\ 0 \end{pmatrix}
\]
where $\widehat \mu = \sum_{n=0}^{L-1} \mu^n$ so $\mu \cdot \widehat \mu = \widehat \mu$,
and symplectic pairs are tied.
The left elements are string-like, and the right surface-like.

When $L$ is even, the following are $\FF_2$-independent elements of $K(L)$.
As there are 6 in total, the largest possible number, we conclude that $K(L)$ is 6-dimensional,
i.e., the number of encoded qubits is 3 when linear dimensions are even.
\[
 \begin{pmatrix}  0 \\ 0 \\ \widehat z \\ \widehat z \end{pmatrix}
\smile
 \widehat x ' \widehat y '
 \begin{pmatrix}  1+y \\ x+xy \\ 0 \\ 1+x+y+xy \end{pmatrix}
 ;
 \begin{pmatrix}  \widehat x \\ \widehat x \\ 0 \\ 0 \end{pmatrix}
\smile
 \widehat y ' \widehat z '
 \begin{pmatrix}  1+z \\ 1+z \\ 1+z \\ y+yz \end{pmatrix}
  ;
 \begin{pmatrix}  \widehat y \\ \widehat y \\ \widehat y \\ \widehat y \end{pmatrix}
\smile 
 \widehat z ' \widehat x '
 \begin{pmatrix}  0 \\ 1+x+z+xz \\ 1+x \\ 1+x \end{pmatrix}
\]
where $\widehat \mu ' = \sum_{i=0}^{L/2-1} \mu^{2i}$ so $(1+\mu)\widehat \mu ' = \widehat \mu$.
The pairs are transformed cyclically by $x \mapsto y \mapsto z \mapsto x$
together with the symplectic transformation $\omega$ of Eq.~\eqref{eq:symplectic-Levin-Wen}.
\hfill $\Diamond$
\end{example}

\section{Discussion}

There are many natural questions left unanswered.
Perhaps, it would be the most interesting to answer how much 
the associated ideal $I(\sigma)$
determines about the Hamiltonian.
Note that the very algebraic set defined by the associated ideal
is not invariant under coarse-graining.
For instance, in the characteristic dimension zero case,
the algebraic set can be a several points in the affine space,
but becomes a single point under a suitable coarse-graining.

It is reasonable to conceive that the algebraic set
is mapped by the affine map $(a_i) \mapsto (a_i^n)$
under the coarse-graining by $x_i' = x_i^n$.
This is true if $t=q$, so the $q$-th determinantal ideal of $\epsilon$,
being the initial Fitting ideal,
has the same radical as $\ann \coker \epsilon$.
In fact,
we have implicitly used this idea in the proofs of
Lemma~\ref{lem:coarse-graining-single-polynomial-in-1D},%
~\ref{lem:zero-dim-ideal-over-finite-field},
and Theorem~\ref{thm:fractal-exists-in-3D}.
The case $t > q$ is not explicitly handled here.

Also, it is interesting on its own to prove or disprove
that the elementary symplectic transformations 
generate the whole symplectic transformation group.


\begin{thebibliography}{10}
\expandafter\ifx\csname url\endcsname\relax
  \def\url#1{\texttt{#1}}\fi
\expandafter\ifx\csname urlprefix\endcsname\relax\def\urlprefix{URL }\fi
\providecommand{\bibinfo}[2]{#2}
\providecommand{\eprint}[2][]{\href{http://arxiv.org/abs/#2}{#2}}

\bibitem{Kitaev2003Fault-tolerant}
\bibinfo{author}{Kitaev, A.~Y.}
\newblock \bibinfo{title}{Fault-tolerant quantum computation by anyons}.
\newblock \emph{\bibinfo{journal}{Annals of Physics}}
  \href{http://dx.doi.org/10.1016/S0003-4916(02)00018-0}{\textbf{\bibinfo{volu%
me}{303}}, \bibinfo{pages}{2--30}} (\bibinfo{year}{2003}).
\newblock \eprint{quant-ph/9707021}.

\bibitem{Wen1991SpinLiquid}
\bibinfo{author}{Wen, X.-G.}
\newblock \bibinfo{title}{Mean-field theory of spin-liquid states with finite
  energy gap and topological orders}.
\newblock \emph{\bibinfo{journal}{Phys. Rev. B}}
  \href{http://dx.doi.org/10.1103/PhysRevB.44.2664}{\textbf{\bibinfo{volume}{4%
4}}, \bibinfo{pages}{2664--2672}} (\bibinfo{year}{1991}).

\bibitem{HasanKane2010TIReview}
\bibinfo{author}{Hasan, M.~Z.} \& \bibinfo{author}{Kane, C.~L.}
\newblock \bibinfo{title}{Topological insulators}.
\newblock \emph{\bibinfo{journal}{Rev. Mod. Phys.}}
  \href{http://dx.doi.org/10.1103/RevModPhys.82.3045}{\textbf{\bibinfo{volume}%
{82}}, \bibinfo{pages}{3045}} (\bibinfo{year}{2010}).
\newblock \eprint{1002.3895}.

\bibitem{DennisKitaevLandahlEtAl2002Topological}
\bibinfo{author}{Dennis, E.}, \bibinfo{author}{Kitaev, A.},
  \bibinfo{author}{Landahl, A.} \& \bibinfo{author}{Preskill, J.}
\newblock \bibinfo{title}{Topological quantum memory}.
\newblock \emph{\bibinfo{journal}{J. Math. Phys.}}
  \href{http://dx.doi.org/10.1063/1.1499754}{\textbf{\bibinfo{volume}{43}},
  \bibinfo{pages}{4452--4505}} (\bibinfo{year}{2002}).
\newblock \eprint{quant-ph/0110143}.

\bibitem{AlickiHorodeckiHorodeckiEtAl2010thermal}
\bibinfo{author}{Alicki, R.}, \bibinfo{author}{Horodecki, M.},
  \bibinfo{author}{Horodecki, P.} \& \bibinfo{author}{Horodecki, R.}
\newblock \bibinfo{title}{On thermal stability of topological qubit in
  {K}itaev's 4d model}.
\newblock \emph{\bibinfo{journal}{Open Syst. Inf. Dyn.}}
  \href{http://dx.doi.org/10.1142/S1230161210000023}{\textbf{\bibinfo{volume}{%
17}}, \bibinfo{pages}{1}} (\bibinfo{year}{2010}).
\newblock \eprint{0811.0033}.

\bibitem{ChesiLossBravyiEtAl2010Thermodynamic}
\bibinfo{author}{Chesi, S.}, \bibinfo{author}{Loss, D.},
  \bibinfo{author}{Bravyi, S.} \& \bibinfo{author}{Terhal, B.~M.}
\newblock \bibinfo{title}{Thermodynamic stability criteria for a quantum memory
  based on stabilizer and subsystem codes}.
\newblock \emph{\bibinfo{journal}{New J. Phys.}}
  \href{http://dx.doi.org/10.1088/1367-2630/12/2/025013}{\textbf{\bibinfo{volu%
me}{12}}, \bibinfo{pages}{025013}} (\bibinfo{year}{2010}).
\newblock \eprint{0907.2807}.

\bibitem{Haah2011Local}
\bibinfo{author}{Haah, J.}
\newblock \bibinfo{title}{Local stabilizer codes in three dimensions without
  string logical operators}.
\newblock \emph{\bibinfo{journal}{Phys. Rev. A}}
  \href{http://dx.doi.org/10.1103/PhysRevA.83.042330}{\textbf{\bibinfo{volume}%
{83}}, \bibinfo{pages}{042330}} (\bibinfo{year}{2011}).
\newblock \eprint{1101.1962}.

\bibitem{BravyiHaah2011Energy}
\bibinfo{author}{Bravyi, S.} \& \bibinfo{author}{Haah, J.}
\newblock \bibinfo{title}{On the energy landscape of 3{D} spin {H}amiltonians
  with topological order}.
\newblock \emph{\bibinfo{journal}{Phys. Rev. Lett.}}
  \href{http://dx.doi.org/10.1103/PhysRevLett.107.150504}{\textbf{\bibinfo{vol%
ume}{107}}, \bibinfo{pages}{150504}} (\bibinfo{year}{2011}).
\newblock \eprint{1105.4159}.

\bibitem{BravyiHaah2011Memory}
\bibinfo{author}{Bravyi, S.} \& \bibinfo{author}{Haah, J.}
\newblock \bibinfo{title}{Analytic and numerical demonstration of quantum
  self-correction in the 3{D} cubic code}  (\bibinfo{year}{2011}).
\newblock \eprint{1112.3252}.

\bibitem{CalderbankRainsShorEtAl1997Quantum}
\bibinfo{author}{Calderbank, A.~R.}, \bibinfo{author}{Rains, E.~M.},
  \bibinfo{author}{Shor, P.~W.} \& \bibinfo{author}{Sloane, N. J.~A.}
\newblock \bibinfo{title}{Quantum error correction and orthogonal geometry}.
\newblock \emph{\bibinfo{journal}{Phys.Rev.Lett.}}
  \href{http://dx.doi.org/10.1103/PhysRevLett.78.405}{\textbf{\bibinfo{volume}%
{78}}, \bibinfo{pages}{405--408}} (\bibinfo{year}{1997}).
\newblock \eprint{quant-ph/9605005}.

\bibitem{KitaevShenVyalyi2002CQC}
\bibinfo{author}{Kitaev, A.~Y.}, \bibinfo{author}{Shen, A.~H.} \&
  \bibinfo{author}{Vyalyi, M.~N.}
\newblock \emph{\bibinfo{title}{Classical and Quantum Computation}}
  (\bibinfo{publisher}{American Mathematical Society}, \bibinfo{year}{2002}).

\bibitem{MartinOdlyzkoAndrewWolfram1984}
\bibinfo{author}{Martin, O.}, \bibinfo{author}{Odlyzko, A.~M.} \&
  \bibinfo{author}{Wolfram, S.}
\newblock \bibinfo{title}{Algebraic properties of cellular automata}.
\newblock \emph{\bibinfo{journal}{Communications in Mathematical Physics}}
  \href{http://dx.doi.org/10.1007/BF01223745}{\textbf{\bibinfo{volume}{93}},
  \bibinfo{pages}{219--258}} (\bibinfo{year}{1984}).

\bibitem{Yoshida2011Classification}
\bibinfo{author}{Yoshida, B.}
\newblock \bibinfo{title}{Classification of quantum phases and topology of
  logical operators in an exactly solved model of quantum codes}.
\newblock \emph{\bibinfo{journal}{Annals of Physics}}
  \href{http://dx.doi.org/10.1016/j.aop.2010.10.009}{\textbf{\bibinfo{volume}{%
326}}, \bibinfo{pages}{15--95}} (\bibinfo{year}{2011}).
\newblock \eprint{1007.4601}.

\bibitem{Bombin2011Structure}
\bibinfo{author}{Bombin, H.}
\newblock \bibinfo{title}{Structure of 2{D} topological stabilizer codes}
  (\bibinfo{year}{2011}).
\newblock \eprint{1107.2707}.

\bibitem{BombinDuclosCianciPoulin2012}
\bibinfo{author}{Bombin, H.}, \bibinfo{author}{Duclos-Cianci, G.} \&
  \bibinfo{author}{Poulin, D.}
\newblock \bibinfo{title}{Universal topological phase of two-dimensional
  stabilizer codes}.
\newblock \emph{\bibinfo{journal}{New Journal of Physics}}
  \href{http://dx.doi.org/10.1088/1367-2630/14/7/073048}{\textbf{\bibinfo{volu%
me}{14}}, \bibinfo{pages}{073048}} (\bibinfo{year}{2012}).

\bibitem{Yoshida2011feasibility}
\bibinfo{author}{Yoshida, B.}
\newblock \bibinfo{title}{Feasibility of self-correcting quantum memory and
  thermal stability of topological order}.
\newblock \emph{\bibinfo{journal}{Annals of Physics}}
  \href{http://dx.doi.org/10.1016/j.aop.2011.06.001}{\textbf{\bibinfo{volume}{%
326}}, \bibinfo{pages}{2566--2633}} (\bibinfo{year}{2011}).
\newblock \eprint{1103.1885}.

\bibitem{MacWilliamsSloane1977}
\bibinfo{author}{MacWilliams, F.~J.} \& \bibinfo{author}{Sloane, N.~J.~A.}
\newblock \emph{\bibinfo{title}{The Theory of Error Correcting Codes}}
  (\bibinfo{publisher}{North-Holland, Amsterdam}, \bibinfo{year}{1977}).

\bibitem{GueneriOezbudak2008}
\bibinfo{author}{G\"uneri, C.} \& \bibinfo{author}{\"Ozbudak, F.}
\newblock \bibinfo{title}{Multidimensional cyclic codes and
  {A}rtin–{S}chreier type hypersurfaces over finite fields}.
\newblock \emph{\bibinfo{journal}{Finite Fields and Their Applications}}
  \href{http://dx.doi.org/10.1016/j.ffa.2006.12.003}{\textbf{\bibinfo{volume}{%
14}}, \bibinfo{pages}{44--58}} (\bibinfo{year}{2008}).

\bibitem{Goppa1983}
\bibinfo{author}{Goppa, V.~D.}
\newblock \bibinfo{title}{Algebraico-geometric codes}.
\newblock \emph{\bibinfo{journal}{Mathematics of the USSR-Izvestiya}}
  \href{http://dx.doi.org/10.1070/IM1983v021n01ABEH001641}{\textbf{\bibinfo{vo%
lume}{21}}, \bibinfo{pages}{75}} (\bibinfo{year}{1983}).

\bibitem{CalderbankShor1996Good}
\bibinfo{author}{Calderbank, A.~R.} \& \bibinfo{author}{Shor, P.~W.}
\newblock \bibinfo{title}{Good quantum error-correcting codes exist}.
\newblock \emph{\bibinfo{journal}{Phys. Rev. A}}
  \href{http://dx.doi.org/10.1103/PhysRevA.54.1098}{\textbf{\bibinfo{volume}{5%
4}}, \bibinfo{pages}{1098--1105}} (\bibinfo{year}{1996}).
\newblock \eprint{quant-ph/9512032}.

\bibitem{Steane1996Multiple}
\bibinfo{author}{Steane, A.}
\newblock \bibinfo{title}{Multiple particle interference and quantum error
  correction}.
\newblock \emph{\bibinfo{journal}{Proc. Roy. Soc. Lond. A}}
  \href{http://dx.doi.org/10.1098/rspa.1996.0136}{\textbf{\bibinfo{volume}{452%
}}, \bibinfo{pages}{2551}} (\bibinfo{year}{1996}).
\newblock \eprint{quant-ph/9601029}.

\bibitem{Gottesman1996Saturating}
\bibinfo{author}{Gottesman, D.}
\newblock \bibinfo{title}{A class of quantum error-correcting codes saturating
  the quantum hamming bound}.
\newblock \emph{\bibinfo{journal}{Phys. Rev. A}}
  \href{http://dx.doi.org/10.1103/PhysRevA.54.1862}{\textbf{\bibinfo{volume}{5%
4}}, \bibinfo{pages}{1862}} (\bibinfo{year}{1996}).
\newblock \eprint{quant-ph/9604038}.

\bibitem{Kim2012qupit}
\bibinfo{author}{Kim, I.~H.}
\newblock \bibinfo{title}{3d local qupit quantum code without string logical
  operator}  (\bibinfo{year}{2012}).
\newblock \eprint{1202.0052}.

\bibitem{MichalakisPytel2011stability}
\bibinfo{author}{Michalakis, S.} \& \bibinfo{author}{Pytel, J.}
\newblock \bibinfo{title}{Stability of frustration-free {H}amiltonians}
  (\bibinfo{year}{2011}).
\newblock \eprint{1109.1588}.

\bibitem{BravyiHastingsMichalakis2010stability}
\bibinfo{author}{Bravyi, S.}, \bibinfo{author}{Hastings, M.} \&
  \bibinfo{author}{Michalakis, S.}
\newblock \bibinfo{title}{Topological quantum order: stability under local
  perturbations}.
\newblock \emph{\bibinfo{journal}{J. Math. Phys.}}
  \href{http://dx.doi.org/10.1063/1.3490195}{\textbf{\bibinfo{volume}{51}},
  \bibinfo{pages}{093512}} (\bibinfo{year}{2010}).
\newblock \eprint{1001.0344}.

\bibitem{BravyiHastings2011short}
\bibinfo{author}{Bravyi, S.} \& \bibinfo{author}{Hastings, M.~B.}
\newblock \bibinfo{title}{A short proof of stability of topological order under
  local perturbations}.
\newblock \emph{\bibinfo{journal}{Communications in Mathematical Physics}}
  \href{http://dx.doi.org/10.1007/s00220-011-1346-2}{\textbf{\bibinfo{volume}{%
307}}, \bibinfo{pages}{609--627}} (\bibinfo{year}{2011}).
\newblock \eprint{1001.4363}.

\bibitem{Eisenbud}
\bibinfo{author}{Eisenbud, D.}
\newblock \emph{\bibinfo{title}{Commutative Algebra with a View Toward
  Algebraic Geometry}} (\bibinfo{publisher}{Springer}, \bibinfo{year}{2004}).

\bibitem{PauerUnterkircher1999}
\bibinfo{author}{Pauer, F.} \& \bibinfo{author}{Unterkircher, A.}
\newblock \bibinfo{title}{Gr\"{o}bner bases for ideals in {L}aurent polynomial
  rings and their application to systems of difference equations}.
\newblock \emph{\bibinfo{journal}{Applicable Algebra in Engineering,
  Communication and Computing}}
  \href{http://dx.doi.org/10.1007/s002000050108}{\textbf{\bibinfo{volume}{9}},
  \bibinfo{pages}{271--291}} (\bibinfo{year}{1999}).

\bibitem{BuchsbaumEisenbud1973Exact}
\bibinfo{author}{Buchsbaum, D.~A.} \& \bibinfo{author}{Eisenbud, D.}
\newblock \bibinfo{title}{What makes a complex exact?}
\newblock \emph{\bibinfo{journal}{Journal of Algebra}}
  \href{http://dx.doi.org/10.1016/0021-8693(73)90044-6}{\textbf{\bibinfo{volum%
e}{25}}, \bibinfo{pages}{259--268}} (\bibinfo{year}{1973}).

\bibitem{Northcott}
\bibinfo{author}{Northcott, D.~G.}
\newblock \emph{\bibinfo{title}{Finite Free Resolutions}}
  (\bibinfo{publisher}{Cambridge University Press}, \bibinfo{year}{1976}).

\bibitem{AtiyahMacDonald}
\bibinfo{author}{Atiyah, M.~F.} \& \bibinfo{author}{MacDonald, I.~G.}
\newblock \emph{\bibinfo{title}{Introduction to Commutative Algebra}}
  (\bibinfo{publisher}{Westview}, \bibinfo{year}{1969}).

\bibitem{LangWeil1954}
\bibinfo{author}{Lang, S.} \& \bibinfo{author}{Weil, A.}
\newblock \bibinfo{title}{Number of points of varieties in finite fields}.
\newblock \emph{\bibinfo{journal}{American Journal of Mathematics}}
  \textbf{\bibinfo{volume}{76}}, \bibinfo{pages}{819--827}
  (\bibinfo{year}{1954}).
\newblock \urlprefix\url{http://www.jstor.org/stable/2372655}.

\bibitem{NewmanMoore1999Glassy}
\bibinfo{author}{Newman, M. E.~J.} \& \bibinfo{author}{Moore, C.}
\newblock \bibinfo{title}{Glassy dynamics and aging in an exactly solvable spin
  model}.
\newblock \emph{\bibinfo{journal}{Phys. Rev. E}}
  \href{http://dx.doi.org/10.1103/PhysRevE.60.5068}{\textbf{\bibinfo{volume}{6%
0}}, \bibinfo{pages}{5068--5072}} (\bibinfo{year}{1999}).
\newblock \eprint{cond-mat/9707273}.

\bibitem{BrunsVetter}
\bibinfo{author}{Bruns, W.} \& \bibinfo{author}{Vetter, U.}
\newblock \emph{\bibinfo{title}{Determinantal Rings}}.
\newblock Lecture Notes in Mathematics 1327
  (\bibinfo{publisher}{Springer-Verlag}, \bibinfo{year}{1988}).
\newblock
  \urlprefix\url{http://www.home.uni-osnabrueck.de/wbruns/brunsw/detrings.pdf}.

\bibitem{Lang}
\bibinfo{author}{Lang, S.}
\newblock \emph{\bibinfo{title}{Algebra}} (\bibinfo{publisher}{Springer},
  \bibinfo{year}{2002}), \bibinfo{edition}{revised 3rd} edn.

\bibitem{Wen2003Plaquette}
\bibinfo{author}{Wen, X.-G.}
\newblock \bibinfo{title}{Quantum orders in an exact soluble model}.
\newblock \emph{\bibinfo{journal}{Phys. Rev. Lett.}}
  \href{http://dx.doi.org/10.1103/PhysRevLett.90.016803}{\textbf{\bibinfo{volu%
me}{90}}, \bibinfo{pages}{016803}} (\bibinfo{year}{2003}).
\newblock \eprint{quant-ph/0205004}.

\bibitem{Chamon2005Quantum}
\bibinfo{author}{Chamon, C.}
\newblock \bibinfo{title}{Quantum glassiness}.
\newblock \emph{\bibinfo{journal}{Phys. Rev. Lett.}}
  \href{http://dx.doi.org/10.1103/PhysRevLett.94.040402}{\textbf{\bibinfo{volu%
me}{94}}, \bibinfo{pages}{040402}} (\bibinfo{year}{2005}).
\newblock \eprint{cond-mat/0404182}.

\bibitem{BravyiLeemhuisTerhal2011Topological}
\bibinfo{author}{Bravyi, S.}, \bibinfo{author}{Leemhuis, B.} \&
  \bibinfo{author}{Terhal, B.~M.}
\newblock \bibinfo{title}{Topological order in an exactly solvable 3{D} spin
  model}.
\newblock \emph{\bibinfo{journal}{Annals of Physics}}
  \href{http://dx.doi.org/10.1016/j.aop.2010.11.002}{\textbf{\bibinfo{volume}{%
326}}, \bibinfo{pages}{839--866}} (\bibinfo{year}{2011}).
\newblock \eprint{1006.4871}.

\bibitem{LevinWen2003Fermions}
\bibinfo{author}{Levin, M.} \& \bibinfo{author}{Wen, X.-G.}
\newblock \bibinfo{title}{Fermions, strings, and gauge fields in lattice spin
  models}.
\newblock \emph{\bibinfo{journal}{Phys. Rev. B}}
  \href{http://dx.doi.org/10.1103/PhysRevB.67.245316}{\textbf{\bibinfo{volume}%
{67}}, \bibinfo{pages}{245316}} (\bibinfo{year}{2003}).
\newblock \eprint{cond-mat/0302460}.

\end{thebibliography}
\end{document}